\documentclass[a4paper,USenglish,cleveref,thm-restate]{lipics-v2021}

\title{Games with \texorpdfstring{$\omega$}{ω}-Automatic Preference Relations}

\author{Véronique Bruyère}{Université de Mons (UMONS), Belgium \and \url{https://informatique-umons.be/bruyere-veronique/}}{veronique.bruyere@umons.ac.be}{https://orcid.org/0000-0002-9680-9140}{}

\author{Christophe Grandmont}{Université de Mons (UMONS), Belgium \and Université libre de Bruxelles (ULB), Belgium \and \url{https://chrisgdt.github.io/}}{christophe.grandmont@umons.ac.be}{https://orcid.org/0009-0009-4573-0123}{}

\author{Jean-François Raskin}{Université libre de Bruxelles (ULB), Belgium \and \url{https://verif.ulb.ac.be/jfr/}}{jean-francois.raskin@ulb.be}{https://orcid.org/0000-0002-3673-1097}{Supported by Fondation ULB (\url{https://www.fondationulb.be/en/}) and the Thelam Fondation.}

\authorrunning{V.\ Bruyère, C.\ Grandmont, and J.-F.\ Raskin}



\nolinenumbers
\pdfoutput=1
\hideLIPIcs

\funding{This work has been supported by the Fonds de la Recherche Scientifique – FNRS under Grant n° T.0023.22 (PDR Rational).}

\Copyright{Véronique Bruyère, Christophe Grandmont, and Jean-François Raskin}

\usepackage{multicol}
\usepackage{tikz}
\usepackage{graphics}
\usepackage{todonotes}

\newtheorem*{problems*}{Problems}

\mathchardef\mhyphen="2D

\newcommand{\leqRelation}[1][]{\ensuremath{\precsim_{#1}}}
\newcommand{\geqRelation}[1][]{\ensuremath{\succsim_{#1}}}
\newcommand{\equivRelation}[1][]{\ensuremath{\sim_{#1}}}
\newcommand{\leqRelationStrict}[1][]{\ensuremath{\prec_{#1}}}
\newcommand{\geqRelationStrict}[1][]{\ensuremath{\succ_{#1}}}

\newcommand{\equivClass}[2]{\ensuremath{[#1]_{#2}}}

\newcommand{\lattice}[1][]{\ensuremath{{\Lambda}_{#1}}}

\newcommand{\prover}{\ensuremath{\mathbb{P}}}
\newcommand{\proverone}{\ensuremath{\mathbb{P}_{1}}}
\newcommand{\provertwo}{\ensuremath{\mathbb{P}_{2}}}
\newcommand{\proveroneTight}{\proverone{}\!}
\newcommand{\provertwoTight}{\provertwo{}\!}
\newcommand{\challenger}{\ensuremath{\mathbb{C}}}

\newcommand{\pcp}{\proverone{}\challenger{}\provertwo{}}
\newcommand{\pc}{\prover{}\challenger{}}

\newcommand{\projOne}[1]{\mathsf{proj}_{V,1}(#1)}
\newcommand{\projTwo}[1]{\mathsf{proj}_{V,2}(#1)}
\newcommand{\projE}[1]{\mathsf{proj}_{E}(#1)}
\newcommand{\projDev}[1]{\mathsf{dev}(#1)}

\newcommand{\R}{\mathbb{R}}
\newcommand{\N}{\mathbb{N}}

\newcommand{\ssetminus}{\! \setminus \!}
\def\rest#1#2{#1_{\restriction#2}}

\newcommand{\aut}[1]{\ensuremath{\mathcal{#1}}}

\newcommand{\lang}[1]{\mathcal{L}(#1)}

\newcommand{\NBW}{\text{NBA}}
\newcommand{\DBW}{\text{DBA}}
\newcommand{\NBWs}{\text{NBAs}}
\newcommand{\DBWs}{\text{DBAs}}

\newcommand{\DRW}{\text{DRA}}

\newcommand{\DPW}{\text{DPA}}
\newcommand{\NPW}{\text{NPA}}
\newcommand{\DPWs}{\text{DPAs}}
\newcommand{\NPWs}{\text{NPAs}}

\newcommand{\arena}{A}

\newcommand{\game}{\mathcal{G}}

\newcommand{\Players}{\mathcal{P}}

\newcommand{\Plays}{\mathsf{Plays}}

\newcommand{\infOcc}[1]{\mathsf{Inf}({#1})}

\newcommand{\last}[1]{\mathsf{last}({#1})}
\newcommand{\first}[1]{\mathsf{first}({#1})}

\newcommand{\strategyfor}[1]{\sigma_{#1}}

\newcommand{\machine}[1]{\mathcal{M}_{#1}}

\newcommand{\outcome}[1]{\langle #1 \rangle}
\newcommand{\outcomefrom}[2]{\outcome{#1}_{#2}}

\newcommand{\bigO}[1]{\mathcal{O}(#1)}

\newcommand{\nl}{$\mathsf{NL}$}
\newcommand{\nlHard}{\nl{}-hard}
\newcommand{\nlComplete}{\nl{}-complete}

\newcommand{\np}{$\mathsf{NP}$}

\newcommand{\conp}{$\mathsf{coNP}$}

\newcommand{\pspace}{$\mathsf{PSPACE}$}
\newcommand{\pspaceHard}{\pspace{}-hard}
\newcommand{\pspaceComplete}{\pspace{}-complete}

\newcommand{\exptime}{$\mathsf{EXPTIME}$}

\usetikzlibrary{
  automata,
  arrows,
  positioning,
  shapes.geometric,
  calc
}
\tikzset{
edge with arrows/.style = {
    ->,
    >=stealth,
    shorten >=1pt,
},
directed/.style = {
    edge with arrows,
    node distance=2.3cm,
    on grid,
    semithick,
    double distance=1.5pt,
},
automaton/.style = {
    directed,
    auto,
    initial text={},
    pin distance = 1ex,
    every pin edge/.style = {
        draw=none
    },
    every state/.style={
      minimum size=1.5mm,
      inner sep=1pt,
    }
},
system/.style = {
    state,
    circle,
    minimum size=0mm,
    inner sep=5pt,
},
system2/.style={
    state,
    ellipse,
    minimum size=0mm,
    inner sep=2pt,
},
environment/.style = {
    state,
    rectangle,
    minimum size=0mm,
    inner sep=8pt,
},
environment2/.style = {
    state,
    diamond,
    minimum size=0mm,
    inner sep=4pt,
},
}

\begin{document}

\maketitle

\begin{abstract}
This paper investigates Nash equilibria (NEs) in multi-player turn-based games on graphs, where player preferences are modeled as $\omega$-automatic relations via deterministic parity automata. Unlike much of the existing literature, which focuses on specific reward functions, our results apply to any preference relation definable by an $\omega$-automatic relation. We analyze the computational complexity of determining the existence of an NE (possibly under some constraints), verifying whether a given strategy profile forms an NE, and checking whether a specific outcome can be realized by an NE. When a (constrained) NE exists, we show that there always exists one with finite-memory strategies. Finally, we explore fundamental properties of $\omega$-automatic relations and their implications for the existence of equilibria.
\end{abstract}

\section{Introduction}
\label{section:intro}

\emph{Non-zero-sum games on graphs} provide a powerful framework for analyzing rational behavior in multi-agent systems, see, e.g.,~\cite{BrenguierCHPRRS16,BriceRB23,Bruyere21,FleschKMSSV10,Gradel-Ummels-08,GutierrezNPW23,KupfermanPV16}. By modeling settings where agents have individual objectives, this approach captures the complexity of real-world scenarios where the interests of agents (modeled by players) are neither fully aligned nor entirely antagonistic. It enables the study of solution concepts such as \emph{Nash} and \emph{subgame-perfect equilibria}~\cite{Nash50,osbornebook}, offering insight into strategic decision making. This, in turn, can aid in designing systems that anticipate and respond to rational behaviors, enriching reactive synthesis methodologies.

In this context, specifying player \emph{objectives}~\cite{lncs2500} is central to reasoning about strategies and equilibria. In qualitative games, objectives determine whether an execution (an infinite path in the graph) is winning or losing for a given player. In quantitative games, executions are instead assigned numerical values, allowing players to compare and rank them based on accumulated rewards, with higher values being preferable. From this perspective, the qualitative setting can be viewed as a special case where the values are Boolean, typically captured by parity acceptance conditions, which encompass all $\omega$-regular objectives. In the quantitative setting, a variety of reward functions have been explored, including total sum, limsup and liminf, discounted-sum, and hybrid models such as cost-optimal reachability. For each of these functions, dedicated techniques have been developed to design algorithms that analyze optimal strategies and, more broadly, equilibria.

However, these solutions are often tightly coupled to the specific reward function used, which limits their generality. When a new reward function or combination thereof is introduced, significant technical effort is required, as existing techniques rarely transfer across different reward models. This lack of general results -- where solutions remain specialized to the underlying evaluation model, preventing knowledge transfer between different classes of objectives -- has been noted in related contexts such as quantitative verification (see, e.g.,~\cite{Bansal-Comparator-Automata-2022}).

To address this, we propose a general approach based on \emph{automata-based} preference relations to compare infinite paths in the graph. This framework provides a structured and unified method for reasoning about strategies and equilibria across various reward models. A similar use of automata-based preference relations has been explored in~\cite{Bansal-Comparator-Automata-2022,Berwanger-Doyen-2023,Bozzelli-Maubert-Pinchinat-2015}, and here we demonstrate how this idea can be adapted to fit the non-zero-sum game setting.

\subparagraph*{Contributions.}
Our contributions center on using \emph{$\omega$-automatic relations} on infinite words~\cite{Thomas-1990-Buchi}, as introduced in~\cite{Frougny-Sakarovitch-1993}, to define a general framework for preference relations over paths in game graphs, thereby establishing a generic method to compare executions for players in non-zero-sum games. These relations are specified by deterministic parity automata that read pairs $(x,y)$ of words synchronously and accept them whenever $y$ is preferred to $x$.

Our main contributions focus on the computational complexity of \emph{four key problems} related to NEs in non-zero-sum games~\cite{Nash50} with $\omega$-automatic preference relations. First, we study the problem of \emph{verifying} whether a given strategy profile, specified by Mealy machines, one per strategy, constitutes an NE in the given game. We prove that this problem is \pspaceComplete{} (\cref{theorem:nash-checking-pspace}). Second, we examine whether a lasso-shaped path (i.e., a regular path) is the \emph{outcome of an NE}, showing that this problem is in \np{} $\cap$ \conp{} and $\mathsf{Parity}$-hard (\cref{theorem:OutcomeCheck}). Third, we establish the existence of games without any NE, motivating the fundamental problem of determining whether a given game \emph{admits at least one NE}. This problem turns out to be particularly challenging, and we reduce it to a three-player zero-sum game with imperfect information. We provide an algorithm for solving this problem with exponential complexity in the size of the graph, the parity automata defining the preference relation, and the number of their priorities, and doubly exponential complexity in the number of players. However, since the number of players is a natural parameter that tends to be small in practical scenarios,\footnote{In robotic systems or in security protocols, the number of agents is usually limited to a few. For example, in a security protocol, the players are Alice and Bob who exchange messages, the trusted third party, and a fourth player for the network (see, e.g.,~\cite{ChatterjeeRaman12,GutierrezNPW20})} we refine this result by proving that for a fixed number of players, the problem lies in \exptime{} and is \pspaceHard{} (\cref{theorem:NEexistenceGeneral}). In addition, our approach has the advantage of being modular and therefore easily adapts to question the existence of a \emph{constrained NE}.
When we attach one constraint to each player given as a lasso-shaped path and ask for an NE whose outcome is preferred to any of those constraints, the adapted algorithm keeps the same complexity except that it becomes doubly exponential in the number of priorities of the parity conditions. Yet the number of priorities is often small\footnote{important classes of objectives such as Büchi, co-Büchi, reachability, and safety require at most three priorities} and when we fix it and the number of players, the algorithm remains in \exptime{} and \pspaceHard{} (\cref{theorem:constrainedNEexistenceGeneral}). Note that our approach allows to show that when there exists an (constrained) NE, there exists one composed of finite-memory strategies.

Additionally, we analyze the algorithmic complexity of verifying whether an $\omega$-automatic relation satisfies the axioms of a \emph{strict partial order} (irreflexivity and transitivity) or of a \emph{preorder} (reflexivity and transitivity) which are two classical requirements for a relation to model preferences. We show that these problems are \nlComplete{} (\cref{prop:PropertiesRelations}).
Finally, we show that when the $\omega$-automatic preference relations are all \emph{$\omega$-recognizable} (a strict subclass of $\omega$-automatic relations where the two input words can be processed independently) and preorders, the existence of at least one NE is always guaranteed (\cref{theorem:existenceNE}).

\subparagraph*{Related work.}

A well-established \emph{hierarchy of rational relations} holds for both finite and infinite words~\cite{Carton2006,BookSakarovitch}. The $\omega$-automatic relations -- also called synchronized $\omega$-rational relations -- were first studied in~\cite{Frougny-Sakarovitch-1993}. Some decision problems about $\omega$-automatic and $\omega$-recognizable relations were solved in~\cite{rational-relations-automatic-loding} and improved in~\cite{lics23-rational-relations}. The study of \emph{automatic structures} has also led to results involving rational relations, notably within first-order logic (see, e.g., \cite{Automatic-Structures-Gradel2000,Automatic-Structures-Gradel2020,hodgson1983decidabilite,Automatic-Structures-handbook-Rubin21}).

The problems we study in this paper were widely investigated in the literature for \emph{specific} reward functions, including functions that mix different objectives, see, e.g.,~\cite{Bouyer-Brenguier-Markey-2010,Gutierrez-MPRSW21,Ummels08,Ummels-Wojtczak-2011}. There are also works that study these problems across \emph{large} classes of reward functions rather than individual ones, or that consider \emph{general} notions of preference relations.
For instance, in~\cite{BrihayePS13}, the authors prove the existence of finite-memory NEs for all cost-semi-linear reward functions.
In~\cite{PatriciaBouyerBMU15}, a complete methodology is developed to solve the (constrained) NE existence problem, thanks to the concept of suspect game, encompassing all reward functions definable by a class of monotone circuits over the set of states
that appear (finitely or infinitely often) along paths in a game graph. The preference relations studied in~\cite{PatriciaBouyerBMU15} are all $\omega$-automatic. In~\cite{Feinstein-Kupferman-Shenwald2025}, the authors study NEs for games with a reward function that, given a finite set $X$ of objectives of the same type, associates an integer with each subset of satisfied objectives of $X$. Again, if the objectives of $X$ are $\omega$-regular, the reward functions of~\cite{Feinstein-Kupferman-Shenwald2025} lead to $\omega$-automatic preference relations. The existence of NEs is guaranteed within a broad setting, both in~\cite{Gradel-Ummels-08} and~\cite{LeRoux-Pauly-Equilibria}, without relying on an automata-based approach, however with no complexity result about the constrained NE existence problem. In case of games with $\omega$-recognizable preference relations, our proof that NEs always exist relies on the technique developed in~\cite{Gradel-Ummels-08}.

The results we obtain with games with $\omega$-recognizable preference relations cover a large part of the games studied classically. In addition, our setting allows any combinations of objectives as soon as they are expressible by automata. However, it does not cover arbitrary games with mean-payoff or energy objectives. Indeed, in the first case, it is proved in~\cite{Bansal-Comparator-Automata-2022} that the related preference relation is not $\omega$-automatic in general; and in the second case, the constrained NE existence problem is undecidable~\cite{BriceBR23}. Note that the general concepts of $\omega$-automatic and $\omega$-recognizable relations have also been used to study imperfect information in games in~\cite{Berwanger-Doyen-2023,Bozzelli-Maubert-Pinchinat-2015} and formal verification of quantitative systems in~\cite{Bansal-Comparator-Automata-2022}.

\section{Preliminaries}
\label{section:preliminaries}

In this section, we introduce the useful definitions of games with $\omega$-automatic preference relations and give several illustrative examples.

\subparagraph*{Automatic Relations.}

Let $\Sigma$ be a fixed finite alphabet. We consider \emph{binary relations} $\mathord{R} \subseteq \Sigma^\omega \times \Sigma^\omega$ on infinite words over $\Sigma$. The relation $R$ is \emph{$\omega$-automatic} if it is accepted by a deterministic finite parity automaton over the alphabet $\Sigma \times \Sigma$, that is, $R$ is an $\omega$-regular language over $\Sigma \times \Sigma$. The automaton reads pairs of letters by advancing synchronously on the two words. This behavior is illustrated in \cref{fig:dpw-for-classical-objectivers} below. A relation $R$ is \emph{$\omega$-recognizable} if it is equal to $\cup_{i=1}^{\ell} X_i \times Y_i$ where $X_i, Y_i \subseteq \Sigma^\omega$ are $\omega$-regular languages over $\Sigma$~\cite{rational-relations-automatic-loding}. Any $\omega$-recognizable relation is $\omega$-automatic~\cite{BookSakarovitch}.

We suppose that the reader is familiar with the usual notion of deterministic parity automaton (\DPW{}) used to accept $\omega$-automatic relations~\cite{rational-relations-automatic-loding}. A run is accepting if the maximum priority seen infinitely often is even. In this paper, we also use other classical notions of automata: deterministic B\"uchi automata (\DBW{}) and Rabin automata. See, e.g.,~\cite{lncs2500} for general definitions, or~\cite{principles-of-model-checking,handbook-of-model-checking-orna-kupferman} for deeper details. We also need the concept of \emph{generalized} parity automaton which is an automaton with a positive\footnote{The negation is not allowed in the Boolean combination.} Boolean combination of parity conditions. Missing transitions lead to a rejecting sink, omitted in the figures; before complementing an automaton, we complete it with this sink. Given an automaton $\aut{A}$, its \emph{size} $|\aut{A}|$ is its number of states.

\subparagraph*{Games with Preference Relations.}
An \emph{arena} is a tuple $\arena = (V,E,\Players,(V_i)_{i\in\Players})$ where $V$ is a finite set of \emph{vertices}, $E \subseteq V \times V$ is a set of \emph{edges}, $\Players$ is a finite set of \emph{players}, and $(V_i)_{i\in\Players}$ is a partition of $V$, where $V_i$ is the set of vertices \emph{owned} by player~$i$. We assume, w.l.o.g., that each $v \in V$ has at least one \emph{successor}, i.e., there exists $v' \in V$ such that $(v,v') \in E$. We define a \emph{play} $\pi\in V^\omega$ (resp.\ a \emph{history} $h\in V^*$) as an infinite (resp.\ finite) sequence of vertices $\pi_0\pi_1\dots$ such that $(\pi_k,\pi_{k+1})\in E$ for any two consecutive vertices $\pi_k, \pi_{k+1}$. The set of all plays of an arena $\arena{}$ is denoted $\Plays_{\arena}\subseteq V^\omega$, and we write $\Plays$ when the context is clear.
The \emph{length} $|h|$ of a history $h$ is the number of its vertices. The empty history is denoted $\varepsilon$.

A \emph{game} $\game = (\arena, (R_i)_{i\in\Players})$ is an arena equipped with $\omega$-automatic relations $R_i$ over the alphabet $V$, one for each player~$i$, called his \emph{preference relation}. For any two plays $\pi, \pi'$, player~$i$ \emph{prefers} $\pi'$ to $\pi$ if $(\pi,\pi') \in R_i$. In the sequel, we write $\leqRelationStrict[i]$ instead of $R_i$ and for all $x,y \in V^\omega$, $x \leqRelationStrict[i] y$, or $y \geqRelationStrict[i] x$, instead of $(x,y) \in \mathord{R_i}$. We also say that $x$ is \emph{maximal} (resp.\ \emph{minimal}) for $\leqRelationStrict[i]$ if for all $y \in V^\omega$, we have $x \not\leqRelationStrict[i] y$ (resp.\ $y \not\leqRelationStrict[i] x$). Below we give various examples of games whose preference relations are all strict partial orders. At this stage, $\leqRelationStrict[i]$ is just an $\omega$-automatic relation without any additional hypotheses. Such hypotheses will be discussed in \Cref{section:properties-relations}.

Given a play $\pi$ and an index $k$, we write $\pi_{\geq k}$ the suffix $\pi_{k} \pi_{k+1} \dots$ of $\pi$. We denote the first vertex of $\pi$ by $\first{\pi}$. These notations are naturally adapted to histories. We also write $\last{h}$ for the last vertex of a history $h \neq \varepsilon$. We can \emph{concatenate} two non-empty histories $h_1$ and $h_2$ into a single one, denoted $h_1\cdot h_2$ or $h_1h_2$ if $(\last{h_1},\first{h_2})\in E$. When a history can be concatenated to itself, we call it \emph{cycle}. A play $\pi= \mu\nu\nu \dots = \mu(\nu)^\omega$, where $\mu\nu$ is a history and $\nu$ a cycle, is called a \emph{lasso}. The \emph{length} of $\pi$ is then the length of $\mu\nu$, denoted $|\pi|$.\footnote{To have a well-defined length for a lasso $\pi$, we assume that $\pi = \mu(\nu)^\omega$ with $\mu\nu$ of minimal length.}

Let $\arena$ be an arena. A \emph{strategy} $\strategyfor{i}:V^*V_i\rightarrow V$ for player~$i$ maps any history $h \in V^*V_i$ to a successor $v$ of $\last{h}$, which is the next vertex that player~$i$ chooses to move after reaching the last vertex in $h$.
A play $\pi = \pi_0\pi_1 \dots$ is \emph{consistent} with $\strategyfor{i}$ if $\pi_{k+1} = \strategyfor{i}(\pi_0 \dots \pi_k)$ for all $k \in \N$ such that $\pi_k \in V_i$.
Consistency is naturally extended to histories.
A tuple of strategies $\sigma = (\strategyfor{i})_{i\in\Players}$ with $\sigma_i$ a strategy for player~$i$, is called a \emph{strategy profile}. The play $\pi$ starting from an initial vertex $v_0$ and consistent with each $\strategyfor{i}$ is denoted by $\outcomefrom{\sigma}{v_0}$ and called \emph{outcome}.

A strategy $\sigma_i$ for player~$i$ is \emph{finite-memory}~\cite{lncs2500} if it can be encoded by a \emph{Mealy machine} $\machine{} = (M, m_0, \alpha_U, \alpha_N)$ where $M$ is the finite set of memory states, $m_0 \in M$ is the initial memory state, $\alpha_U : M \times V \rightarrow M$ is the update function, and $\alpha_N : M \times V_i \rightarrow V$ is the next-move function. Such a machine defines the strategy $\sigma_i$ such that $\sigma_i(hv) = \alpha_N(\widehat{\alpha}_U(m_0,h),v)$ for all histories $hv \in V^*V_i$, where $\widehat{\alpha}_U$ extends $\alpha_U$ to histories as expected. A strategy $\sigma_i$ is \emph{memoryless} if it is encoded by a Mealy machine with only one state.

We suppose that the reader is familiar with the concepts of two-player \emph{zero-sum} games with $\omega$-regular objectives and of winning strategy~\cite{Games-on-Graphs,lncs2500}.

\subparagraph*{Generality of the $\omega$-Automatic Preference Framework.}
Let us show that the above notion of game $\game = (\arena, (\leqRelationStrict[i])_{i \in \Players})$ encompasses many cases of classic games and more. We begin with games where each player~$i$ has an $\omega$-regular objective $\Omega_i \subseteq V^\omega$, such as a reachability or a B\"uchi objective~\cite{Games-on-Graphs,lncs2500}. In this case, the preference relation $\mathord{\leqRelationStrict[i]} \subseteq V^\omega \times V^\omega$ is defined by $x \leqRelationStrict[i] y$ if and only if $\Omega_i(x) < \Omega_i(y)$, where $\Omega_i$ is seen as a function $\Omega_i : V^\omega \rightarrow \{0,1\}$. As $\Omega_i$ is $\omega$-regular, it follows that $\leqRelationStrict[i]$ is $\omega$-automatic. For instance, given a target set $T \subseteq V$, the first \DPW{} of \cref{fig:dpw-for-classical-objectivers} accepts $\leqRelationStrict[i]$ when $\Omega_i$ is a \emph{reachability} objective $\{x = x_0x_1 \ldots \in V^\omega \mid \exists k, x_k \in T\}$; the second \DPW{} accepts $\leqRelationStrict[i]$ when $\Omega_i$ is a \emph{B\"uchi} objective $\{x \in V^\omega \mid \infOcc{x} \cap T \neq \ \varnothing \}$, where $\infOcc{x}$ is the set of vertices seen infinitely many times in $x$.

More general preference relations can be defined from several $\omega$-regular objectives $(\Omega_i^j)_{1\leq j \leq n}$ for player~$i$. With each $x \in V^\omega$ is associated the payoff vector $\bar \Omega_i(x) = (\Omega_i^1(x), \ldots, \Omega_i^n(x) ) \in \{0,1\}^n$. Given a strict partial order $<$
on these payoff vectors, we define a preference relation $\leqRelationStrict[i]$ such that $x \leqRelationStrict[i] y$ if and only if $\bar\Omega_i(x) < \bar\Omega_i(y)$~\cite{PatriciaBouyerBMU15}. There exist several strict partial orders on the payoff vectors, like, for example, the lexicographic order, or the counting order, i.e., $\bar\Omega_i(x) < \bar\Omega_i(y)$ if and only if $|\{j \mid \Omega_i^j(x) = 1\}| < |\{j \mid \Omega_i^j(y) = 1\}|$. One can check that all preference relations studied in~\cite{PatriciaBouyerBMU15} are $\omega$-automatic.

Let us move on to classical quantitative objectives, such as quantitative reachability, limsup or discounted-sum objectives~\cite{Games-on-Graphs,lncs2500}. In this case, an objective for player~$i$ is now a function $\Omega_i: V^\omega \rightarrow \R \cup \{\pm\infty\}$.\footnote{It can also be a function $\Omega: E^\omega \rightarrow \R \cup \{\pm\infty\}$.} We then define a preference relation $\leqRelationStrict[i]$ such that $x \leqRelationStrict[i] y$ if and only if $\Omega_i(x) < \Omega_i(y)$. Bansal et al.\ showed in~\cite{Bansal-Comparator-Automata-2022} that such a relation is $\omega$-automatic for a limsup objective and for a discounted-sum objective with an integer discount factor greater than $1$. They also proved that $\leqRelationStrict[i]$ is not $\omega$-automatic for the class of mean-payoff objectives. The first \DPW{} of \cref{fig:dpw-for-classical-objectivers} where the label on the loop on the vertex with priority $0$ is replaced by $(*,*)$, accepts a preference relation $\leqRelationStrict[i]$ defined from a \emph{unit-cost reachability} objective, called \emph{min-cost-reachability} here, as follows: $x \leqRelationStrict[i] y$ if and only if there exists $\ell$ such that $y_\ell \in T$ and, for all $k$, $x_k \in T \Rightarrow \exists \ell < k, ~y_\ell \in T$ (player~$i$ prefers plays with fewer steps to reach the target set $T$). The variant where player~$i$ prefers to maximize the number of steps to reach $T$,\footnote{as each step corresponds to a reward.} called \emph{max-reward-reachability}, is accepted by the third \DPW{} in \cref{fig:dpw-for-classical-objectivers}.

Hence, there are many ways to envision $\omega$-automatic relations.
Note that in our framework, the preference relations $\leqRelationStrict[i]$ of a game $\game$ can vary from one player to another, where each relation $\leqRelationStrict[i]$ can be defined from a combination of several objectives (see \cref{ex:example} below).

\begin{figure}
    \begin{subfigure}[t]{0.27\textwidth}
        \centering
        \begin{tikzpicture}[automaton,scale=.7,every node/.style={scale=.7},node distance=1.9]
            \node[initial,system] (q0) {$1$};
            \node[system] (q1) [right=of q0] {$0$};

            \path (q0) edge            node {$\neg T, T$} (q1)
                       edge[loop above]node {$\neg T, \neg T$} (q0)
                  (q1) edge[loop above]node {$\neg T, *$} (q1);
        \end{tikzpicture}
    \end{subfigure}\hfill
    \begin{subfigure}[t]{0.33\textwidth}
        \centering
        \begin{tikzpicture}[automaton,scale=.7,every node/.style={scale=.7},node distance=3]
            \node[initial,system] (q0) {$1$};
            \node[system] (q2) [right=of q0] {$3$};
            \node[system] (q1) at ($(q0)!0.5!(q2) - (0,1.9)$) {$2$};

            \path (q0) edge                node[above] {$T, *$} (q2)
                       edge[bend right=15] node[sloped,below] {$\neg T, T$} (q1)
                       edge[loop above]    node {$\neg T, \neg T$} (q0)
                  (q1) edge[bend right=15] node[sloped] {$\neg T, \neg T$} (q0)
                       edge[bend right=15] node[sloped,below] {$T, *$} (q2)
                       edge[loop right]     node {$\neg T, T$} (q1)
                  (q2) edge[bend right=25] node[above] {$\neg T, \neg T$} (q0)
                       edge[bend right=15] node[sloped] {$\neg T, T$} (q1)
                       edge[loop above]    node {$T, *$} (q2);
        \end{tikzpicture}
    \end{subfigure}\hfill
    \begin{subfigure}[t]{0.35\textwidth}
        \centering
            \begin{tikzpicture}[automaton,scale=.7,every node/.style={scale=.7},node distance=1.6]
              \node[initial,system] (q0) {$1$};
              \node[system] (q1) [below=of q0] {$1$};
              \node[system] (q2) [right=of q1] {$0$};
              \node[system] (q3) [right=of q0] {$0$};

              \path (q0) edge            node {$T, \neg T$} (q1)
                         edge            node {$\neg T, T$} (q3)
                         edge[loop above]node {$\neg T, \neg T$} (q0)
                    (q1) edge[loop left] node {$*, \neg T$} (q1)
                         edge            node {$*, T$} (q2)
                    (q2) edge[loop above]node {$*, *$} (q2)
                    (q3) edge[loop above]node {$\neg T, *$} (q3);
            \end{tikzpicture}
    \end{subfigure}
    \caption{\DPWs{} accepting preference relations, corresponding respectively to reachability, B\"uchi, and max-reward-reachability objectives. The priorities are indicated inside each state, and an edge label $T$, $\neg T$, or $*$ means that there is an edge for each label $v \in T$, $v \in V \setminus T$, and $v \in V$, respectively.}
    \label{fig:dpw-for-classical-objectivers}
\end{figure}

\section{Decision Problems about Nash Equilibria}
\label{section:nash-decisionProb}

In this section, we state the decision problems studied in this paper and we provide our main results regarding their complexity classes.

\subparagraph*{Studied Problems.}
A \emph{Nash Equilibrium} (NE) from an initial vertex $v_0$ is a strategy profile $(\sigma_i)_{i \in \Players}$ such that for all players~$i$ and all strategies $\tau_i$ of player~$i$, we have $\outcomefrom{\sigma}{v_0} \not\leqRelationStrict[i] \outcomefrom{\tau_i,\sigma_{-i}}{v_0}$, where $\sigma_{-i}$ denotes $(\sigma_j)_{j \in \Players \setminus \{i\}}$. So, NEs are strategy profiles where no single player has an incentive to unilaterally deviate from his strategy. When there exists a strategy $\tau_i$ such that $\outcomefrom{\sigma}{v_0} \leqRelationStrict[i] \outcomefrom{\tau_i,\sigma_{-i}}{v_0}$, we say that $\tau_i$ (or, by notation abuse, $\outcomefrom{\tau_i,\sigma_{-i}}{v_0}$) is a \emph{profitable deviation} for player~$i$. The set of players $\Players \ssetminus \{i\}$ is called \emph{coalition}~$-i$, sometimes seen as one player opposed to player~$i$.

\begin{figure}
    \centering
    \begin{tikzpicture}[automaton,every node/.style={scale=.7}]
        \node[system] (v0) {$v_0$};
        \node[environment] (v3) [right of=v0] {$v_3$};
        \node[system] (v1) [above=2.7em of v0] {$v_1$};
        \node[system] (v2) [left of=v0] {$v_2$};

        \path (v0) edge[bend right=20] (v3)
                   edge[bend right=20] (v2)
              (v3) edge[bend right=20] (v0)
                   edge (v1)
              (v1) edge (v0)
              (v2) edge[bend right=20] (v0);
    \end{tikzpicture}
    \caption{An arena with round (resp.\ square) vertices owned by player~$1$ (resp.\ player~$2$).}
    \label{fig:example-buchi-reach-game}
\end{figure}

\begin{example}\label{ex:example}
    Let us illustrate the NE definition with two examples. We consider the two-player arena depicted in \cref{fig:example-buchi-reach-game} such that player~$2$ owns only $v_3$ and player~$1$ owns all other vertices. The preference relation $\leqRelationStrict[1]$ for player~$1$ is defined from a min-cost-reachability objective with a target set $T_1 = \{v_1\}$. The preference relation $\leqRelationStrict[2]$ for player~$2$ is defined from a B\"uchi objective with a target set $T_2 = \{v_2\}$.
    Let us consider the strategy profile $\sigma = (\sigma_1,\sigma_2)$ defined by two memoryless strategies such that $\sigma_1(v_0) = v_3$ and $\sigma_2(v_3) = v_0$.\footnote{As $v_1$ and $v_2$ have only one successor, the strategy is trivially defined for those vertices.} It is an NE from the initial vertex $v_0$ with outcome $\outcomefrom{\sigma}{v_0} = (v_0v_3)^\omega$. Player~$1$ has no profitable deviation if player~$2$ sticks to his strategy $\sigma_2$: it is not possible to visit vertex $v_1$. Player~$2$ has also no profitable deviation. There exists another NE $\sigma' = (\sigma'_1,\sigma'_2)$ from $v_0$ such that
    \begin{itemize}
        \item $\sigma_1'(hv_0) = v_2$ if the history $h$ visits $v_1$, and to $\sigma_1'(hv_0) = v_3$ otherwise,
        \item $\sigma'_2$ is the memoryless strategy such that $\sigma_2'(v_3) = v_1$.
    \end{itemize}
    In that case, the NE outcome is $\outcomefrom{\sigma'}{v_0} = v_0v_3v_1(v_0v_2)^\omega$. Note that both players prefer the second NE as $\outcomefrom{\sigma}{v_0} \leqRelationStrict[i] \outcomefrom{\sigma'}{v_0}$ for $i = 1,2$.

    Let us slightly modify the relation of player~$1$ such that $\leqRelationStrict[1]$ is defined from a lexicographic order using two objectives: a min-cost-reachability objective $\Omega_1^1$ with $T_1$ and a B\"uchi objective $\Omega_1^2$ with $T_2$.
    We have $x \leqRelationStrict[1] y$ if and only if ($\Omega_1^1(x) < \Omega_1^1(y)$) or ($\Omega_1^1(x) = \Omega_1^1(y)$ and $\Omega_1^2(x) < \Omega_1^2(y)$). If we consider the two previous strategy profiles, $\sigma'$ is still an NE, but $\sigma$ is no longer an NE as player~$1$ has a profitable deviation. Indeed, with the memoryless strategy $\tau_1$ such that $\tau_1(v_0) = v_2$, we get $\outcomefrom{\sigma}{v_0} = (v_0v_3)^\omega \leqRelationStrict[1] \outcomefrom{\tau_1,\sigma_2}{v_0} = (v_0v_2)^\omega$.
    \lipicsEnd
\end{example}

\begin{example}\label{ex:no-nash-not-omega-recognizable}
    In this example, we show that there does not always exist an NE in games with $\omega$-automatic preference relations. Consider the simple one-player game $\game$ with two vertices $v_0, v_1$, the edges $(v_0,v_0), (v_0,v_1), (v_1,v_1)$, and whose preference relation $\leqRelationStrict[1]$ is defined from a max-reward-reachability objective with a target set $T = \{v_1\}$. This game has no NE from the initial vertex $v_0$.\footnote{A similar example is given in~\cite{LeRoux-Pauly-Equilibria}.} Indeed, if the strategy of player~$1$ is to loop on $v_0$, then he has a profitable deviation by going to $T$ at some point, and if his strategy is to loop $k$ times in $v_0$ and then go to $T$, then he has a profitable deviation by looping one more time in $v_0$ before going to $T$.
    \lipicsEnd
\end{example}

In this paper, we investigate the following problems.

\begin{problems*} \label{problem}
    \begin{itemize}
        \item The \emph{NE checking problem} is to decide, given a game $\game$, an initial vertex $v_0$, and a strategy profile $\sigma = (\sigma_i)_{i \in \Players}$ where each strategy $\sigma_i$ is defined by a Mealy machine $\machine{i}$, whether $\sigma$ is an NE from $v_0$ in $\game$.
        \item The \emph{NE outcome checking problem} is to decide, given a game $\game$ and a lasso $\pi$, whether $\pi$ is the outcome of an NE in $\game$.
        \item The \emph{NE existence problem} is to decide, given a game $\game$ and an initial vertex $v_0$, whether there exists an NE from $v_0$ in $\game$.
        \item The \emph{constrained NE existence problem} is to decide, given a game $\game$, an initial vertex $v_0$, and a lasso $\pi_i$ for each player~$i$, whether there exists an NE from $v_0$ in $\game$ with an outcome $\rho$ such that $\pi_i \leqRelationStrict[i] \rho$ for all players~$i \in \Players$.
    \end{itemize}
\end{problems*}

\subparagraph*{Main Results.}
 Let us state our main results. We consider games $\game = (\arena,(\leqRelationStrict[i])_{i \in \Players})$ on the arena $\arena = (V,E,\Players,(V_i)_{i\in\Players})$, where each preference relation $\mathord{\leqRelationStrict[i]} \subseteq V^\omega \times V^\omega$ is $\omega$-automatic. We denote by $\aut{A}_i$ the \DPW{} accepting $\leqRelationStrict[i]$ and by $\{0,1,\ldots,d_i\}$ its set of priorities. We say that a problem $L$ is $\mathsf{Parity}$-hard if there exists a polynomial reduction from the problem of deciding the winner of a two-player zero-sum parity game to $L$.

\begin{theorem}[restate=nashchecking,name=]
\label{theorem:nash-checking-pspace}
    The NE checking problem is \pspaceComplete{}.
\end{theorem}

\begin{theorem}\label{theorem:OutcomeCheck}
    The NE outcome checking problem is in \np{} $\cap$ \conp{} and $\mathsf{Parity}$-hard.
\end{theorem}

\begin{theorem}[restate=NEexistenceGeneral,name=]
\label{theorem:NEexistenceGeneral}
    The NE existence problem is exponential in $|V|$, $\Pi_{i \in \Players}|\aut{A}_i|$, and $\Sigma_{i\in \Players}d_i$, thus doubly exponential in $|\Players|$. If the number of players is fixed (resp.\ for a one-player game), this problem is in \exptime{} and \pspace{}-hard (resp.\ \pspace{}-complete).
\end{theorem}

\begin{theorem}[restate=constrainedNEexistenceGeneral,name=]
\label{theorem:constrainedNEexistenceGeneral}
    The constrained NE existence problem, with the constraints given by lassoes $(\pi_i)_{i\in \Players}$, is exponential in $|V|$, $\Pi_{i\in \Players}|\aut{A}_i|$, $\Pi_{i\in \Players}|\pi_i|$, and doubly exponential in $\Sigma_{i\in \Players}d_i$, thus also doubly exponential in $|\Players|$.
    If the number of players and each $d_i$ are fixed (resp.\ for a one-player game), this problem is in \exptime{} and \pspace{}-hard (resp.\ \pspaceComplete{}).
\end{theorem}

The proofs of these theorems are detailed in the next two sections. In \cref{section:omegaRec}, we reconsider the studied problems in the special case of games with \emph{$\omega$-recognizable} relations.

\section{NE Checking and NE Outcome Checking Problems}
\label{section:nash-checking}

We first prove \cref{theorem:nash-checking-pspace}, stating the \pspace{}-completeness of the NE checking problem. The hardness is limited to a sketch of proof; the full technical details are given in \cref{app:hardness-nash-checking}.

\begin{proof}[Proof of \cref{theorem:nash-checking-pspace}]
    We begin with the membership result. Given the Mealy machines $\machine{i} = (M_i, m_0^i, \alpha_U^i, \alpha_N^i)$, $i \in \Players$, and the strategies $\sigma_i$ they define, we have to check whether $\sigma = (\sigma_i)_{i \in \Players}$ is an NE from a given initial vertex $v_0$. Equivalently, we have to check whether there exists a strategy $\tau_i$ for some player~$i$ such that $\outcomefrom{\sigma}{v_0} \leqRelationStrict[i] \outcomefrom{\tau_i,\sigma_{-i}}{v_0}$ (in which case $\sigma$ is not an NE). That is, whether there exists $i$ such that the language
    \[
    L_i = \{(x,y) \in V^\omega \times V^\omega \mid x \leqRelationStrict[i] y, ~ x = \outcomefrom{\sigma}{v_0}, ~ y \text{ consistent with $\sigma_{-i}$ and starting at $v_0$} \}
    \] is non-empty. We are going to describe a generalized \DPW{} $\aut{B}_i$, with a conjunction of three parity conditions, that accepts $L_i$. We proceed as follows.
    \begin{enumerate}
        \item The set $\{(x,y) \in V^\omega \times V^\omega \mid x \leqRelationStrict[i] y\}$ is accepted by the given \DPW{} $\aut{A}_i$ that accepts $\leqRelationStrict[i]$.
        \item The outcome $\outcomefrom{\sigma}{v_0}$ is a lasso obtained from the product of the arena $\arena$ and all $\machine{j}$. We can define a \DPW{}, of size exponential in the number of players, that only accepts $\outcomefrom{\sigma}{v_0}$.
        \item Finally, consider the product $\arena'$ of the arena $\arena$ with all $\machine{j}$, with $j \ne i$. We denote by $V'$ the set of vertices of $\arena'$, where each vertex is of the form $(v,(m_j)_{j \neq i})$, with $v \in V$ and $m_j$ a memory state of $\machine{j}$. The set of plays $y$ consistent with $\sigma_{-i}$ and starting at $v_0$ is accepted by a \DPW{} whose set of states is $V' \cup \{s_0,s_{rej}\}$ with initial state $s_0$, a rejecting sink $s_{rej}$ of priority $1$, and all other states of priority $0$, and whose transition function $\delta$ is such that $\delta(s_0,v_0) = (v_0,(m^j_0)_{j\neq i})$, and $\delta((v,(m_j)_{j\neq i}),v') = (v',(m'_j)_{j\neq i})$ for $(v,v') \in E$, $m'_j = \alpha_U^j(m_j,v)$, and when $v \in V_j$ for $j \neq i$, $v' = \alpha_N^j(m_j,v)$. Other letters lead to a rejecting sink state. Note that $\delta$ is a function as each $\machine{j}$ is deterministic and that this \DPW{} is of exponential size in the number of players.
    \end{enumerate}
    The announced automaton $\aut{B}_i$ is the product of the automata defined in the previous steps. It has exponential size and can be constructed on the fly. To check whether $L_i$ is non-empty, we guess an accepting lasso of exponential length, retaining the current state, the cycle-start state, binary length counters, and the maximum priority occurring on the cycle for each priority function. Each state has a polynomial-size encoding, so this uses polynomial space. Guessing $i$ gives a nondeterministic polynomial-space procedure for the complement of NE checking. This gives that the NE checking problem belongs to \pspace{}.

    \medskip

    \begin{figure}[t]
    \centering
        \begin{tikzpicture}[x=0.75pt,y=0.75pt,yscale=-1,scale=.8,every node/.style={scale=.8}]
            \draw (243.76,108.21) .. controls (243.76,108.21) and (243.76,108.21) .. (243.76,108.21) -- (274.91,108.21) .. controls (274.91,108.21) and (274.91,108.21) .. (274.91,108.21) -- (274.91,126) .. controls (274.91,126) and (274.91,126) .. (274.91,126) -- (243.76,126) .. controls (243.76,126) and (243.76,126) .. (243.76,126) -- cycle;
            \draw [fill={rgb,255:red,0;green,0;blue,0},fill opacity=0.03,dash pattern={on 4.5pt off 4.5pt}] (285.19,114.01) .. controls (285.19,104.47) and (292.93,96.73) .. (302.47,96.73) -- (420.43,96.73) .. controls (429.98,96.73) and (437.71,104.47) .. (437.71,114.01) -- (437.71,174.12) .. controls (437.71,183.66) and (429.98,191.4) .. (420.43,191.4) -- (302.47,191.4) .. controls (292.93,191.4) and (285.19,183.66) .. (285.19,174.12) -- cycle;
            \draw (301.36,110) .. controls (301.36,106.64) and (304.08,103.91) .. (307.44,103.91) -- (345.46,103.91) .. controls (348.82,103.91) and (351.55,106.64) .. (351.55,110) -- (351.55,122.83) .. controls (351.55,126.19) and (348.82,128.91) .. (345.46,128.91) -- (307.44,128.91) .. controls (304.08,128.91) and (301.36,126.19) .. (301.36,122.83) -- cycle;
            \draw (274.91,116.21) -- (297.96,116.02);
            \draw [shift={(300.96,116)}, rotate=179.53,fill={rgb,255:red,0;green,0;blue,0},line width=0.08,draw opacity=0] (5.36,-2.57) -- (0,0) -- (5.36,2.57) -- (3.56,0) -- cycle;
            \draw (250.8,150.46) .. controls (250.8,150.46) and (250.8,150.46) .. (250.8,150.46) -- (269.06,150.46) .. controls (269.06,150.46) and (269.06,150.46) .. (269.06,150.46) -- (269.06,168.07) .. controls (269.06,168.07) and (269.06,168.07) .. (269.06,168.07) -- (250.8,168.07) .. controls (250.8,168.07) and (250.8,168.07) .. (250.8,168.07) -- cycle;
            \draw (259.56,126) -- (259.49,146.38);
            \draw [shift={(259.48,149.38)}, rotate=270.19,fill={rgb,255:red,0;green,0;blue,0},line width=0.08,draw opacity=0] (5.36,-2.57) -- (0,0) -- (5.36,2.57) -- (3.56,0) -- cycle;
            \draw (262.51,168.21) .. controls (270.01,186.6) and (249.69,188.47) .. (254.81,171.19);
            \draw [shift={(255.71,168.61)}, rotate=111.8,fill={rgb,255:red,0;green,0;blue,0},line width=0.08,draw opacity=0] (5.36,-2.57) -- (0,0) -- (5.36,2.57) -- (3.56,0) -- cycle;
            \draw (301.36,160.4) .. controls (301.36,157.04) and (304.08,154.31) .. (307.44,154.31) -- (345.46,154.31) .. controls (348.82,154.31) and (351.55,157.04) .. (351.55,160.4) -- (351.55,173.23) .. controls (351.55,176.59) and (348.82,179.31) .. (345.46,179.31) -- (307.44,179.31) .. controls (304.08,179.31) and (301.36,176.59) .. (301.36,173.23) -- cycle;
            \draw (380.36,110) .. controls (380.36,106.64) and (383.08,103.91) .. (386.44,103.91) -- (424.46,103.91) .. controls (427.82,103.91) and (430.55,106.64) .. (430.55,110) -- (430.55,122.83) .. controls (430.55,126.19) and (427.82,128.91) .. (424.46,128.91) -- (386.44,128.91) .. controls (383.08,128.91) and (380.36,126.19) .. (380.36,122.83) -- cycle;
            \draw (380.56,161.4) .. controls (380.56,158.04) and (383.28,155.31) .. (386.64,155.31) -- (424.66,155.31) .. controls (428.02,155.31) and (430.75,158.04) .. (430.75,161.4) -- (430.75,174.23) .. controls (430.75,177.59) and (428.02,180.31) .. (424.66,180.31) -- (386.64,180.31) .. controls (383.28,180.31) and (380.56,177.59) .. (380.56,174.23) -- cycle;
            \draw (533.23,129.53) .. controls (533.23,129.53) and (533.23,129.53) .. (533.23,129.53) -- (566.27,129.53) .. controls (566.27,129.53) and (566.27,129.53) .. (566.27,129.53) -- (566.27,152.05) .. controls (566.27,152.05) and (566.27,152.05) .. (566.27,152.05) -- (533.23,152.05) .. controls (533.23,152.05) and (533.23,152.05) .. (533.23,152.05) -- cycle;
            \draw (259.89,94.27) -- (259.83,105.65);
            \draw [shift={(259.81,108.65)}, rotate=270.31,fill={rgb,255:red,0;green,0;blue,0},line width=0.08,draw opacity=0] (5.36,-2.57) -- (0,0) -- (5.36,2.57) -- (3.56,0) -- cycle;
            \draw (426.86,143.36) .. controls (453.96,134.07) and (509.9,136.57) .. (530.21,141.69);
            \draw [shift={(532.83,142.43)}, rotate=197.65,fill={rgb,255:red,0;green,0;blue,0},line width=0.08,draw opacity=0] (5.36,-2.57) -- (0,0) -- (5.36,2.57) -- (3.56,0) -- cycle;
            \draw (544.91,129.01) .. controls (537.66,107.62) and (561.04,108.12) .. (555.5,125.97);
            \draw [shift={(554.51,128.61)}, rotate=293.33,fill={rgb,255:red,0;green,0;blue,0},line width=0.08,draw opacity=0] (5.36,-2.57) -- (0,0) -- (5.36,2.57) -- (3.56,0) -- cycle;
            \draw (245.9,113.4) node [anchor=north west,inner sep=0.75pt,align=left] {$v_{init}$};
            \draw (308.1,109) node [anchor=north west,inner sep=0.75pt,font=\normalsize,align=left] {$(q_{1},1)$};
            \draw (252.6,152) node [anchor=north west,inner sep=0.75pt,align=left] {$\#$};
            \draw (304.1,159.4) node [anchor=north west,inner sep=0.75pt,font=\normalsize,align=left] {$(q_{m},1)$};
            \draw (387.1,109) node [anchor=north west,inner sep=0.75pt,font=\normalsize,align=left] {$(q_{1},n)$};
            \draw (383.3,160.4) node [anchor=north west,inner sep=0.75pt,font=\normalsize,align=left] {$(q_{m},n)$};
            \draw (448.05,155.13) node [anchor=north west,inner sep=0.75pt,font=\footnotesize,align=left] {if $q_{i} \in \{q_{accept},q_{reject}\}$};
            \draw (536.63,136.8) node [anchor=north west,inner sep=0.75pt,font=\normalsize,align=left] {$v_{end}$};
            \draw (321.93,135.47) node [anchor=north west,inner sep=0.75pt,font=\footnotesize,align=left] {Connected part};
        \end{tikzpicture}
        \caption{The game used for \pspace{}-hardness of \cref{theorem:nash-checking-pspace}.}
        \label{fig:reduction-nash-checking-game}
    \end{figure}

    For the \pspaceHard{}ness, we use a reduction from the membership problem for linear bounded deterministic Turing machines (LBTMs), known to be \pspaceComplete{}~\cite{books-lba-pspace-complete-GareyJ79}, to the complement of the NE checking problem. Recall that an LBTM $T$ has a limited memory such that the tape head must remain in the $n$ cells that contain the input word $w$.

    We give only a sketch of proof. First, let us show how we encode any configuration of the LBTM. For the current word written on the tape, we associate one player per cell, and we say that the letter in the $i$-th cell, $i \in \{1, \ldots, n\}$, is the current memory state of the Mealy machine $\aut{M}_i$ of player~$i$. Then we define an arena where each vertex is of the form $(q,i)$, for a state $q$ of $T$ and the current position $i$ of the tape head, and such that player~$i$ owns all the vertices $(q,i)$. Second, we simulate transitions of the LBTM with the Mealy machines: $\aut{M}_i$ can describe the next vertex according to its memory state. For example, from vertex $(q,i)$ and memory state $a$ for player~$i$, $\aut{M}_i$ moves to vertex $(q',i+1)$ and updates its memory state to $a'$ if the LBTM says that from state $q$ and letter $a$, the tape head must write $a'$ and go right, and that the next state is $q'$.

    This construction allows us to completely simulate the LBTM with an arena, described in \cref{fig:reduction-nash-checking-game}. We add an extra player~$n+1$ who decides whether to let the other players follow their Mealy machine to simulate the LBTM on the given word, or go to a sink state $\#$. With his preference relation $\leqRelationStrict[n+1]$, player~$n+1$ prefers a play visiting a vertex $(q_{accept},i)$, for any $i$, to any other play. His Mealy machine goes from $v_{init}$ to $\#$. Thus, by defining $\mathord{\leqRelationStrict[i]} = \emptyset$ for other players, the strategy profile given by all Mealy machines is not an NE if and only if it is profitable for player~$n+1$ to let the other players simulate the LBTM on $w$, i.e., this simulation visits $q_{accept}$.
\end{proof}

Let us now prove \cref{theorem:OutcomeCheck} stating the complexity of the NE outcome checking problem.

\begin{proof}[Proof of \cref{theorem:OutcomeCheck}]
    Let us begin with the membership result. Given a lasso $\pi$ starting at $v_0$, we have to decide whether there exists a strategy profile $\sigma = (\sigma_i)_{i\in\Players}$ with outcome $\pi$ such that, for all $i \in \Players$ and all strategies $\tau_i$, we have $\pi \not\leqRelationStrict[i] \outcomefrom{\tau_i,\sigma_{-i}}{v_0}$. We construct a parity game for each player~$i$ to test whether the other players can prevent a profitable deviation from $\pi$.

    Let us consider $L_i = \{x \in V^\omega \mid \pi \not\leqRelationStrict[i] x\}$. This set is accepted by a \DPW{} $\aut{B}_i$ constructed as the product of the complement of $\aut{A}_i$ and the lasso $\pi$. Clearly, the size of $\aut{B}_i$ is polynomially bounded in the sizes of $\aut{A}_i$ and $\pi$.
    Thus, $L_i$ contains all plays that are not preferred to $\pi$ by player~$i$. We construct a zero-sum parity game $\mathcal{H}_i$ of polynomial size whose winning regions describe where the coalition~$-i$ can ensure membership in $L_i$.

    Suppose that $\aut{B}_i$ has a set $Q$ of states, an initial state $q_0$, and a transition function $\delta_{\aut{B}_i} : Q \times V \rightarrow Q$. Let us define the game $\mathcal{H}_i$, where the two players are $A$ and $B$. Its set of vertices is the Cartesian product $V \times Q$, such that player~$A$ (resp.\ player~$B$) controls the vertices $(v,q)$ with $v \in V_{i}$ (resp.\ $v \not\in V_{i}$). In other words, $A$ has the role of player~$i$ while $B$ has the role of the coalition~$-i$. As $\aut{B}_i$ is deterministic, it is seen as an observer, and its states are information added to the vertices of $V$. Hence, the edges of $\mathcal{H}_i$ are of the form $((v,q),(v',q'))$ such that $(v,v') \in E$ and $q' = \delta_{\aut{B}_i}(q,v)$. We define a parity objective for player~$B$ as follows: the priority of each vertex $(v,q)$ of $\mathcal{H}_i$ is equal to the priority of $q$ in $\aut{B}_i$. Consequently, a play in $\mathcal{H}_i$ starting at $(v_0,q_0)$ is won by player~$B$ if and only if the projection on its first component belongs to $L_i$. For the constructed game $\mathcal{H}_i$, from every vertex $(v,q)$, we can decide in \np{} $\cap$ \conp{} which player wins in $\mathcal{H}_i$ together with a memoryless winning strategy for that player~\cite{lncs2500}.

    Let $q_n^i$ be the state of $\aut{B}_i$ after reading $\pi_0\cdots\pi_{n-1}$, with $q_0^i=q_0$. We claim that $\pi$ is an NE outcome if and only if $\pi \not\leqRelationStrict[i] \pi$ for every player~$i$ and every vertex $(\pi_n,q_n^i)$ with $\pi_n\in V_i$ is winning for player~$B$ in $\mathcal{H}_i$. Necessity of the first condition follows by taking the unchanged strategy. If the second condition fails, player~$i$ can follow $\pi$ up to $\pi_n$ and then use a winning strategy for player~$A$, yielding a profitable deviation.

    Conversely, follow $\pi$ until the first divergent edge, whose source identifies the deviator~$i$. Thereafter, the other players use their components of a winning strategy for player~$B$ in $\mathcal{H}_i$. Since the source is owned by player~$A$ and winning for player~$B$, every successor is winning for player~$B$. Hence the first divergent edge preserves the punishment guarantee. Distinct first-divergence histories have disjoint continuation subtrees, so these prescriptions define one strategy profile. A deviation leaving the outcome unchanged is excluded by the first condition. This proves the claim.

    The conditions $\pi \not\leqRelationStrict[i] \pi$ are polynomial-time lasso checks. The second condition concerns polynomially many vertices on the product lasso and can be certified by memoryless winning strategies for player~$B$. Its failure can be certified by one such vertex and a memoryless winning strategy for player~$A$. This gives membership in \np{} $\cap$ \conp{}.

    \medskip

    We continue with the hardness result, with a reduction from the problem of deciding whether player~$1$ has a winning strategy in a zero-sum parity game. We reduce this problem to the complement of the NE outcome checking problem, to establish its $\mathsf{parity}$-hardness. Let $\mathcal{H}$ be a parity game with players~$1$ and $2$, an arena $\arena$ with $V$ as set of vertices, an initial vertex $v_0$, and a priority function $\alpha: V \rightarrow \{0,\ldots,d\}$. We construct a new game $\game = (\arena', \leqRelationStrict[1],\leqRelationStrict[2])$ with the same players, whose arena $\arena'$ is a copy of $\arena$ with an additional vertex $v_0'$ owned by player~$1$, with $v_0$ and itself as successors (see \cref{fig:outcome-checking-parity-reduction}). Given $V' = V \cup \{v'_0\}$, the preference relation $\leqRelationStrict[2]$ is empty, accepted by a one-state \DPW{}, while the preference relation $\leqRelationStrict[1]$ is defined as follows: $x \leqRelationStrict[1] y$ if and only if $x = (v_0')^\omega$ and $y = (v_0')^m y'$, with $m \geq 0$ and $y'$ is a play in $\aut{H}$ starting at $v_0$ and satisfying the parity condition $\alpha$. A \DPW{} $\aut{A}_1$ accepting $\leqRelationStrict[1]$ is depicted in \cref{fig:outcome-checking-parity-reduction}, it is constructed with a copy of the arena $\arena$ and a new state $q_0$ with priority~1.

    The proposed reduction is correct. If player~$1$ wins in $\mathcal{H}$, then from every profile with outcome $\pi=(v'_0)^\omega$, he can deviate by moving immediately to $v_0$ and using his winning strategy. Conversely, if player~$2$ wins in $\mathcal{H}$, let player~$1$ stay at $v'_0$ and let player~$2$ use her winning strategy after any finite exit delay. This profile has outcome $\pi$ and admits no profitable deviation. By determinacy, player~$1$ wins in $\mathcal{H}$ if and only if $\pi$ is not an NE outcome. Exchanging the players of the source parity game and complementing its priorities also gives a reduction to NE outcome checking itself.
\end{proof}

\begin{figure}[t]
    \centering
    \begin{minipage}[c]{0.4\textwidth}
        \centering
        \begin{tikzpicture}[x=0.75pt,y=0.75pt,yscale=-1]
            \draw (250.91,120.61) .. controls (250.91,115.98) and (254.46,112.23) .. (258.85,112.23) .. controls (263.24,112.23) and (266.79,115.98) .. (266.79,120.61) .. controls (266.79,125.23) and (263.24,128.98) .. (258.85,128.98) .. controls (254.46,128.98) and (250.91,125.23) .. (250.91,120.61) -- cycle;
            \draw (259.11,99.78) -- (258.91,109.23);
            \draw [shift={(258.85,112.23)},rotate=271.21,fill={rgb,255:red,0;green,0;blue,0},line width=0.08,draw opacity=0] (5.36,-2.57) -- (0,0) -- (5.36,2.57) -- (3.56,0) -- cycle;
            \draw [fill={rgb,255:red,0;green,0;blue,0},fill opacity=0.03,dash pattern={on 4.5pt off 4.5pt}] (320.66,48.95) .. controls (320.66,44.23) and (324.48,40.41) .. (329.19,40.41) -- (363.25,40.41) .. controls (367.97,40.41) and (371.79,44.23) .. (371.79,48.95) -- (371.79,78.63) .. controls (371.79,83.35) and (367.97,87.17) .. (363.25,87.17) -- (329.19,87.17) .. controls (324.48,87.17) and (320.66,83.35) .. (320.66,78.63) -- cycle;
            \draw (290.41,59.61) .. controls (290.41,55.22) and (293.96,51.67) .. (298.34,51.67) .. controls (302.72,51.67) and (306.27,55.22) .. (306.27,59.61) .. controls (306.27,63.99) and (302.72,67.54) .. (298.34,67.54) .. controls (293.96,67.54) and (290.41,63.99) .. (290.41,59.61) -- cycle;
            \draw (306.27,59.61) -- (329.2,59.69);
            \draw [shift={(332.2,59.7)},rotate=180.21,fill={rgb,255:red,0;green,0;blue,0},line width=0.08,draw opacity=0] (5.36,-2.57) -- (0,0) -- (5.36,2.57) -- (3.56,0) -- cycle;
            \draw (332.41,59.61) .. controls (332.41,55.22) and (335.96,51.67) .. (340.34,51.67) .. controls (344.72,51.67) and (348.27,55.22) .. (348.27,59.61) .. controls (348.27,63.99) and (344.72,67.54) .. (340.34,67.54) .. controls (335.96,67.54) and (332.41,63.99) .. (332.41,59.61) -- cycle;
            \draw (254.67,127.84) .. controls (250.15,146.02) and (265.81,146.57) .. (262.76,130.46);
            \draw [shift={(262.11,127.78)},rotate=73.97,fill={rgb,255:red,0;green,0;blue,0},line width=0.08,draw opacity=0] (5.36,-2.57) -- (0,0) -- (5.36,2.57) -- (3.56,0) -- cycle;
            \draw (266.79,120.61) -- (305.15,120.53);
            \draw [shift={(308.15,120.53)},rotate=179.89,fill={rgb,255:red,0;green,0;blue,0},line width=0.08,draw opacity=0] (5.36,-2.57) -- (0,0) -- (5.36,2.57) -- (3.56,0) -- cycle;
            \draw (294.97,51.94) .. controls (292.7,34.93) and (304,33.22) .. (301.88,48.84);
            \draw [shift={(301.37,51.74)},rotate=282.09,fill={rgb,255:red,0;green,0;blue,0},line width=0.08,draw opacity=0] (5.36,-2.57) -- (0,0) -- (5.36,2.57) -- (3.56,0) -- cycle;
            \draw [fill={rgb,255:red,0;green,0;blue,0},fill opacity=0.03,dash pattern={on 4.5pt off 4.5pt}] (300.36,116.54) .. controls (300.36,111.22) and (304.67,106.91) .. (309.99,106.91) -- (371.83,106.91) .. controls (377.15,106.91) and (381.46,111.22) .. (381.46,116.54) -- (381.46,150.05) .. controls (381.46,155.37) and (377.15,159.68) .. (371.83,159.68) -- (309.99,159.68) .. controls (304.67,159.68) and (300.36,155.37) .. (300.36,150.05) -- cycle;
            \draw (307.57,120.91) .. controls (307.57,116.52) and (311.13,112.97) .. (315.51,112.97) .. controls (319.89,112.97) and (323.44,116.52) .. (323.44,120.91) .. controls (323.44,125.29) and (319.89,128.84) .. (315.51,128.84) .. controls (311.13,128.84) and (307.57,125.29) .. (307.57,120.91) -- cycle;
            \draw (361.91,123.76) .. controls (361.91,119.38) and (365.46,115.83) .. (369.84,115.83) .. controls (374.22,115.83) and (377.77,119.38) .. (377.77,123.76) .. controls (377.77,128.14) and (374.22,131.69) .. (369.84,131.69) .. controls (365.46,131.69) and (361.91,128.14) .. (361.91,123.76) -- cycle;
            \draw (326.15,129.61) .. controls (331.38,123.26) and (345.15,123.67) .. (359.02,123.75);
            \draw [shift={(361.91,123.76)},rotate=180,fill={rgb,255:red,0;green,0;blue,0},line width=0.08,draw opacity=0] (5.36,-2.57) -- (0,0) -- (5.36,2.57) -- (3.56,0) -- cycle;
            \draw (290.8,51.33) node [anchor=north west,inner sep=0.75pt,font=\small,align=left] {$v'_{0}$};
            \draw (247.39,77.97) node [anchor=north west,inner sep=0.75pt,align=left] {$\aut{A}_1$};
            \draw (353.76,68.47) node [anchor=north west,inner sep=0.75pt,align=left] {$\aut{H}$};
            \draw (252.1,116.23) node [anchor=north west,inner sep=0.75pt,font=\small,align=left] {$q_{0}$};
            \draw (333.2,55.93) node [anchor=north west,inner sep=0.75pt,font=\small,align=left] {$v_{0}$};
            \draw (255.69,39.07) node [anchor=north west,inner sep=0.75pt,align=left] {$\game$};
            \draw (267.43,105.4) node [anchor=north west,inner sep=0.75pt,font=\small]  {$v'_{0},v_{0}$};
            \draw (241.83,137.6) node [anchor=north west,inner sep=0.75pt,font=\small,align=left] {$v'_{0},v'_{0}$};
            \draw (303.06,143.77) node [anchor=north west,inner sep=0.75pt,align=left] {$\aut{H}$};
            \draw (308.36,117.23) node [anchor=north west,inner sep=0.75pt,font=\small,align=left] {$v_{0}$};
            \draw (329.47,124.4) node [anchor=north west,inner sep=0.75pt,font=\small,align=left] {$v'_{0},v$};
            \draw (365.3,120.88) node [anchor=north west,inner sep=0.75pt,font=\small,align=left] {$v$};
        \end{tikzpicture}
        \caption{The game $\game$ and the \DPW{} $\aut{A}_1$ accepting $\leqRelationStrict[1]$ for the reduction of \cref{theorem:OutcomeCheck}.}
        \label{fig:outcome-checking-parity-reduction}
    \end{minipage}
    \hfill
    \begin{minipage}[c]{0.58\textwidth}
        \centering
        \begin{tikzpicture}[x=0.75pt,y=0.75pt,yscale=-1,scale=.85,every node/.style={scale=.85}]
            \draw (255.77,25.02) .. controls (255.77,23.52) and (256.98,22.3) .. (258.49,22.3) -- (266.65,22.3) .. controls (268.15,22.3) and (269.37,23.52) .. (269.37,25.02) -- (269.37,34.38) .. controls (269.37,35.88) and (268.15,37.1) .. (266.65,37.1) -- (258.49,37.1) .. controls (256.98,37.1) and (255.77,35.88) .. (255.77,34.38) -- cycle;
            \draw (262.5,38.3) -- (262.5,50.1);
            \draw [shift={(262.5,53.1)}, rotate=270,fill={rgb,255:red,0;green,0;blue,0},line width=0.08,draw opacity=0] (3.57,-1.72) -- (0,0) -- (3.57,1.72) -- cycle;
            \draw (233.12,57.69) .. controls (233.12,55.74) and (234.7,54.17) .. (236.64,54.17) -- (290.4,54.17) .. controls (292.34,54.17) and (293.92,55.74) .. (293.92,57.69) -- (293.92,68.25) .. controls (293.92,70.19) and (292.34,71.77) .. (290.4,71.77) -- (236.64,71.77) .. controls (234.7,71.77) and (233.12,70.19) .. (233.12,68.25) -- cycle;
            \draw (262.5,71.9) -- (224.16,98.59);
            \draw [shift={(221.7,100.3)}, rotate=325.16,fill={rgb,255:red,0;green,0;blue,0},line width=0.08,draw opacity=0] (3.57,-1.72) -- (0,0) -- (3.57,1.72) -- cycle;
            \draw (262.5,71.9) -- (300.25,98.57);
            \draw [shift={(302.7,100.3)}, rotate=215.24,fill={rgb,255:red,0;green,0;blue,0},line width=0.08,draw opacity=0] (3.57,-1.72) -- (0,0) -- (3.57,1.72) -- cycle;
            \draw (280.27,104.31) .. controls (280.27,102.48) and (281.75,101) .. (283.58,101) -- (321.29,101) .. controls (323.12,101) and (324.6,102.48) .. (324.6,104.31) -- (324.6,114.25) .. controls (324.6,116.08) and (323.12,117.57) .. (321.29,117.57) -- (283.58,117.57) .. controls (281.75,117.57) and (280.27,116.08) .. (280.27,114.25) -- cycle;
            \draw (297.4,100.7) -- (297.4,117.9);
            \draw (220.1,116.3) -- (220.1,138.1);
            \draw [shift={(220.1,141.1)}, rotate=270,fill={rgb,255:red,0;green,0;blue,0},line width=0.08,draw opacity=0] (3.57,-1.72) -- (0,0) -- (3.57,1.72) -- cycle;
            \draw (302.5,117.3) -- (302.5,135.1);
            \draw [shift={(302.5,138.1)}, rotate=270,fill={rgb,255:red,0;green,0;blue,0},line width=0.08,draw opacity=0] (3.57,-1.72) -- (0,0) -- (3.57,1.72) -- cycle;
            \draw (253.43,141.98) .. controls (253.43,140) and (255.04,138.4) .. (257.01,138.4) -- (347.35,138.4) .. controls (349.33,138.4) and (350.93,140) .. (350.93,141.98) -- (350.93,152.72) .. controls (350.93,154.7) and (349.33,156.3) .. (347.35,156.3) -- (257.01,156.3) .. controls (255.04,156.3) and (253.43,154.7) .. (253.43,152.72) -- cycle;
            \draw (321.99,138.3) -- (321.99,156.5);
            \draw (302.5,156.9) -- (302.5,186.7);
            \draw [shift={(302.5,189.7)}, rotate=270,fill={rgb,255:red,0;green,0;blue,0},line width=0.08,draw opacity=0] (3.57,-1.72) -- (0,0) -- (3.57,1.72) -- cycle;
            \draw (302.4,206.5) -- (302.4,215.3);
            \draw [shift={(302.4,218.3)}, rotate=270,fill={rgb,255:red,0;green,0;blue,0},line width=0.08,draw opacity=0] (3.57,-1.72) -- (0,0) -- (3.57,1.72) -- cycle;
            \draw (213.77,103.82) .. controls (213.77,102.32) and (214.98,101.1) .. (216.49,101.1) -- (224.65,101.1) .. controls (226.15,101.1) and (227.37,102.32) .. (227.37,103.82) -- (227.37,113.18) .. controls (227.37,114.68) and (226.15,115.9) .. (224.65,115.9) -- (216.49,115.9) .. controls (214.98,115.9) and (213.77,114.68) .. (213.77,113.18) -- cycle;
            \draw (277.27,193.11) .. controls (277.27,191.28) and (278.75,189.8) .. (280.58,189.8) -- (324.59,189.8) .. controls (326.42,189.8) and (327.9,191.28) .. (327.9,193.11) -- (327.9,203.05) .. controls (327.9,204.88) and (326.42,206.37) .. (324.59,206.37) -- (280.58,206.37) .. controls (278.75,206.37) and (277.27,204.88) .. (277.27,203.05) -- cycle;
            \draw (297.4,189.5) -- (297.4,206.7);
            \draw (257.47,25.2) node [anchor=north west,inner sep=0.75pt,align=left] {$v$};
            \draw (191.67,39) node [anchor=north west,inner sep=0.75pt,font=\small,align=left] {\proverone{}: $v\rightarrow v'$?};
            \draw (177.67,76) node [anchor=north west,inner sep=0.75pt,font=\small,align=left] {\challenger{} accepts};
            \draw (235.67,55.6) node [anchor=north west,inner sep=0.75pt,font=\normalsize,align=left] {$v,(v,v')$?};
            \draw (301.2,69.1) node [anchor=north west,inner sep=0.75pt,font=\small,align=left] {{\small \challenger{} deviates with $v\rightarrow u$}\\{\footnotesize (where $v \in V_{j}$)}};
            \draw (282.47,103) node [anchor=north west,inner sep=0.75pt,align=left] {$v'$};
            \draw (299.27,104) node [anchor=north west,inner sep=0.75pt,align=left] {$u,j$};
            \draw (142.27,118) node [anchor=north west,inner sep=0.75pt,font=\small,align=left] {\proverone{}: $v'\rightarrow v''$?};
            \draw (213.53,147.2) node [anchor=north west,inner sep=0.75pt,align=left] {...};
            \draw (311.87,120.6) node [anchor=north west,inner sep=0.75pt,font=\small,align=left] {\proverone{}: $v'\rightarrow v''$?};
            \draw (254.88,139.8) node [anchor=north west,inner sep=0.75pt,font=\normalsize,align=left] {$v',(v',v'')$?};
            \draw (325.23,141.4) node [anchor=north west,inner sep=0.75pt,align=left] {$u,j$};
            \draw (312.2,157.3) node [anchor=north west,inner sep=0.75pt,font=\small,align=left] {{\small \challenger{}: $u\rightarrow u'$}{\footnotesize  (if $u\in V_{j}$)}\\{\small \provertwo{}: $u\rightarrow u'$}{\footnotesize  (if $u'\notin V_{j}$)}};
            \draw (295.47,221.53) node [anchor=north west,inner sep=0.75pt,align=left] {...};
            \draw (213.47,102) node [anchor=north west,inner sep=0.75pt,align=left] {$v'$};
            \draw (279.47,191.8) node [anchor=north west,inner sep=0.75pt,align=left] {$v''$};
            \draw (299.27,190.8) node [anchor=north west,inner sep=0.75pt,align=left] {$u',j$};
        \end{tikzpicture}
        \caption{An illustration of the \pcp{} game intuition: \proverone{} observes the left part of a vertex (all $v,v',v''$), while \challenger{} and \provertwo{} represent the deviating player~$j$ and the coalition~$-j$ in the right part (all $u,u'$). Given $\rho=vv'v''\dots$ and $\rho'=vuu'\dots$, \proverone{} and \provertwo{} aim to ensure $\rho \not\leqRelationStrict[j] \rho'$.}
        \label{fig:intuition-pcp-game}
    \end{minipage}
\end{figure}

\section{NE Existence and Constrained NE Existence Problems}
\label{section:NE-existence}

This section is devoted to the NE existence problem and its constrained variant. We mainly focus on the NE existence problem and explain at the end of the section how to take into account the constraints imposed on the NE outcome.

To solve the NE existence problem, we adapt a recent approach proposed in~\cite{SPE-NCRS-quanti-parity-P1CP2}. The idea is to reduce our problem to solving a \emph{three-player game with imperfect information}.\footnote{Although the underlying game structure of \cite{SPE-NCRS-quanti-parity-P1CP2} is reused, the player roles and correctness arguments differ entirely.} Let us first give some intuition (see also \cref{fig:intuition-pcp-game}) and then the formal definition. We use a reduction to a game with three players: two \emph{Provers} \proverone{} and \provertwo{} and one \emph{Challenger} \challenger. The two Provers aim to build an NE outcome $\rho$ while Challenger contests that it is an NE outcome: \proverone{} has the task of building $\rho$ edge by edge, while \provertwo{} has the task of showing that the deviation $\rho'$ of player~$i$ proposed by \challenger{} is not profitable, i.e., $\rho \not\leqRelationStrict[i] \rho'$.
We need two Provers (we cannot use a two-player zero-sum game), as the construction of $\rho$ cannot depend on one specific deviation and must be fixed, i.e., its construction cannot change according to the deviation $\rho'$ to artificially force $\rho \not\leqRelationStrict[i] \rho'$. This also means that $\proverone$ has to build $\rho$ without knowing when \challenger{} deviates: he has \emph{partial observation} of the game, while \challenger{} and \provertwo{} have perfect information.
This game, called \emph{\pcp{} game}, is articulated in two parts. The first part consists of vertices where \challenger{} does not deviate, where an action of \proverone{} is to suggest an edge $(v,v')$ to extend the current construction of $\rho$, and an action of \challenger{} is either to accept it or to deviate from $\rho$ by choosing another edge $(v,u)$ with $u \neq v'$. Such a deviation corresponds to a deviation by the player~$j$ who owns $v$, leading to the second part of the game. In this part, the vertices must retain the construction of the play $\rho$, the construction of the deviation $\rho'$, and the component $j$ to identify the player who deviated: \proverone{} continues to propose an extension $(v,v')$ for $\rho$ with no interaction with \challenger{}, and \challenger{} and \provertwo{}, representing respectively the deviating player~$j$ and the opposed coalition~$-j$, interact to construct $\rho'$. When the game stays in the first part, the objective of \proverone{} is to produce an NE outcome $\rho$, and if it goes in the second part, \proverone{} has the same goal and the aim of \provertwo{} is to retaliate on the deviations proposed by $\challenger{}$ to guarantee that $\rho'$ is not a profitable deviation. Hence, the vertices of the \pcp{} game also store the current states of the \DPWs{} accepting the preference relations, in a way to compare the outcome $\rho$ with the deviation $\rho'$.

We now proceed to the formal definition of the \pcp{} game. Suppose that we are given a game $\game = (\arena, (\leqRelationStrict[i])_{i \in \Players})$ with $\arena = (V,E,\Players,(V_i)_{i\in\Players})$ and $v_0 \in V$ as the initial vertex, and each relation $\leqRelationStrict[i]$ accepted by a \DPW{} $\aut{A}_i$. We denote each automaton as $\aut{A}_i = (Q_i,q_i^0,V \times V, \delta_i,\alpha_i)$ with $Q_i$ its set of states, $q_i^0$ its initial state, $V \times V$ its alphabet, $\delta_i : Q_i \times (V \times V) \rightarrow Q_i$ its transition function, and $\alpha_i : Q_i \rightarrow \{0,1,\ldots,d_i\}$ its priority function. The game
\[
\pcp(\game) = (S,(S_{\proveroneTight},S_{\challenger},S_{\provertwoTight}),(A_{\proveroneTight},A_{\challenger},A_{\provertwoTight}),
\Delta, Obs, W_{\proverone\provertwoTight})
\]
is a three-player game with partial observation for \proverone, defined as follows.
\begin{itemize}
    \item The set $S$ of vertices are of the form $(v,j,u,(q_i)_{i \in \Players})$ or $(v,j,u,(q_i)_{i \in \Players},(v,v'))$
    such that $v,u \in V$, $j \in \Players \cup \{\bot\}$, $q_i \in Q_i$, and $(v,v') \in E$. Coming back to the intuition given above, $v$ is the current vertex of $\rho$, $j$ is the deviating player (or $\bot$ if \challenger{} did not deviate yet), $u$ is the current vertex of $\rho'$ (if it exists, otherwise $u = v$), $q_i$ is the current state of $\aut{A}_i$ while comparing $\rho$ and $\rho'$.

    Given that we are looking for an NE in $\game$ from some initial vertex $v_0$, we consider the initial vertex $s_0 = (v_0,\bot,v_0,(q_i^0)_i)$ in the \pcp{} game.

    \item The set $S$ is partitioned as $S_{\proveroneTight} \cup S_{\challenger} \cup S_{\provertwoTight}$ such that $S_{\proveroneTight}$ is composed of the vertices $(v,j,u,(q_i)_{i})$, $S_{\challenger}$ is composed of the vertices $(v,j,u,(q_i)_{i},(v,v'))$ such that either $j = \bot$ and $v=u$, or $j \neq \bot$ and $u \in V_j$, and $S_{\provertwoTight}$ is composed of the vertices $(v,j,u,(q_i)_{i},(v,v'))$ such that $j \neq \bot$ and $u \in V \ssetminus V_j$.
    \item The set of actions\footnote{We introduce actions in a way to easily define the transition function $\Delta$.} is, respectively for each player, equal to:
 $A_{\proveroneTight} = \{(v,v') \mid (v,v') \in E\}$ (\proverone{} chooses an edge $(v,v')$ to extend the current construction of $\rho$) and $A_{\challenger} = A_{\provertwoTight} = V$ (\challenger{} and \provertwo{} choose the next vertex $u$ of $\rho'$ in case \challenger{} deviates, otherwise \challenger{} accepts the vertex $v'$ of the edge $(v,v')$ proposed by \proverone{}.)
    \item The transition function is defined as follows:
    \begin{itemize}
        \item for \proverone{}: for each $s = (v,j,u,(q_i)_{i}) \in S_{\proveroneTight}$ and each $(v,v') \in A_{\proveroneTight}$, we have $\Delta(s,(v,v')) = (v,j,u,(q_i)_{i},(v,v'))$,
        \item for \challenger{} who has not yet deviated: for each $s = (v,\bot,v,(q_i)_{i},(v,v')) \in S_{\challenger}$ and each $u \in A_{\challenger}$ with $(v,u) \in E$, we have
        either $u = v'$ and $\Delta(s,u) = (v',\bot,v',(q'_i)_{i})$ (which means that \challenger{} accepts the edge proposed by \proverone{}),
        or $u \neq v'$, $v \in V_j$, and $\Delta(s,u) = (v',j,u,(q'_i)_{i})$ (which means that \challenger{} starts deviating),
        with $q'_i = \delta(q_i,(v,v))$, for each $i \in \Players$, in both cases, i.e., the states of the \DPWs{} are updated.
        \item for \challenger{} who has deviated and \provertwo{}: for each $s = (v,j,u,(q_i)_{i},(v,v'))$ and each $u'$ with $(u,u') \in E$, we have either $u \in V_j$ and thus $s \in S_{\challenger}$, or $u \in V \ssetminus V_j$ and thus $s \in S_{\provertwoTight}$,
        and in both cases, $\Delta(s,u') = (v',j,u',(q'_i)_{i})$ with $q'_i = \delta(q_i,(v,u))$, for each $i \in \Players$.
    \end{itemize}
    \item The observation function $Obs$ for \proverone{}\footnote{Recall that \challenger{} and \provertwo{} have total observation of the \pcp{} game.} is such that $Obs((v,j,u,(q_i)_{i},(v,v'))) = (v,v')$ and
    $Obs((v,j,u,(q_i)_{i})) = v$.
    When $s, s' \in S$ and $Obs(s) = Obs(s')$, we consider that \proverone{} cannot distinguish $s$ and $s'$. Hence, \proverone{} can only observe the vertices $v$ of the initial game $\game$ and the edges $(v,v')$ that he proposes. We naturally extend $Obs$ to histories and plays of the \pcp{} game by applying the observation function on each of their vertices.
    \item To complete the definition of the \pcp{} game, it remains to define the winning condition $W_{\proverone\provertwoTight}$. Let us introduce some notation. Given a vertex $s$, we denote by $\projOne{s}$ (resp.\ $\projDev{s}$, $\projTwo{s}$) the projection on its first (resp.\ second, third) component. For a vertex $s \in S_{\challenger} \cup S_{\provertwoTight}$, we denote by $\projE{s}$ this last component of $s$. Note that if $s \in S_{\proveroneTight}$, then $Obs(s) = \projOne{s}$, and if $s \in S_{\challenger} \cup S_{\provertwoTight}$, then $Obs(s) = \projE{s}$. Given a play $\pi = \pi_0\pi_1\pi_2 \ldots $ of the \pcp{} game starting at the initial vertex $s_0$, $\pi$ is an alternation of vertices of $S_{\proveroneTight}$ and vertices of $S_{\challenger} \cup S_{\provertwoTight}$. Moreover, looking at the first (resp.\ third) components of the vertices of $\pi$, each such component is repeated from one vertex to the next one. Thus, we denote by $\projOne{\pi}$ the projection on the first component of the vertices of $\pi_0\pi_2 \ldots \pi_{2k} \ldots$. Similarly, we use notation $\projTwo{\pi}$ for the projection on the third component. We also define the notation $\projE{\pi}$ for the projection of $\pi_1\pi_3 \ldots \pi_{2k+1} \ldots$ on the last component of its vertices. Note that $\projOne{\pi} = Obs(\pi_0\pi_2 \ldots \pi_{2k} \ldots)$ and $\projE{\pi} = Obs(\pi_1\pi_3 \ldots \pi_{2k+1} \ldots)$. In the play $\pi$, either the second component always remains equal to $\bot$ or ultimately becomes equal to some $j \in \Players$. We use notation $\projDev{\pi}$ to denote this value $\bot$ or $j$. All these notations are also used for histories.

    The set $W_{\proverone\provertwoTight}$ is defined as $W_{\proverone\provertwoTight} = W_{acc} \cup W_{dev}$ where $W_{acc}$ is the set of plays where \challenger{} always agreed with \proverone{} and $W_{dev}$ is the set of plays where \challenger{} deviated but \provertwo{} was able to show that this deviation is not profitable, i.e.,
    \begin{itemize}
        \item $W_{acc} = \{\pi \in \Plays(\pcp(\game)) \mid \projDev{\pi} = \bot \}$,\footnote{In practice, we also have to check that every NE outcome $\rho$ is such that $\rho \not\leqRelationStrict[i] \rho$ for every player $i$, which is omitted in the current construction for readability. We can modify the construction by adding an initial choice $k \in \Players$ by \challenger{}, unobservable to \proverone{}, and modify $W_{acc}$ to $\{\pi \in \Plays(\pcp(\game)) \mid \projDev{\pi} = \bot,\, \projOne{\pi} \not\leqRelationStrict[k] \projOne{\pi}\}$ for $k$ the hidden choice along $\pi$. This increases the size by a factor $|\Players|$ and checks the condition.}
        \item $W_{dev} = \{\pi \in \Plays(\pcp(\game)) \mid \exists j \in \Players, \projDev{\pi} = j \text{ and } \projOne{\pi} \not\leqRelationStrict[j] \projTwo{\pi} \}$.
    \end{itemize}
    This set $W_{\proverone\provertwoTight}$ is the winning condition for both \proverone{} and \provertwo{} while \challenger{} has the complementary winning condition $S^\omega \ssetminus W_{\proverone\provertwoTight}$.

\end{itemize}

The next theorem states how the \pcp{} game helps to solve the NE existence problem. A strategy $\tau_{\proveroneTight}$ of \proverone{} is \emph{observation-based} if for all histories $h, h'$ ending in a vertex of \proverone{} such that $Obs(h) = Obs(h')$, we have $\tau_{\proveroneTight}(h) = \tau_{\proveroneTight}(h')$.

\begin{theorem}[restate=correspondence,name=]
\label{theorem:correspondence}
    The following statements are equivalent:
    \begin{itemize}
        \item In $\game$, there exists an NE $\sigma = (\sigma_i)_{i \in \Players}$ from $v_0$,
        \item In $\pcp(\game)$, there exists an observation-based strategy $\tau_{\proveroneTight}$ of \proverone{} such that for all strategies $\tau_{\challenger}$ of \challenger{}, there is a strategy $\tau_{\provertwoTight}$ of \provertwo{} such that $\outcomefrom{\tau_{\proveroneTight},\tau_{\challenger},\tau_{\provertwoTight}}{s_0} \in W_{\proverone\provertwoTight}$.
    \end{itemize}
\end{theorem}

\cref{theorem:correspondence} is the key tool to solve the NE existence problem. It is proved in detail in \cref{app:PCPgame}. We give hereafter a sketch of proof for the membership result of \cref{theorem:NEexistenceGeneral}, which follows the approach proposed in~\cite{SPE-NCRS-quanti-parity-P1CP2}. The \pspaceHard{}ness already holds for one-player games, with a reduction from the existence of a maximal element in a relation $\leqRelationStrict$, which is a \pspaceComplete{} problem and close to the existence of NEs in one-player games. All details are given in \cref{app:nash-existence-tools}.

\begin{proof}[Sketch of Proof of \cref{theorem:NEexistenceGeneral}, Membership]
    By \cref{theorem:correspondence}, deciding whether there exists an NE from $v_0$ in $\game$ reduces to deciding whether there exists an observation-based strategy $\tau_{\proveroneTight}$ of \proverone{} in $\pcp(\game)$ such that for all strategies $\tau_{\challenger}$ of \challenger{}, there is a strategy $\tau_{\provertwoTight}$ of \provertwo{} such that $\outcomefrom{\tau_{\proveroneTight},\tau_{\challenger},\tau_{\provertwoTight}}{s_0} \in W_{\proverone\provertwoTight}$. In~\cite{SPE-NCRS-quanti-parity-P1CP2}, the authors solve the problem they study by solving a similar three-player game with imperfect information. They proceed as follows: \emph{(i)} the winning condition is translated into a Rabin condition\footnote{Recall that a Rabin condition uses a finite set of pairs $(E_j,F_j)_{j \in J}$ in a way to accept plays $\pi$ such that there exists $j \in J$ with $\infOcc{\pi} \cap E_j = \varnothing$ and $\infOcc{\pi} \cap F_j \neq \varnothing$.} on the arena of the \pcp{} game, \emph{(ii)} the three-player game is transformed into a two-player zero-sum Rabin game with imperfect information, and finally \emph{(iii)} classical techniques to remove imperfect information are used to obtain a two-player zero-sum parity game with perfect information.

    In this sketch of proof, we only explain the first step, i.e., how to translate $W_{\proverone\provertwoTight} = W_{acc} \cup W_{dev}$ into a Rabin condition, as the second and third steps heavily use the arguments of~\cite{SPE-NCRS-quanti-parity-P1CP2}.
    To translate $W_{acc}$, we use one pair $(E_1,F_1)$ such that $E_1 = \varnothing$ and $F_1 = \{s \in S \mid \projDev{s} = \bot \}$. To translate $W_{dev}$, notice that $\projDev{\pi} = j$ is equivalent to $\projDev{\pi} \not\in \{\bot\} \cup \Players\ssetminus\{j\}$, and thus $W_{dev} = \cup_{j \in \Players} \{\pi \in \Plays(\pcp(\game)) \mid \projDev{\pi} \not\in \{\bot\} \cup \Players\ssetminus\{j\} \text{ and } \projOne{\pi} \not\leqRelationStrict[j] \projTwo{\pi} \}$. Recall that each relation $\leqRelationStrict[j]$ is accepted by the \DPW{} $\aut{A}_j$ with the priority function $\alpha_j : Q_j \rightarrow \{0,1, \ldots d_j\}$, thus also $\not\leqRelationStrict[j]$ with the modified priority function $\alpha_j+1$. Therefore, $W_{dev}$ can be translated into a Rabin condition on the vertices of $S$ with $\Sigma_{j \in \Players} d_j$ Rabin pairs~\cite{handbook-of-model-checking-orna-kupferman}. Steps \emph{(ii)} and \emph{(iii)} are detailed in \cref{app:nash-existence-tools}, leading to the announced complexity: the NE existence problem is exponential in $|V|$, $\Pi_{i\in \Players}|\aut{A}_i|$, and $\Sigma_{_{i\in \Players}}d_i$.
\end{proof}

Let us finally comment on \cref{theorem:constrainedNEexistenceGeneral} stating the complexity class of the constrained NE existence problem. The detailed proof is presented in \cref{app:constrained-nash-existence}. The approach to proving membership is very similar to that of the NE existence problem, as we only need to modify $W_{acc}$ in a way to include the constraints imposed on the NE outcome. A constraint imposed by a lasso $\pi_i$ can be represented by a \DPW{} $\aut{A}_i'$ accepting the language $\{\rho \in V^\omega \mid \pi_i \leqRelationStrict[i] \rho\}$, with a polynomial size $|\aut{A}_i|\cdot |\pi_i|$. Then it suffices to extend the arena of the \pcp{} game with the states of each $\aut{A}_i'$. The hardness result is obtained by a reduction from the NE existence problem, already for one-player games.

Note that by steps \emph{(i)}-\emph{(iii)}, solving the (constrained) NE existence problem is equivalent to solving a zero-sum parity game with memoryless winning strategies for both players. Therefore, we get the following property:

\begin{corollary}\label{cor:existence-ne-finite-memory}
If there exists a (constrained) NE, then there exists one with finite-memory strategies.
\end{corollary}

There is a great interest in using the concept of \pcp{} game, as it provides a unified approach to solve the NE existence problem and its constrained variant. With this approach, we could also decide the existence of an NE whose outcome $\rho$ satisfies various combinations of constraints, such as, e.g., $\pi_i \leqRelationStrict[i] \rho \leqRelationStrict[i] \pi'_i$ for one or several players~$i$. The chosen constraints only impact the winning condition $W_{\proverone\provertwo}$ and thus its translation into a Rabin condition.

\section{Hypotheses on Preference Relations}
\label{section:properties-relations}

In the previous sections, we presented several decision algorithms. Since the players' relations $\leqRelationStrict[i]$ are intended to formalize how they prefer one play to another, we may naturally expect them to satisfy certain properties, such as \emph{irreflexivity} and \emph{transitivity}. However, since the relation $\leqRelationStrict[i]$ is accepted by a \DPW{}, its structure can be intricate, and it becomes relevant to verify whether such properties hold. In this section, we address decision problems related to the algorithmic verification of properties of $\leqRelationStrict[i]$. We also explore alternative approaches to modeling preferences between plays, focusing in particular on cases where the \DPW{} accepts a non-strict preference relation $\leqRelation[i]$, i.e., a preorder, rather than a strict partial order.

\subparagraph*{Hypotheses on Relations.} Given a relation $R \subseteq \Sigma^\omega \times \Sigma^\omega$, it is:
\begin{itemize}
    \item \emph{reflexive} (resp.\ \emph{irreflexive}) if for all $x \in \Sigma^\omega$, we have $(x,x) \in R$ (resp.\ $(x,x) \not\in R$),
    \item \emph{transitive} (resp.\ \emph{$\neg$-transitive}) if for all $x,y,z \in \Sigma^\omega$, we have $(x,y) \in R \wedge (y,z) \in R \Rightarrow (x,z) \in R$ (resp.\ $(x,y) \not\in R \wedge (y,z) \not\in R \Rightarrow (x,z) \not\in R$),
    \item \emph{total} if for all $x,y \in \Sigma^\omega$, we have $x = y \vee (x,y) \in R \vee (y,x) \in R$,
\end{itemize}
A reflexive and transitive relation is a \emph{preorder}. An irreflexive and transitive relation is a \emph{strict partial order}. When, in addition, a strict partial order $R$ is $\neg$-transitive, it is a \emph{strict weak order}. The next proposition states that all the relevant properties mentioned above can be \emph{efficiently} verified on the \DPW{} accepting the relation $R$. It is proved in \cref{app:proof-PropertiesRelations}.

\begin{proposition}[restate=propertiesrelations,name=]
\label{prop:PropertiesRelations}
    The problem of deciding whether an $\omega$-automatic relation $R$ is reflexive (resp.\ irreflexive, transitive, $\neg$-transitive, total) is \nlComplete{}.
\end{proposition}

\subparagraph*{Variants on our Setting.}
Let us first observe that all the lower bounds established for the decision problems about NEs remain valid even when the players' preference relations $R_i$ are assumed to be \emph{strict partial orders}.\footnote{The upper bounds do not need more than the $\omega$-automaticity of each $R_i$} This implies that taking this additional property into account does not yield any important advantage in terms of complexity class.

We now consider an alternative setting in which each relation $R_i$ is not a strict preference, but rather a \emph{preorder} (where some plays are declared equivalent). To discuss this further, let us recall the relationship between strict partial orders and preorders. From a strict partial order $\leqRelationStrict$, one can obtain a preorder $\leqRelation$ by taking its reflexive closure, i.e., $x \leqRelation y$ if $x \leqRelationStrict y$ or $x = y$. When $\leqRelationStrict$ is a strict weak order, another preorder which is total, is obtained by defining $x \leqRelation y$ if $y \not\leqRelationStrict x$. In both cases, if $\leqRelationStrict$ is $\omega$-automatic, $\leqRelation$ is also $\omega$-automatic (resp.\ by a generalized \DPW{} with a disjunction of two parity conditions, or by a \DPW{}). Conversely, from any preorder $\leqRelation$, we can define a strict partial order $\leqRelationStrict$ such that $x \leqRelationStrict y$ if $x \leqRelation y \wedge y \not\leqRelation x$. Every strict partial order can be constructed this way. Moreover, if $\leqRelation$ is total, then $\leqRelationStrict$ is a strict weak order. We can also define the equivalence relation $\equivRelation$ such that $x \equivRelation y$ if $x \leqRelation y \wedge y \leqRelation x$. The equivalence class of $x$ is denoted $\equivClass{x}{}$. Again, if $\leqRelation$ is $\omega$-automatic, then $\leqRelationStrict$ and $\equivRelation$ are both $\omega$-automatic (by a generalized \DPW{} with a conjunction of two parity conditions).

For games with preorders $R_i$, we keep the same upper bounds, except that we only have an \np{} membership for~\cref{theorem:OutcomeCheck}. Indeed, it requires solving a generalized parity game with a disjunction of two parity conditions (instead of a parity game), solvable in \np{}~\cite{generalized-parity-game}. All lower bounds remain valid by carefully modifying the preference relations used in the reductions into preorders, see \cref{app:reduction-preorders}.

Finally, recall that given a lasso $\pi$, it is easy to construct an automaton that accepts all plays related to $\pi$ according to an $\omega$-automatic preference relation. This set being $\omega$-regular, all standard verification techniques for $\omega$-regular languages can be applied. For example, one may wish to verify that all plays preferred to $\pi$ satisfy a given $\omega$-regular property. In such cases, the full range of verification algorithms developed for $\omega$-regular languages can be used.

\begin{figure}
    \centering
    \begin{tikzpicture}[automaton,every node/.style={scale=.7},node distance=1.1]
        \node[system] (v0) {$v_0$};
        \node[system] (v1) [left=of v0] {$v_1$};
        \node[system] (v2) [right=of v0] {$v_2$};

        \path (v0) edge            (v1)
                   edge            (v2)
              (v1) edge[loop left] (v1)
              (v2) edge[loop right] (v2);
    \end{tikzpicture}
    \caption{An arena with one player.}
    \label{fig:aut-no-nash-game-multireach}
\end{figure}

\subparagraph*{Alternative Definition of NE.} Given a game $\game = (A,(\leqRelation[i])_{i\in \Players})$ with preorders $\leqRelation[i]$, an NE is a strategy profile $\sigma$ such that for all players~$i$ and all strategies $\tau_i$ of player~$i$, we have $\outcomefrom{\sigma}{v_0} \not\leqRelationStrict[i] \outcomefrom{\tau_i,\sigma_{-i}}{v_0}$. An alternative definition asks for all~$i$ and~$\tau_i$ that $\outcomefrom{\tau_i,\sigma_{-i}}{v_0} \leqRelation[i] \outcomefrom{\sigma}{v_0}$~\cite{PatriciaBouyerBMU15}. The two definitions yield different notions of NE (unless all $\leqRelation[i]$ are total). In this paper, we do not consider the second definition, due to the nonexistence of NEs in very simple games. Let us consider the one-player game $\game$ depicted in \cref{fig:aut-no-nash-game-multireach}, where, from the initial vertex $v_0$, player~$1$ has the choice between $v_0v_1^\omega$ and $v_0v_2^\omega$. We consider the preorder $\leqRelation[1]$ equal to $\{(x,x) \mid x \in \{v_0,v_1,v_2\}^\omega\}$. Clearly, $v_0v_1^\omega \not\leqRelation[1] v_0v_2^\omega$ and $v_0v_2^\omega \not\leqRelation[1] v_0v_1^\omega$, showing that there is no NE from $v_0$ for this alternative definition (while $v_0v_1^\omega$ and $v_0v_2^\omega$ are both NE outcomes with the first definition). This phenomenon appears as soon as there are two incomparable plays.

\section{\texorpdfstring{$\omega$}{ω}-Recognizable Relations}
\label{section:omegaRec}

In this section, we suppose that we have a game $\game = (\arena,(\leqRelation[i])_{i \in \Players})$ whose relations $\leqRelation[i]$ are \emph{$\omega$-recognizable} and \emph{preorders}. We recall that $\leqRelation[i]$ is $\omega$-recognizable if it is of the form $\cup_{i=1}^{\ell} X_i \times Y_i$ where $X_i, Y_i \subseteq \Sigma^\omega$ are $\omega$-regular languages over $\Sigma$. Any $\omega$-recognizable relation is $\omega$-automatic (see~\cite{BookSakarovitch}), and deciding whether an $\omega$-automatic relation accepted by a \DPW{} is $\omega$-recognizable is \nlComplete{}~\cite{lics23-rational-relations}. For each $\leqRelation[i]$, we use the related relations $\leqRelationStrict[i]$ and $\equivRelation[i]$ as defined in the previous section.

In \cref{ex:no-nash-not-omega-recognizable}, we presented a one-player game with no NE. The reason for the absence of NE is that $\leqRelationStrict[1]$ has an unbounded infinite ascending chain. This situation cannot happen for $\omega$-recognizable preorders, as highlighted in the next proposition, easily derived from~\cite{rational-relations-automatic-loding} (its proof is given in \cref{app:recognizable-finite-index}). This motivates the interest in games with $\omega$-recognizable preference relations.

\begin{proposition}[restate=recognizablefiniteindex,name=]
\label{prop:recognizable-finite-index}
    An $\omega$-automatic preorder $\mathord{\leqRelation} \subseteq \Sigma^\omega \times \Sigma^\omega$ is $\omega$-recognizable if and only if its induced equivalence relation \equivRelation{} has finite index.
\end{proposition}

Thanks to this result, the equivalence classes of $\equivRelation$ form a \emph{finite partially ordered set} under the order induced by $\leqRelation$. In particular, there always exists a maximal (resp.\ minimal) equivalence class. Examples of $\omega$-recognizable preorders are numerous: those deriving from any Boolean combination of $\omega$-regular objectives or any multidimensional objective where each dimension is defined using an $\omega$-regular objective.
In the subclass of games with $\omega$-recognizable preorders, the main difference is the \emph{existence of NEs}.

\begin{theorem}[restate=existenceNE,name=]
\label{theorem:existenceNE}
    When the preference relations of a game are all $\omega$-recognizable preorders, then there always exists an NE composed of finite-memory strategies.
\end{theorem}

The proof of \cref{theorem:existenceNE} requires two steps. We first prove the existence of an NE under the assumption that each preference relation $\leqRelation[i]$ is a \emph{total} preorder and then without this assumption. (Note that we get an NE composed of finite-memory strategies by \Cref{cor:existence-ne-finite-memory}). The first step also follows from~\cite[Theorem 15]{LeRoux-Pauly-Equilibria}: the induced preferences are strict weak orders, regular threshold games are determined, and finitely many equivalence classes ensure attained extrema and exclude infinite strict chains of guarantees. Recall that the relation $\leqRelationStrict[i]$ induced by a preorder $\leqRelation[i]$ is a strict weak order if $\leqRelation[i]$ is total. Nevertheless, we provide a proof of this first step in \cref{app:NEtotal}, inspired by the work of~\cite{Gradel-Ummels-08} and~\cite{BrihayePS13}, where the existence of NEs is studied through the concept of value and optimal strategy (see below and in \cref{app:NEtotal}).

The second step is obtained by extending partial preorders to total preorders, as described in the next proposition. \cref{theorem:existenceNE} easily follows (see \cref{app:generalNE}).

\begin{proposition}[restate=embeddingpartialtotal,name=]
\label{prop:partial-preorder-to-total-preorder}
    Any $\omega$-recognizable preorder \leqRelation{} has an $\omega$-recognizable total extension $\leqRelation'$ with the same equivalence classes. Moreover, for all $x,y$, if $x \Join y$, then $x \Join' y$, for $\Join$ $\in \{\leqRelation,\leqRelationStrict,\geqRelation,\geqRelationStrict,\equivRelation\}$.
\end{proposition}

We now focus on \emph{prefix-independent} relations $R$, such that for all $x,y \in \Sigma^\omega$ and $u,v \in \Sigma^*$, $(x,y) \in R \Leftrightarrow (ux,vy) \in R$.
From our proof of \cref{theorem:existenceNE}, when the relations $\leqRelation[i]$ are total and prefix-independent, we can derive the following \emph{characterization} of NE outcomes in terms of values (the proof is given in \cref{app:carac}). In this context, for each player~$i$ and vertex $v$ of $\game$, there always exists a \emph{value} $val_i(v)$ (which is an equivalence class of $\equivRelation[i]$) and \emph{optimal strategies} $\sigma_i^v$ for player~$i$ and $\sigma_{-i}^v$ for the coalition~$-i$ such that $\sigma_i^v$ (resp.\ $\sigma_{-i}^v$) ensures consistent plays $\pi$ starting at $v$ such that $val_i(v) \leqRelation[i] [\pi]_{i}$ (resp.\ $[\pi]_{i} \leqRelation[i] val_i(v)$) (see \cref{app:NEtotal}). Such an NE characterization is well-known for games with classical objectives (see, e.g., the survey~\cite{Bruyere17}).

\begin{theorem}[restate=characterization,name=]
\label{theorem:characterization}
    Let $\game$ be a game such that each preference relation $\leqRelation[i]$ is an $\omega$-recognizable preorder, total, and prefix-independent. Then a play $\rho = \rho_0\rho_1\ldots$ is an NE outcome if and only if for all vertices $\rho_n$ of $\rho$, if $\rho_n \in V_i$, then $val_i(\rho_n) \leqRelation[i] [\rho]_{i}$.
\end{theorem}

A relation $R$ is \emph{prefix-linear} if, for all $x,y \in \Sigma^\omega$ and $u \in \Sigma^*$, $(x,y) \in R$ implies $(ux,uy) \in R$. Under this weaker assumption, the conditions $val_i(\rho_n) \leqRelation[i] [\rho_{\geq n}]_{i}$ for all $\rho_n \in V_i$ are sufficient for $\rho$ to be an NE outcome, but need not be necessary: adding a prefix may erase a strict improvement. A characterization without any prefix assumption uses the values of the subgames after each prefix, as stated in \cref{theorem:characterization-residual}. Moreover, deciding whether a relation $R$ is prefix-independent (resp.\ prefix-linear) is \nlComplete{}. All proofs and details are provided in \cref{app:carac}.

\section{Conclusion}
\label{section:conclusion-future-work}

In this work, we have introduced a general framework for defining players' preferences via $\omega$-automatic preference relations instead of fixed reward functions. It subsumes several classical settings, including the Boolean setting with $\omega$-regular objectives and quantitative models such as min-cost-reachability, as well as combinations of several such objectives.

In this framework, we have studied the complexity of four fundamental problems related to NEs, notably with a novel use of the \pcp{} game setting recently introduced in \cite{SPE-NCRS-quanti-parity-P1CP2}. This approach enables a broader applicability and more reusable results. It contrasts sharply with most existing work that is typically focused on specific reward functions.

We hope that our framework will serve as a basis for exploring additional problems such as decision problems about subgame perfect equilibria (which are NEs in any subgame of a game~\cite{Osborne1994}), or the rational synthesis problem as studied in \cite{KupfermanPV16}. New results will lead to the development of general and modular solutions for a wider class of questions in the theory of infinite games played on graphs.

\bibliography{bibliography}

\appendix

\section{Hardness of \texorpdfstring{\cref{theorem:nash-checking-pspace}}{Theorem~\ref{theorem:nash-checking-pspace}}}
\label{app:hardness-nash-checking}

\nashchecking*

\begin{proof}[Proof of \cref{theorem:nash-checking-pspace}, hardness]
    We use a reduction from the membership problem for linear bounded deterministic Turing machines (LBTMs), also called Linear Bounded Automata~\cite{sipser13}, known to be \pspaceComplete{}~\cite{books-lba-pspace-complete-GareyJ79}, to the complement of the NE checking problem. Recall that an LBTM $T$ has a limited memory such that the tape head must remain in cells containing the input word $w$. The membership problem asks whether $w$ is accepted by $T$.

    Suppose that $T$ is defined on a set of states $Q = \{q_1,\dots,q_m\}$ with $q_1$ its initial state and containing $q_{accept}, q_{reject}$, the word $w$ is equal to $w_1w_2\dots w_n$, and the alphabet is $\Sigma = \{a_1,\dots, a_s\}$. We construct a $(n+1)$-player game $\game = (V,E,(V_i)_{1\leq i\leq n+1},(\leqRelationStrict[i])_{1 \leq i\leq n+1})$ as illustrated in \cref{fig:reduction-nash-checking-game}, and $n+1$ Mealy machines $(\machine{i})_{1\leq i\leq n+1}$ in the following way. Intuitively, each input tape cell is associated with a player and will be simulated by a Mealy machine for this player, encoding actions based on the current configuration of the LBTM. We begin with the description of the game:
    \begin{itemize}
        \item $V = (Q \times \{1,\dots,n\}) \cup \{v_{init},\#,v_{end}\}$, with three fresh vertices $v_{init}$, $\#$, and $v_{end}$,
        \item $v_{init}$ is the initial vertex and has two successors: $\#$ and $(q_1,1)$,
        \item $\#$ and $v_{end}$ are sink vertices with a self-loop,
        \item $V_{n+1} = \{v_{init},\#,v_{end}\}$ and for each $i \in \{1,\dots,n\}$, $V_i = Q \times \{i\}$,
        \item $(q,i) \rightarrow (q',j)$ for all $q,i,q',j$ such that $q \not\in \{q_{accept},q_{reject}\}$, otherwise $(q,i) \rightarrow v_{end}$,
        \item for each $i \in \{1,\dots,n\}$, the preference relation $\leqRelationStrict[i]$ is empty, therefore accepted by a one-state \DPW{},
        \item the preference relation $\leqRelationStrict[n+1]$ is defined such that $v_{init} \#^\omega \leqRelationStrict[n+1] v_{init} x$ if and only if $x$ contains $\{q_{accept}\} \times \{1,\dots,n\}$. More precisely, we define $\mathord{\leqRelationStrict[n+1]} = L^c \times L$ where $L = \{v_{init}x \in V^\omega \mid x \text{ visits } \{q_{accept}\} \times \{1,\dots,n\} \}$ and $L^c = V^\omega \setminus L$.
    \end{itemize}
    Finally, let us show that $\leqRelationStrict[n+1]$ is accepted by a \DBW{} of constant size. Both languages $L$ and $L^c$ are accepted by a \DBW{} as illustrated in \cref{fig:reduction-nash-checking-pref-n+1}. Thus, the relation $L^c \times L$ is accepted by a generalized \DBWs{} with a conjunction of two B\"uchi conditions of constant size, thus by a \DBW{} which is in particular a \DPW{}.

    Let us now describe the Mealy machines. The Mealy machine of player~$n+1$ defines the memoryless strategy $\sigma_{n+1}$ such that $\sigma_{n+1}(v_{init}) = \#$, $\sigma_{n+1}(\#)=\#$, and $\sigma_{n+1}(v_{end})=v_{end}$. For player~$i \in \{1,\dots,n\}$, the Mealy machine $\machine{i}$ simulates the LBTM $T$ when the tape head scans the $i$-th tape cell:
    \begin{itemize}
        \item the set of memory states is $M = \Sigma \cup \{a_{init},a_{end}\}$, with fresh symbols $a_{init}, a_{end} \not\in \Sigma$, where $a_{init}$ is the initial memory state,
        \item the update function $\alpha_U : M \times V \rightarrow M$ is defined by $\alpha_U(a_{init},v_{init})=w_i$. For $a\in\Sigma$ and a vertex $(q,j)$, it returns $a_{end}$ if $q\in\{q_{accept},q_{reject}\}$; otherwise, it returns $a'$ if $j=i$ and $(q,a)\rightarrow(q',a',Move)$ is a transition of $T$, and $a$ if $j\ne i$. All other arguments leave the memory unchanged.
        \item the next-move function $\alpha_N : M \times V_i \rightarrow V$ outputs $v_{end}$ at every vertex $(q,i)$ with $q\in\{q_{accept},q_{reject}\}$. Otherwise, for $a\in\Sigma$, it outputs $(q',j)$ where $(q,a)\rightarrow(q',a',Move)$ is the transition of $T$, with $j=\min(n,i+1)$ if $Move=R$ and $j=\max(1,i-1)$ if $Move=L$. For $a\in\{a_{init},a_{end}\}$ at a nonhalting vertex, it outputs $(q_1,1)$, a legal successor. We assume, w.l.o.g., that $T$ has a transition on every symbol in each nonhalting state.
    \end{itemize}
    \begin{figure}
        \centering
        \begin{tikzpicture}[x=0.75pt,y=0.75pt,yscale=-1,scale=.85,every node/.style={scale=.85}]
            \draw (143.29,138.11) -- (176.84,138.11);
            \draw [shift={(179.84,138.11)}, rotate=180,fill={rgb,255:red,0;green,0;blue,0},line width=0.08,draw opacity=0] (5.36,-2.57) -- (0,0) -- (5.36,2.57) -- (3.56,0) -- cycle;
            \draw (127.41,138.11) .. controls (127.41,133.48) and (130.96,129.73) .. (135.35,129.73) .. controls (139.74,129.73) and (143.29,133.48) .. (143.29,138.11) .. controls (143.29,142.73) and (139.74,146.48) .. (135.35,146.48) .. controls (130.96,146.48) and (127.41,142.73) .. (127.41,138.11) -- cycle;
            \draw (179.84,138.11) .. controls (179.84,133.48) and (183.4,129.73) .. (187.79,129.73) .. controls (192.17,129.73) and (195.73,133.48) .. (195.73,138.11) .. controls (195.73,142.73) and (192.17,146.48) .. (187.79,146.48) .. controls (183.4,146.48) and (179.84,142.73) .. (179.84,138.11) -- cycle;
            \draw (269.18,138.52) .. controls (269.18,133.9) and (272.74,130.15) .. (277.13,130.15) .. controls (281.51,130.15) and (285.07,133.9) .. (285.07,138.52) .. controls (285.07,143.15) and (281.51,146.9) .. (277.13,146.9) .. controls (272.74,146.9) and (269.18,143.15) .. (269.18,138.52) -- cycle;
            \draw (271.39,138.52) .. controls (271.39,135.19) and (273.96,132.48) .. (277.13,132.48) .. controls (280.29,132.48) and (282.86,135.19) .. (282.86,138.52) .. controls (282.86,141.86) and (280.29,144.56) .. (277.13,144.56) .. controls (273.96,144.56) and (271.39,141.86) .. (271.39,138.52) -- cycle;
            \draw (195.73,138.11) -- (266.18,138.5);
            \draw [shift={(269.18,138.52)}, rotate=180.33,fill={rgb,255:red,0;green,0;blue,0},line width=0.08,draw opacity=0] (5.36,-2.57) -- (0,0) -- (5.36,2.57) -- (3.56,0) -- cycle;
            \draw (184.39,129.77) .. controls (181.46,108.48) and (195.74,108.28) .. (192.56,127.91);
            \draw [shift={(191.99,130.81)}, rotate=282.86,fill={rgb,255:red,0;green,0;blue,0},line width=0.08,draw opacity=0] (5.36,-2.57) -- (0,0) -- (5.36,2.57) -- (3.56,0) -- cycle;
            \draw (274.12,130.19) .. controls (271.19,108.89) and (285.48,108.7) .. (282.29,128.32);
            \draw [shift={(281.72,131.23)}, rotate=282.86,fill={rgb,255:red,0;green,0;blue,0},line width=0.08,draw opacity=0] (5.36,-2.57) -- (0,0) -- (5.36,2.57) -- (3.56,0) -- cycle;
            \draw (109.26,138.48) -- (124.41,138.17);
            \draw [shift={(127.41,138.11)}, rotate=178.81,fill={rgb,255:red,0;green,0;blue,0},line width=0.08,draw opacity=0] (5.36,-2.57) -- (0,0) -- (5.36,2.57) -- (3.56,0) -- cycle;
            \draw (375.29,128.61) -- (426.84,128.61);
            \draw [shift={(429.84,128.61)}, rotate=180,fill={rgb,255:red,0;green,0;blue,0},line width=0.08,draw opacity=0] (5.36,-2.57) -- (0,0) -- (5.36,2.57) -- (3.56,0) -- cycle;
            \draw (359.41,128.61) .. controls (359.41,123.98) and (362.96,120.23) .. (367.35,120.23) .. controls (371.74,120.23) and (375.29,123.98) .. (375.29,128.61) .. controls (375.29,133.23) and (371.74,136.98) .. (367.35,136.98) .. controls (362.96,136.98) and (359.41,133.23) .. (359.41,128.61) -- cycle;
            \draw (429.84,128.61) .. controls (429.84,123.98) and (433.4,120.23) .. (437.79,120.23) .. controls (442.17,120.23) and (445.73,123.98) .. (445.73,128.61) .. controls (445.73,133.23) and (442.17,136.98) .. (437.79,136.98) .. controls (433.4,136.98) and (429.84,133.23) .. (429.84,128.61) -- cycle;
            \draw (432.06,128.61) .. controls (432.06,125.27) and (434.62,122.56) .. (437.79,122.56) .. controls (440.95,122.56) and (443.52,125.27) .. (443.52,128.61) .. controls (443.52,131.94) and (440.95,134.65) .. (437.79,134.65) .. controls (434.62,134.65) and (432.06,131.94) .. (432.06,128.61) -- cycle;
            \draw (434.39,120.27) .. controls (431.46,98.98) and (445.74,98.78) .. (442.56,118.41);
            \draw [shift={(441.99,121.31)}, rotate=282.86,fill={rgb,255:red,0;green,0;blue,0},line width=0.08,draw opacity=0] (5.36,-2.57) -- (0,0) -- (5.36,2.57) -- (3.56,0) -- cycle;
            \draw (341.26,128.98) -- (356.41,128.67);
            \draw [shift={(359.41,128.61)}, rotate=178.81,fill={rgb,255:red,0;green,0;blue,0},line width=0.08,draw opacity=0] (5.36,-2.57) -- (0,0) -- (5.36,2.57) -- (3.56,0) -- cycle;
            \draw (442.18,153.52) .. controls (442.18,148.9) and (445.74,145.15) .. (450.13,145.15) .. controls (454.51,145.15) and (458.07,148.9) .. (458.07,153.52) .. controls (458.07,158.15) and (454.51,161.9) .. (450.13,161.9) .. controls (445.74,161.9) and (442.18,158.15) .. (442.18,153.52) -- cycle;
            \draw (444.39,153.52) .. controls (444.39,150.19) and (446.96,147.48) .. (450.13,147.48) .. controls (453.29,147.48) and (455.86,150.19) .. (455.86,153.52) .. controls (455.86,156.86) and (453.29,159.56) .. (450.13,159.56) .. controls (446.96,159.56) and (444.39,156.86) .. (444.39,153.52) -- cycle;
            \draw (457.62,150.19) .. controls (475.79,145.34) and (477.06,162.46) .. (459.28,157.91);
            \draw [shift={(456.64,157.12)}, rotate=18.85,fill={rgb,255:red,0;green,0;blue,0},line width=0.08,draw opacity=0] (5.36,-2.57) -- (0,0) -- (5.36,2.57) -- (3.56,0) -- cycle;
            \draw (373.85,132.98) -- (440.32,147.9);
            \draw [shift={(443.25,148.56)}, rotate=192.65,fill={rgb,255:red,0;green,0;blue,0},line width=0.08,draw opacity=0] (5.36,-2.57) -- (0,0) -- (5.36,2.57) -- (3.56,0) -- cycle;
            \draw (148.2,126.23) node [anchor=north west,inner sep=0.75pt,font=\footnotesize,align=left] {$v_{init}$};
            \draw (209.78,126.4) node [anchor=north west,inner sep=0.75pt,font=\footnotesize,align=left] {$q_{accept},*$};
            \draw (145.16,98.15) node [anchor=north west,inner sep=0.75pt,font=\small,align=left] {$v$, for $v\neq (q_{accept},j)$};
            \draw (388.8,116.73) node [anchor=north west,inner sep=0.75pt,font=\footnotesize,align=left] {$v_{init}$};
            \draw (368.23,132.93) node [anchor=north west,inner sep=0.75pt,font=\footnotesize,rotate=-12.8,align=left] {$v$, for $v\neq v_{init}$};
            \draw (327.39,94) node [anchor=north west,inner sep=0.75pt,align=left] {$L^c$};
            \draw (99.89,104.5) node [anchor=north west,inner sep=0.75pt,align=left] {$L$};
            \draw (274.16,97.28) node [anchor=north west,inner sep=0.75pt,font=\small,align=left] {$V$};
            \draw (473.66,147.48) node [anchor=north west,inner sep=0.75pt,font=\footnotesize,align=left] {$V$};
            \draw (394.56,90.15) node [anchor=north west,inner sep=0.75pt,font=\small,align=left] {$v$, for $v \neq (q_{accept},j)$};
        \end{tikzpicture}
        \caption{The \DBWs{} accepting the languages $L$ and $L^c$ used to define the preference relation $\leqRelationStrict[n+1]$ in the game of \cref{fig:reduction-nash-checking-game}.}
        \label{fig:reduction-nash-checking-pref-n+1}
    \end{figure}

    By the construction given above, we know that $(\machine{i})_i$ produces the outcome $v_{init} \#^\omega$. As player~$n+1$ controls $v_{init}$, no deviation of player~$i \in \{1,\dots,n\}$ modifies this outcome, i.e., these deviations are not profitable. However, player~$n+1$ has exactly one possible deviation, to go from $v_{init}$ to $(q_1,1)$, thus leading to some play $v_{init} \pi$. This deviation is profitable if $v_{init} \#^\omega \leqRelationStrict[n+1] v_{init} \pi$, and we know by construction that $v_{init} \#^\omega \leqRelationStrict[n+1] v_{init} \pi$ if and only if $\pi$ visits the state $q_{accept}$, i.e., $T$ accepts $w$ (as $\pi$ represents an execution of the LBTM $T$ thanks to the definition of $\machine{1},\dots,\machine{n}$). This shows the correctness of the reduction. Moreover, this is a polynomial reduction, and $L^c\times L$ is irreflexive and transitive. Since \pspace{} is closed under complementation, NE checking is \pspace{}-hard, even for strict partial orders.
\end{proof}

\section{Correctness of the \texorpdfstring{\pcp{}}{P1CP2} Game}
\label{app:PCPgame}

\correspondence*
\addvspace{\topsep}

Before proving this result, we study in more details the correspondence between strategies in $\game$ and strategies in $\pcp(\game)$. We first note that if a strategy $\tau_{\proveroneTight}$ of \proverone{} is observation-based, we get that for all plays $\pi, \pi'$ starting at $s_0$ and consistent with $\tau_{\proveroneTight}$,
\begin{align}
\label{eq:obsBased}
\projOne{\pi} = \projOne{\pi'} \text{ and } \projE{\pi} = \projE{\pi'}.
\end{align}
Given a play $\rho$ of $\game$ starting at $v_0$, we call \emph{simulation of $\rho$} any strategy $\tau_{\proveroneTight}$ of \proverone{} such that $\rho = \projOne{\pi}$ for all plays $\pi$ starting at $s_0$ and consistent with $\tau_{\proveroneTight}$.

\begin{lemma}\label{lem:existence-simulation-prover-one}
  Let $\rho$ be a play of $\game$ starting at $v_0$. There always exists a simulation $\tau_{\proveroneTight}$ of $\rho$ that is observation-based.
\end{lemma}

\begin{proof}
    Consider a play $\rho = \rho_0\rho_1 \ldots$ of $\game$ starting at $v_0$. The required strategy $\tau_{\proveroneTight}$ is constructed as follows. For all histories $hs$ such that $s \in S_{\proveroneTight}$, we define $\tau_{\proveroneTight}(hs) = (\rho_{k},\rho_{k+1})$ if $\projOne{hs} = \rho_0\dots\rho_{k}$ is a prefix of $\rho$. Otherwise, we define $\tau_{\proveroneTight}(hs)$ as an arbitrary action $e \in E$ such that $\tau_{\proveroneTight}(h's') = e$ for all $h's'$ with $Obs(h's') = Obs(hs)$. By construction, $\tau_{\proveroneTight}$ is observation-based because whenever $Obs(h's') = Obs(hs)$, then $\projOne{h's'} = \projOne{hs}$.
    Moreover, for any play $\pi$ starting at $s_0$ and consistent with $\tau_{\proveroneTight}$, we have $\rho = \projOne{\pi}$.
\end{proof}

Let $\tau_{\proveroneTight}$ be an observation-based strategy of \proverone{} and $\rho$ be the play of $\game$ of which it is the simulation.
\begin{itemize}
    \item We say that a strategy $\tau_{\challenger}$ of \challenger{} is \emph{$\tau_{\proveroneTight}$-accepting} if for all histories $hs$ starting at $s_0$ and consistent with $\tau_{\proveroneTight}$, if $\projE{s} = (v,v')$, then $\tau_{\challenger}(hs) = v'$. In this case, the plays $\pi$ consistent with $\tau_{\proveroneTight}$ and $\tau_{\challenger}$ all satisfy $\projDev{\pi} = \bot$.

    \item Let $hs$ be a history starting at $s_0$ and consistent with $\tau_{\proveroneTight}$ such that $\projDev{h} = \bot$ and $\projDev{hs} = j$. Then, all plays $\pi$ consistent with $\tau_{\proveroneTight}$ having $hs$ as prefix, satisfy $\projDev{\pi} = j$ in addition to \eqref{eq:obsBased}. Thus, only the third component of the vertices of those $\pi$ can vary.\footnote{The fourth component equal to $(q_i)_i$ derive from the first and third components.} Therefore, when $\tau_{\proveroneTight}$ is fixed:
    \begin{itemize}
        \item Any strategy $\tau_{\challenger}$ is equivalent to a strategy $\sigma_j$ of player~$j$ such that $\tau_{\challenger}(h's') = \sigma_{j}(\projTwo{h's'})$ for all histories $h's'$ ending in $s' \in S_{\challenger}$ and having $hs$ as prefix. It is called a \emph{$\sigma_j$-deviation from $\tau_{\proveroneTight}$}.
        \item Any strategy $\tau_{\provertwoTight}$ is equivalent to a strategy profile $\sigma_{-j} = (\sigma_i)_{i \neq j}$ such that $\tau_{\provertwoTight}(h's') = \sigma_{i}(\projTwo{h's'})$ for all histories $h's'$ ending in $s' \in S_{\provertwoTight}$ with $\projOne{s} \in V_i$ and having $hs$ as prefix. It is called a \emph{$\sigma_{-j}$-punishment} strategy.
    \end{itemize}
\end{itemize}

\begin{proof}[Proof of \cref{theorem:correspondence}]
    Let us suppose that there exists an NE $\sigma = (\sigma_i)_{i\in\Players}$ from $v_0$ in $\game$ and let $\rho = \outcomefrom{\sigma}{v_0}$ be its outcome. By \cref{lem:existence-simulation-prover-one}, there exists an observation-based strategy $\tau_{\proveroneTight}$ of \proverone{} such that for any play $\pi$ starting at $s_0$ and consistent with $\tau_{\proveroneTight}$, we have $\projOne{\pi} = \rho$. Let $\tau_{\challenger}$ be a strategy of \challenger{}.
    \begin{itemize}
        \item If $\tau_{\challenger}$ is $\tau_{\proveroneTight}$-accepting, then we get $\projDev{\pi} = \bot$ for all plays $\pi$ starting at $s_0$ and consistent with both $\tau_{\proveroneTight}$ and $\tau_{\challenger}$, and
        any strategy $\tau_{\provertwoTight}$, i.e., $\outcomefrom{\tau_{\proveroneTight},\tau_{\challenger},\tau_{\provertwoTight}}{s_0} \in W_{acc}$.
        \item Otherwise, $\tau_{\challenger}$ is a $\sigma'_j$-deviation from $\tau_{\proveroneTight}$, for a strategy $\sigma'_j$ of some player~$j$. Consequently, we define $\tau_{\provertwoTight}$ from the strategy profile $\sigma_{-j} = (\sigma_i)_{i \neq j}$ such that $\tau_{\provertwoTight}$ is a $\sigma_{-j}$-punishment strategy. Therefore, for $\pi' = \outcomefrom{\tau_{\proveroneTight},\tau_{\challenger},\tau_{\provertwoTight}}{s_0}$, we get $\projTwo{\pi'} = \outcomefrom{\sigma'_j,\sigma_{-j}}{v_0}$. As $\sigma$ is an NE by hypothesis, we have $\projOne{\pi'} = \rho \not\leqRelationStrict[j] \projTwo{\pi'}$, meaning that $\pi' \in W_{dev}$.
    \end{itemize}

    We now suppose that in $\pcp(\game)$, there exists an observation-based strategy $\tau_{\proveroneTight}$ of \proverone{} such that for all strategies $\tau_{\challenger}$ of \challenger{}, there is a strategy $\tau_{\provertwoTight}$ of \provertwo{} such that $ \outcomefrom{\tau_{\proveroneTight},\tau_{\challenger},\tau_{\provertwoTight}}{s_0} \in W_{\proverone\provertwoTight}$. Let us explain how to deduce an NE $\sigma = (\sigma_i)_{i}$ from $v_0$ in $\game$. We choose a strategy $\tau_{\challenger}$ that is $\tau_{\proveroneTight}$-accepting, which gives by hypothesis a strategy $\tau_{\provertwoTight}$ of \provertwo. Observe that the resulting outcome $\pi = \outcomefrom{\tau_{\proveroneTight},\tau_{\challenger},\tau_{\provertwoTight}}{s_0}$ is such that $\projDev{\pi} = \bot$, so $\pi \in W_{acc}$. Let us define $\rho = \projOne{\pi}$, it will be the outcome of the NE $\sigma$ we want to construct.

    We now partially define the strategy profile $\sigma$ in a way to produce the outcome $\rho$. It remains to define $\sigma_{i}(hv)$ for all $i \in \Players$ and all histories $hv$ which are not prefix of $\rho$.

    Consider a strategy $\sigma'_j$ for player~$j$ such that there exists a history $hv$, with $v \in V_j$, prefix of $\rho$, but where $hv \cdot \sigma'_j(hv)$ is not. Then, we can consider $\tau_{\challenger}$ that is a $\sigma'_j$-deviation from $\tau_{\proveroneTight}$. By hypothesis, there exists a strategy $\tau_{\provertwoTight}$ such that $\pi' = \outcomefrom{\tau_{\proveroneTight},\tau_{\challenger},\tau_{\provertwoTight}}{s_0} \in W_{\proverone,\provertwoTight}$. Since $\projDev{\pi'} = j$, it must be the case that $\pi' \in W_{dev}$, i.e., $\projOne{\pi'} \not\leqRelationStrict[j] \projTwo{\pi'}$. Hence, we can complete the definition of $\sigma$ outside $\rho$ by $\sigma_{-j}$ thanks to the strategy $\tau_{\provertwoTight}$ seen as a $\sigma_{-j}$-punishment strategy. It follows that $\projTwo{\pi'} = \outcomefrom{\sigma'_j,\sigma_{-j}}{v_0}$. Therefore, by \eqref{eq:obsBased}, $\rho = \projOne{\pi} = \projOne{\pi'} \not\leqRelationStrict[j] \projTwo{\pi'} = \outcomefrom{\sigma'_j,\sigma_{-j}}{v_0}$, showing that $\outcomefrom{\sigma'_j,\sigma_{-j}}{v_0}$ is not a profitable deviation. We conclude that the constructed $\sigma$ is an NE.
\end{proof}

\section{NE Existence Problem}
\label{app:nash-existence-tools}

\NEexistenceGeneral*

\begin{proof}[Proof of \cref{theorem:NEexistenceGeneral}, membership]
    By \cref{theorem:correspondence}, deciding whether in $\game$, there exists an NE from $v_0$ reduces to deciding whether in $\pcp(\game)$, there exists an observation-based strategy $\tau_{\proveroneTight}$ of \proverone{} such that for all strategies $\tau_{\challenger}$ of \challenger{}, there is a strategy $\tau_{\provertwoTight}$ of \provertwo{} such that $\outcomefrom{\tau_{\proveroneTight},\tau_{\challenger},\tau_{\provertwoTight}}{s_0} \in W_{\proverone\provertwoTight}$. In~\cite{SPE-NCRS-quanti-parity-P1CP2}, the authors solve the problem they study by solving a similar three-player game with imperfect information. They proceed with the following three steps.
    \begin{enumerate}
        \item The winning condition is translated into a Rabin condition on the arena of the \pcp{} game.
        \item Then, the three-player game is transformed into a two-player zero-sum Rabin game with imperfect information.
        \item Finally, classical techniques to remove imperfect information are used to obtain a two-player zero-sum parity game with perfect information.
    \end{enumerate}

For our \pcp{} game, we already explain how to translate the winning condition $W_{\proverone\provertwo}$ into a Rabin condition (first step). With the notations of \cref{section:NE-existence}, this game has
\begin{itemize}
    \item a number of states in $\bigO{|V|^4 \cdot |\Players|^2 \cdot \Pi_{i\in \Players}|\aut{A}_i|}$,
    \item a number of actions in $\bigO{|V|^2}$
    \item a number of Rabin pairs in $\bigO{\Sigma_{i\in \Players}d_i}$
\end{itemize}

Let us explain the second step. To fit the context of~\cite{SPE-NCRS-quanti-parity-P1CP2}, we need to slightly modify our \pcp{} game in three ways:
\begin{itemize}
\item We add dummy vertices such that the three players play in a turn-based way, i.e., according to the turn sequence $(\proverone \challenger \provertwo)^\omega$.
\item Plays and histories include actions, i.e., they are sequences alternating vertices and actions, such that the histories always end with a vertex.
\item Due to the modification of plays and histories, the observation function $Obs$ is extended to actions, i.e., $Obs(a) = a$ is $a \in A_{\proverone}$ and $Obs(a) = \#$ is $a \in A_{\challenger} \cup A_{\provertwo}$ (\proverone{} only observes his own actions, the other actions are not visible\footnote{The situation is a little different in~\cite{SPE-NCRS-quanti-parity-P1CP2} as \proverone{} can also observe the actions of \provertwo{}.}).
\end{itemize}
With these modifications, we can translate our \pcp{} game into a two-player zero-sum Rabin game with imperfect information exactly as in~\cite{SPE-NCRS-quanti-parity-P1CP2} (this is derived from a
transformation introduced in~\cite{Chatterjee014}, specifically adapted to the \pcp{} game context). The main idea is to merge the two Provers into a new single Prover \prover. Imperfect information is used to ensure that this merging does not grant too much knowledge to \prover. In order to let the new Prover have as much actions available as \provertwo{} and stay observation-based, his action set includes, in addition to the actions of $A_{\proverone}$, all functions from $S_{\provertwo}$ to $A_{\provertwo}$. The set of vertices and the Rabin objective remain the same as in the \pcp{} game. This yields a new game, called \pc{} game,
\begin{itemize}
    \item with the same number of vertices and Rabin pairs,
    \item but with an exponential blowup in the number of actions.
\end{itemize}
Details about the new game can be found in~\cite{SPE-NCRS-quanti-parity-P1CP2}. In a way to guarantee the equivalence between both \pcp{} and \pc{} games, some technical properties must be satisfied by the observation function $Obs$ of the \pcp{} game:
\begin{itemize}
    \item The function $Obs$ must be strongly player-stable and thus player-stable (see Corollary 4.1.4 and Lemma 4.1.5 of~\cite{SPE-NCRS-quanti-parity-P1CP2}). This means that for any two histories $h = s_0a_0s_1a_1 \ldots a_{k-1}s_k = h_1s_k$ and $h'= s'_0a'_0s'_1a'_1 \ldots a'_{k-1}s'_k = h'_1s'_k$ such that $s_0 = s'_0$ is the initial vertex of the \pcp{} game and $Obs(h_1) = Obs(h'_1)$, it must follow that each pair of vertices $s_\ell, s'_\ell$, $1 \leq \ell \leq k$, are owned by the same player and $Obs(s_k) = Obs(s'_k)$. It is easy to verify that $s_\ell, s'_\ell$ are owned by the same player as the players play in a turn-based way. Let us explain why $Obs(s_k) = Obs(s_k)$. If $s_{k-1}$ is owned by \proverone{}, then $Obs(s_{k-1}) = Obs(s'_{k-1}) = v$ for some $v$ and $a_{k-1} = Obs(a_{k-1}) = Obs(a'_{k-1}) =a'_{k-1} = (v,v')$ for some $(v,v')$. Hence, $Obs(s_k) = Obs(s'_k) = (v,v')$. If $s_{k-1}$ is owned by \challenger{} or \provertwo{}, then $Obs(s_{k-1}) = Obs(s'_{k-1}) = (v,v')$ for some $(v,v')$ and thus $Obs(s_{k}) = Obs(s'_{k}) = v'$.
    \item The function $Obs$ must be action-stable (see Lemma 3.4.3 of~\cite{SPE-NCRS-quanti-parity-P1CP2}). This means that given $\Delta(s_1,a) = s_2$ and $\Delta(s'_1,a') = s'_2$ such that $Obs(s_1) = Obs(s'_1)$, then
    \begin{itemize}
        \item if $a = a'$, then $Obs(s_2) = Obs(s'_2)$;
        \item if $a,a'$ are visible (i.e., are actions of \proverone) and $Obs(s_2) = Obs(s'_2)$, then $a = a'$.
    \end{itemize}
It is easy to verify this property when $Obs(s_1) = Obs(s'_1) = v$ for some $v$, as $s_1, s'_1$ are then owned by \proverone{} and his actions $a,a'$ are of the form $(v,v')$ for some $v'$. This is also easy to verify when $Obs(s_1) = Obs(s'_1) = (v,v')$ for some $(v,v')$ as $s_1, s'_1$ are then owned by \challenger{} or \provertwo{} whose actions $a,a'$ are not visible.
\end{itemize}
With these technical properties and arguments similar to those of~\cite{SPE-NCRS-quanti-parity-P1CP2}, one can carefully verify that our \pcp{} game can be transformed into an equivalent two-player zero-sum Rabin game with imperfect information exactly as in~\cite{SPE-NCRS-quanti-parity-P1CP2}.

In the third step, the latter \pc{} game is transformed into an equivalent two-player zero-sum parity game with perfect information. The translation is exactly the same as in~\cite{SPE-NCRS-quanti-parity-P1CP2}. The idea to get rid of the imperfect information is to apply standard game-theoretic techniques~\cite{lpar/ChatterjeeD10,RaskinCDH07,Reif84}, by (1) making the Rabin condition visible, that is, such that any two similarly observed plays agree on the winning condition, and by (2) applying the subset construction to recall the set of possible visited vertices, and letting them be observed. Done carefully, this leads to a parity game with
\begin{itemize}
    \item a number of vertices exponential in the number of vertices and in the number of pairs of the \pc{} game (thus of the \pcp{} game),
    \item a parity condition with a number of priorities linear in the number of vertices and the number of pairs of the \pc{} game.
\end{itemize}

Finally, the constructed parity game can be solved in time $n^{\bigO{log(d)}}$ where $n$ is its number of vertices and $d$ its number of priorities~\cite{calude-quasi-poly-parity-game}. It follows that the initial NE existence problem can be decided in time exponential in $|V|$, $\Pi_{_{i\in \Players}}|\aut{A}_i|$, and $\Sigma_{_{i\in \Players}}d_i$.
\end{proof}

\subparagraph{Proof of \cref{theorem:NEexistenceGeneral}, hardness and particular case of one-player games.} We are going to prove that the NE existence problem is \pspaceComplete{} for one-player games. In this way, we get the hardness of \cref{theorem:NEexistenceGeneral}.

Before that, we need to study a deeply connected problem: the problem of deciding the existence of a maximal element for a preference relation $\leqRelationStrict$ (we also consider the existence of a minimal element).

\begin{proposition}
\label{prop:maximum-dpw-pspace}
    The problem of deciding whether an $\omega$-automatic relation $\leqRelationStrict$, accepted by a \DPW{} $\aut{A}$, has a maximal (resp.\ minimal) element is \pspaceComplete{}.
\end{proposition}

\begin{proof}
    We mainly focus on the existence of a maximal element. We briefly discuss the existence of a minimal element at the end of the proof, as the arguments are similar. We begin with the \pspace{}-membership, by studying the non-existence of a maximal element, i.e., $\forall x ~ \exists y, ~ x \leqRelationStrict y$. The algorithm constructs a nondeterministic parity automaton $\aut{B}$ by taking the projection of $\aut{A}$ to the first component of its labels. This automaton is of polynomial size and accepts the set $\{x \in \Sigma^\omega \mid \exists y \in \Sigma^\omega, x \leqRelationStrict y\}$. Then, one checks whether $\aut{B}$ is universal, with an algorithm in \pspace{}~\cite{handbook-of-model-checking-orna-kupferman}.

    \medskip

    Let us shift to the \pspace{}-hardness. We use a reduction from the non-universality problem of nondeterministic Büchi automata (\NBWs{}) which is \pspaceComplete{}~\cite{handbook-of-model-checking-orna-kupferman}. Given an \NBW{} $\aut{A} = (Q,\Sigma,q_0,\delta,F)$, we define a relation $\leqRelationStrict$ over $\Gamma^\omega$ for the alphabet $\Gamma = \Sigma \times Q \times \{0,1\}$. We will show that $\leqRelationStrict$ is $\omega$-automatic and has a maximal element if and only if $\lang{\aut{A}} \neq \Sigma^\omega$.

    For any $X \in \Gamma^\omega$, we denote its projections onto $\Sigma^\omega$, $Q^\omega$, and $\{0,1\}^\omega$ by $X_w$, $X_\rho$, and $X_b$ respectively, such that $X = (X_w,X_\rho,X_b)$. To simplify the notation, let $\mathsf{isAcc}(w, \rho)$ denote that $\rho \in Q^\omega$ encodes a valid, accepting run of the word $w \in \Sigma^\omega$ in $\aut{A}$, with $\rho_n$ the state reached after reading $w_0\cdots w_n$. Furthermore, let $b \prec_1 b'$ denote that $b' \in \{0,1\}^\omega$ is strictly preferred to $b \in \{0,1\}^\omega$ for the max-reward-reachability objective over the target $1$.

    For $X, Y \in \Gamma^\omega$, we define $X \leqRelationStrict Y$ if and only if one of the following holds:
    \begin{itemize}
        \item $X_w = Y_w$ and $\neg\mathsf{isAcc}(X_w, X_\rho)$ and $\mathsf{isAcc}(Y_w, Y_\rho)$, or
        \item $X_w = Y_w$ and $\mathsf{isAcc}(X_w, X_\rho)$ and $\mathsf{isAcc}(Y_w, Y_\rho)$ and $X_b \prec_1 Y_b$.
    \end{itemize}

    We now prove the correctness of the reduction, i.e., that $\leqRelationStrict$ has a maximal element if and only if $\lang{\aut{A}} \neq \Sigma^\omega$.
    \begin{itemize}
        \item Suppose $\lang{\aut{A}} \neq \Sigma^\omega$. There exists a word $w \in \Sigma^\omega \ssetminus \lang{\aut{A}}$. Let $X \in \Gamma^\omega$ be an element such that $X_w = w$. Because $w \notin \lang{\aut{A}}$, there is no valid accepting run for $w$, meaning $\mathsf{isAcc}(w, \rho)$ is false for all $\rho \in Q^\omega$. Thus, there exists no $Y \in \Gamma^\omega$ satisfying $\mathsf{isAcc}(Y_w, Y_\rho)$ given $Y_w = X_w$. Consequently, neither condition of $\leqRelationStrict$ can be met. Therefore, $X$ is a maximal element.
        
        \item Suppose $\lang{\aut{A}} = \Sigma^\omega$. Let $X = (X_w, X_\rho, X_b)$ be any element in $\Gamma^\omega$. Because $X_w \in \lang{\aut{A}}$, there exists an accepting run $\rho^* \in Q^\omega$ such that $\mathsf{isAcc}(X_w, \rho^*)$ holds. We can construct $Y$ depending on whether $X_\rho$ is an accepting run:
        \begin{itemize}
            \item If $\neg\mathsf{isAcc}(X_w, X_\rho)$, we define $Y = (X_w, \rho^*, X_b)$. By the first condition, $X \leqRelationStrict Y$.
            \item If $\mathsf{isAcc}(X_w, X_\rho)$, since the max-reward-reachability objective has no maximal element, there exists $b^* \in \{0,1\}^\omega$ such that $X_b \prec_1 b^*$.\footnote{If $X_b = 0^\omega$, we can take $b^* = 10^\omega$, otherwise take $b^* = 0 \cdot X_b$.} We define $Y = (X_w, X_\rho, b^*)$. By the second condition, $X \leqRelationStrict Y$.
        \end{itemize}
        In all cases, every element $X$ admits an element $Y$ such that $X \leqRelationStrict Y$, and hence no maximal element exists.
    \end{itemize}

    \begin{figure}
        \centering
        \begin{tikzpicture}[x=0.75pt,y=0.75pt,yscale=-1]
            \draw (247.07,114.79) .. controls (247.07,109.26) and (251.55,104.79) .. (257.07,104.79) .. controls (262.59,104.79) and (267.07,109.26) .. (267.07,114.79) .. controls (267.07,120.31) and (262.59,124.79) .. (257.07,124.79) .. controls (251.55,124.79) and (247.07,120.31) .. (247.07,114.79) -- cycle;
            \draw (311.27,107.59) -- (344.67,107.59);
            \draw [shift={(347.67,107.59)}, rotate = 180] [fill={rgb, 255:red, 0; green, 0; blue, 0},line width=0.08,draw opacity=0] (5.36,-2.57) -- (0,0) -- (5.36,2.57) -- (3.56,0) -- cycle;
            \draw (231.77,115.09) -- (244.07,114.85);
            \draw [shift={(247.07,114.79)}, rotate = 178.87] [fill={rgb, 255:red, 0; green, 0; blue, 0},line width=0.08,draw opacity=0] (5.36,-2.57) -- (0,0) -- (5.36,2.57) -- (3.56,0) -- cycle;
            \draw  [fill={rgb, 255:red, 0; green, 0; blue, 0 }  ,fill opacity=0.03 ][dash pattern={on 4.5pt off 4.5pt}] (238.79,99.99) .. controls (238.79,91.78) and (245.45,85.12) .. (253.66,85.12) -- (360.09,85.12) .. controls (368.3,85.12) and (374.96,91.78) .. (374.96,99.99) -- (374.96,151.73) .. controls (374.96,159.94) and (368.3,166.6) .. (360.09,166.6) -- (253.66,166.6) .. controls (245.45,166.6) and (238.79,159.94) .. (238.79,151.73) -- cycle;
            \draw (291.27,107.59) .. controls (291.27,102.06) and (295.75,97.59) .. (301.27,97.59) .. controls (306.79,97.59) and (311.27,102.06) .. (311.27,107.59) .. controls (311.27,113.11) and (306.79,117.59) .. (301.27,117.59) .. controls (295.75,117.59) and (291.27,113.11) .. (291.27,107.59) -- cycle;
            \draw (347.67,107.59) .. controls (347.67,102.06) and (352.15,97.59) .. (357.67,97.59) .. controls (363.19,97.59) and (367.67,102.06) .. (367.67,107.59) .. controls (367.67,113.11) and (363.19,117.59) .. (357.67,117.59) .. controls (352.15,117.59) and (347.67,113.11) .. (347.67,107.59) -- cycle;
            \draw (306.87,115.71) -- (339.28,133.56);
            \draw [shift={(341.91,135)}, rotate = 208.84,fill={rgb, 255:red, 0; green, 0; blue, 0 },line width=0.08,draw opacity=0] (5.36,-2.57) -- (0,0) -- (5.36,2.57) -- (3.56,0) -- cycle;
            \draw (342.91,135) .. controls (345.52,130.14) and (351.58,128.32) .. (356.45,130.93) .. controls (361.31,133.55) and (363.14,139.61) .. (360.52,144.47) .. controls (357.9,149.34) and (351.84,151.16) .. (346.98,148.55) .. controls (342.11,145.93) and (340.29,139.87) .. (342.91,135) -- cycle;
            \draw (320.34,90.3) node [anchor=north west,inner sep=0.75pt,font=\footnotesize,rotate=-0.5,align=left] {$a,q'$};
            \draw (250.07,109.4) node [anchor=north west,inner sep=0.75pt,align=left] {$q_{0}$};
            \draw (296.67,103) node [anchor=north west,inner sep=0.75pt,align=left] {$q$};
            \draw (351.47,99.8) node [anchor=north west,inner sep=0.75pt,align=left] {$q'$};
            \draw (343.87,131.2) node [anchor=north west,inner sep=0.75pt,align=left] {$q''$};
            \draw (312.89,122.96) node [anchor=north west,inner sep=0.75pt,font=\footnotesize,rotate=-28.85,align=left] {$a,q''$};
            \draw (380.34,95.1) node [anchor=north west,inner sep=0.75pt,font=\footnotesize,align=left] {if $q'\in \delta(q,a)$};
            \draw (380.34,133.5) node [anchor=north west,inner sep=0.75pt,font=\footnotesize,align=left] {if $q''\in \delta (q,a)$};
        \end{tikzpicture}
        \caption{A \DBW{} used for the \pspaceHard{}ness of \cref{prop:maximum-dpw-pspace}.}
        \label{fig:reduction-pspace-maximum-dpw}
    \end{figure}

    Finally, we show that $\leqRelationStrict$ is accepted by a \DPW{} of polynomial size. The equality $X_w = Y_w$ can be verified by a one-state \DBW{} $\aut{A}_=$ reading elements of $\Gamma \times \Gamma$ with the same $\Sigma$-component. The property $\mathsf{isAcc}(w,\rho)$ can be checked by a \DBW{} $\aut{A}_{w,\rho}$ of polynomial size, as depicted in \cref{fig:reduction-pspace-maximum-dpw}. Its state space is $Q$ and its set of accepting states is $F$. From a state $q$ (initially $q_0$), reading a pair $(a,q') \in \Sigma \times Q$ leads to $q'$ if $q' \in \delta(q,a)$, and to a rejecting sink otherwise. We can modify this automaton to accept $\neg\mathsf{isAcc}(w,\rho)$ by interpreting its acceptance condition as a co-Büchi condition, and we denote it $\aut{A}_{w,\rho}^c$, meaning that an accepting run visits the states in $F$ only finitely often or falls into the sink state. The relation $\prec_1$ for the max-reward-reachability objective is recognized by a constant-size \DBW{} $\aut{A}_{\prec_1}$ (as depicted in \cref{fig:dpw-for-classical-objectivers}).

    To combine these components, we first lift the alphabet of each automaton to $\Gamma \times \Gamma$. For example, $\aut{A}_{X_w,X_\rho}$ reads $(X,Y) \in \Gamma^\omega \times \Gamma^\omega$ and ignores the $Y$-component and the boolean component of $X$. We now construct a deterministic Rabin automaton (\DRW{}) $\aut{A}_{\leqRelationStrict}$. By distributing the conditions of $\leqRelationStrict$, its language is a disjunction of two conjunctions:
    \[L(\aut{A}_= \cap \aut{A}_{X_w,X_\rho}^c \cap \aut{A}_{Y_w,Y_\rho}) \cup L(\aut{A}_= \cap \aut{A}_{X_w,X_\rho} \cap \aut{A}_{Y_w,Y_\rho} \cap \aut{A}_{\prec_1})\]
    We build a synchronous product of these automata. For the first conjunction, the intersection forms one Rabin pair: a Büchi condition from $\aut{A}_= \cap \aut{A}_{Y_w,Y_\rho}$ (collapsed into a single Büchi condition~\cite{handbook-of-model-checking-orna-kupferman}) and a co-Büchi condition from $\aut{A}_{X_w,X_\rho}^c$. The second conjunction intersects multiple Büchi conditions, and thus a second Rabin pair with an empty co-Büchi set. This yields a \DRW{} of size $\bigO{|Q|^2}$ with exactly $d = 2$ pairs. It is known that a \DRW{} with $n$ states and $d$ pairs can be translated into an equivalent \DPW{} of size $\bigO{n \cdot 2^{d \log d}}$~\cite{phd-safra}. Since $d=2$ is a constant, this translation induces only a constant-factor blowup on the state space. We thus obtain a \DPW{} recognizing $\leqRelationStrict$ of polynomial size. The accepting and rejecting tests for the same certificate share their transition structure, so only two certificate trackers are needed.

    The relation $\leqRelationStrict$ is a strict partial order. Irreflexivity is immediate. If $X \leqRelationStrict Y$ and $Y \leqRelationStrict Z$, then $X_w=Y_w=Z_w$ and the certificates of $Y,Z$ are accepting. If the certificate of $X$ is not accepting, the first condition gives $X \leqRelationStrict Z$; otherwise, this follows from transitivity of $\prec_1$.

    Finally, the existence of a minimal element reduces to the existence of a maximal element by reversing the relation, i.e., exchanging the two tapes of its automaton. This preserves the size and the strict partial order property, and proves both bounds for minimal elements.
\end{proof}

\begin{figure}[t]
    \centering
    \begin{minipage}[c]{0.49\textwidth}
        \centering
        \begin{tikzpicture}[x=0.75pt,y=0.75pt,yscale=-1]
            \draw (247.69,82.91) -- (288.61,82.78) ;
            \draw [shift={(291.61,82.77)}, rotate = 179.82] [fill={rgb, 255:red, 0; green, 0; blue, 0 }  ][line width=0.08]  [draw opacity=0] (5.36,-2.57) -- (0,0) -- (5.36,2.57) -- (3.56,0) -- cycle;
            \draw (231.81,82.91) .. controls (231.81,78.28) and (235.36,74.53) .. (239.75,74.53) .. controls (244.14,74.53) and (247.69,78.28) .. (247.69,82.91) .. controls (247.69,87.53) and (244.14,91.28) .. (239.75,91.28) .. controls (235.36,91.28) and (231.81,87.53) .. (231.81,82.91) -- cycle;
            \draw (219.82,82.99) -- (228.81,82.93) ;
            \draw [shift={(231.81,82.91)}, rotate = 179.58] [fill={rgb, 255:red, 0; green, 0; blue, 0 }  ][line width=0.08]  [draw opacity=0] (5.36,-2.57) -- (0,0) -- (5.36,2.57) -- (3.56,0) -- cycle;
            \draw  [fill={rgb, 255:red, 0; green, 0; blue, 0 }  ,fill opacity=0.03 ][dash pattern={on 4.5pt off 4.5pt}] (282.39,78.58) .. controls (282.39,72.52) and (287.31,67.6) .. (293.37,67.6) -- (422.25,67.6) .. controls (428.32,67.6) and (433.23,72.52) .. (433.23,78.58) -- (433.23,116.78) .. controls (433.23,122.85) and (428.32,127.77) .. (422.25,127.77) -- (293.37,127.77) .. controls (287.31,127.77) and (282.39,122.85) .. (282.39,116.78) -- cycle ;
            \draw (291.33,78.73) .. controls (291.33,76.93) and (292.79,75.47) .. (294.58,75.47) -- (345.56,75.47) .. controls (347.36,75.47) and (348.82,76.93) .. (348.82,78.73) -- (348.82,88.5) .. controls (348.82,90.3) and (347.36,91.76) .. (345.56,91.76) -- (294.58,91.76) .. controls (292.79,91.76) and (291.33,90.3) .. (291.33,88.5) -- cycle ;
            \draw (294.13,109.93) .. controls (294.13,108.13) and (295.59,106.67) .. (297.38,106.67) -- (342.16,106.67) .. controls (343.96,106.67) and (345.42,108.13) .. (345.42,109.93) -- (345.42,119.7) .. controls (345.42,121.5) and (343.96,122.96) .. (342.16,122.96) -- (297.38,122.96) .. controls (295.59,122.96) and (294.13,121.5) .. (294.13,119.7) -- cycle ;
            \draw (378.33,109.93) .. controls (378.33,108.13) and (379.79,106.67) .. (381.58,106.67) -- (427.35,106.67) .. controls (429.15,106.67) and (430.61,108.13) .. (430.61,109.93) -- (430.61,119.7) .. controls (430.61,121.5) and (429.15,122.96) .. (427.35,122.96) -- (381.58,122.96) .. controls (379.79,122.96) and (378.33,121.5) .. (378.33,119.7) -- cycle ;
            \draw (345.42,114.99) -- (376.02,114.99) ;
            \draw [shift={(379.02,114.99)}, rotate = 180] [fill={rgb, 255:red, 0; green, 0; blue, 0 }  ][line width=0.08]  [draw opacity=0] (5.36,-2.57) -- (0,0) -- (5.36,2.57) -- (3.56,0) -- cycle    ;
            \draw   (232.01,147.51) .. controls (232.01,142.88) and (235.56,139.13) .. (239.95,139.13) .. controls (244.34,139.13) and (247.89,142.88) .. (247.89,147.51) .. controls (247.89,152.13) and (244.34,155.88) .. (239.95,155.88) .. controls (235.56,155.88) and (232.01,152.13) .. (232.01,147.51) -- cycle ;
            \draw    (235.87,154.94) .. controls (231.35,173.12) and (247.01,173.67) .. (243.96,157.56) ;
            \draw [shift={(243.31,154.88)}, rotate = 73.97] [fill={rgb, 255:red, 0; green, 0; blue, 0 }  ][line width=0.08]  [draw opacity=0] (5.36,-2.57) -- (0,0) -- (5.36,2.57) -- (3.56,0) -- cycle;
            \draw (323.62,167.63) .. controls (296,189.08) and (260.45,182.97) .. (235.83,156.08) ;
            \draw [shift={(233.96,153.97)}, rotate = 49.29] [fill={rgb, 255:red, 0; green, 0; blue, 0 }  ][line width=0.08]  [draw opacity=0] (5.36,-2.57) -- (0,0) -- (5.36,2.57) -- (3.56,0) -- cycle;
            \draw (239.75,91.28) -- (239.94,136.13) ;
            \draw [shift={(239.95,139.13)}, rotate = 269.76] [fill={rgb, 255:red, 0; green, 0; blue, 0 }  ][line width=0.08]  [draw opacity=0] (5.36,-2.57) -- (0,0) -- (5.36,2.57) -- (3.56,0) -- cycle;
            \draw (294.67,156.51) .. controls (294.67,151.88) and (298.23,148.13) .. (302.62,148.13) .. controls (307,148.13) and (310.56,151.88) .. (310.56,156.51) .. controls (310.56,161.13) and (307,164.88) .. (302.62,164.88) .. controls (298.23,164.88) and (294.67,161.13) .. (294.67,156.51) -- cycle ;
            \draw (294.96,122.3) -- (250,143.68) ;
            \draw [shift={(247.29,144.97)}, rotate = 334.57] [fill={rgb, 255:red, 0; green, 0; blue, 0 }  ][line width=0.08]  [draw opacity=0] (5.36,-2.57) -- (0,0) -- (5.36,2.57) -- (3.56,0) -- cycle;
            \draw [fill={rgb, 255:red, 0; green, 0; blue, 0 }  ,fill opacity=0.03 ][dash pattern={on 4.5pt off 4.5pt}] (280.99,151.47) .. controls (280.99,147.66) and (284.08,144.57) .. (287.89,144.57) -- (424.66,144.57) .. controls (428.47,144.57) and (431.56,147.66) .. (431.56,151.47) -- (431.56,175.47) .. controls (431.56,179.28) and (428.47,182.37) .. (424.66,182.37) -- (287.89,182.37) .. controls (284.08,182.37) and (280.99,179.28) .. (280.99,175.47) -- cycle ;
            \draw (320.34,161.17) .. controls (320.34,156.55) and (323.9,152.8) .. (328.28,152.8) .. controls (332.67,152.8) and (336.23,156.55) .. (336.23,161.17) .. controls (336.23,165.8) and (332.67,169.55) .. (328.28,169.55) .. controls (323.9,169.55) and (320.34,165.8) .. (320.34,161.17) -- cycle ;
            \draw (395.01,161.17) .. controls (395.01,156.55) and (398.56,152.8) .. (402.95,152.8) .. controls (407.34,152.8) and (410.89,156.55) .. (410.89,161.17) .. controls (410.89,165.8) and (407.34,169.55) .. (402.95,169.55) .. controls (398.56,169.55) and (395.01,165.8) .. (395.01,161.17) -- cycle;
            \draw (336.29,161.3) -- (392.01,161.18) ;
            \draw [shift={(395.01,161.17)}, rotate = 179.88] [fill={rgb, 255:red, 0; green, 0; blue, 0 }  ][line width=0.08]  [draw opacity=0] (5.36,-2.57) -- (0,0) -- (5.36,2.57) -- (3.56,0) -- cycle;
            \draw (245.42,89.61) -- (296.61,146.08) ;
            \draw [shift={(298.62,148.3)}, rotate = 227.8] [fill={rgb, 255:red, 0; green, 0; blue, 0 }  ][line width=0.08]  [draw opacity=0] (5.36,-2.57) -- (0,0) -- (5.36,2.57) -- (3.56,0) -- cycle;
            \draw    (306.86,123.37) -- (321.3,151.69) ;
            \draw [shift={(322.66,154.37)}, rotate = 242.99] [fill={rgb, 255:red, 0; green, 0; blue, 0 }  ][line width=0.08]  [draw opacity=0] (5.36,-2.57) -- (0,0) -- (5.36,2.57) -- (3.56,0) -- cycle;
            \draw (248.8,71.22) node [anchor=north west,inner sep=0.75pt,font=\footnotesize] [align=left] {$v_{0} ,v_{0}$};
            \draw (291.87,76.97) node [anchor=north west,inner sep=0.75pt,font=\footnotesize] [align=left] {$( q_{1} ,v_{0} ,v_{0})$};
            \draw (233,74.92) node [anchor=north west,inner sep=0.75pt,font=\footnotesize] [align=left] {$q'_{0}$};
            \draw (294.6,108.84) node [anchor=north west,inner sep=0.75pt,font=\footnotesize] [align=left] {$( q,v_{1} ,v_{2})$};
            \draw (377.3,108.11) node [anchor=north west,inner sep=0.75pt,font=\footnotesize] [align=left] {$( q',v'_{1} ,v'_{2})$};
            \draw (346.8,100.73) node [anchor=north west,inner sep=0.75pt,font=\scriptsize,rotate=-359.96] [align=left] {$v'_{1} ,v'_{2}$};
            \draw (234.2,142.52) node [anchor=north west,inner sep=0.75pt,font=\footnotesize] [align=left] {$q_{s}$};
            \draw (229.57,169.37) node [anchor=north west,inner sep=0.75pt,font=\footnotesize] [align=left] {$*,*$};
            \draw (202.03,104.57) node [anchor=north west,inner sep=0.75pt,font=\scriptsize] [align=left] {\begin{minipage}[lt]{24.59pt}\setlength\topsep{0pt}
            \begin{center}
            $v,*$ for\\$v\neq v_{0}$
            \end{center}
            \end{minipage}};
            \draw (242.21,135.87) node [anchor=north west,inner sep=0.75pt,font=\tiny,rotate=-333.82] [align=left] {$v',*$ for\\$(v_{1},v')\notin E$};
            \draw (255.85,87.13) node [anchor=north west,inner sep=0.75pt,font=\tiny,rotate=-47.86] [align=left] {$v_{0} ,w$ for\\$w\neq v_{0}$};
            \draw (294.86,152.53) node [anchor=north west,inner sep=0.75pt,font=\footnotesize,rotate=-359.96] [align=left] {$r_{v_{0}}$};
            \draw (321.53,157.53) node [anchor=north west,inner sep=0.75pt,font=\footnotesize,rotate=-359.96] [align=left] {$r_{v}$};
            \draw (395.2,157.52) node [anchor=north west,inner sep=0.75pt,font=\footnotesize] [align=left] {$r_{v'}$};
            \draw (337.47,146.88) node [anchor=north west,inner sep=0.75pt,font=\footnotesize] [align=left] {$v' ,*$ for\\$( v,v') \in E$};
            \draw (246.14,176.88) node [anchor=north west,inner sep=0.75pt,font=\scriptsize] [align=left] {$v' ,*$ for\\$( v,v') \notin E$};
            \draw (313.94,129.28) node [anchor=north west,inner sep=0.75pt,font=\scriptsize] [align=left] {$v,v'$ if $( v_{1} ,v) \in E$, $( v_{2} ,v') \notin E$};  
        \end{tikzpicture}
        \caption{The \DPW{} $\aut{A}'$ for the \pspace{} membership of \cref{theorem:NEexistenceGeneral} for one-player games.}
        \label{fig:maximum-ne-one-player-product}
    \end{minipage}
    \hfill
    \begin{minipage}[c]{0.49\textwidth}
        \centering
        \begin{tikzpicture}[x=0.75pt,y=0.75pt,yscale=-1]
            \draw (143.26,134.48) -- (158.41,134.17);
            \draw [shift={(161.41,134.11)},rotate=178.81,fill={rgb,255:red,0;green,0;blue,0},line width=0.08,draw opacity=0] (5.36,-2.57) -- (0,0) -- (5.36,2.57) -- (3.56,0) -- cycle;
            \draw (341.29,130.61) -- (378.01,130.33);
            \draw [shift={(381.01,130.31)},rotate=179.57,fill={rgb,255:red,0;green,0;blue,0},line width=0.08,draw opacity=0] (5.36,-2.57) -- (0,0) -- (5.36,2.57) -- (3.56,0) -- cycle;
            \draw (325.41,130.61) .. controls (325.41,125.98) and (328.96,122.23) .. (333.35,122.23) .. controls (337.74,122.23) and (341.29,125.98) .. (341.29,130.61) .. controls (341.29,135.23) and (337.74,138.98) .. (333.35,138.98) .. controls (328.96,138.98) and (325.41,135.23) .. (325.41,130.61) -- cycle;
            \draw (307.26,130.98) -- (322.41,130.67);
            \draw [shift={(325.41,130.61)},rotate=178.81,fill={rgb,255:red,0;green,0;blue,0},line width=0.08,draw opacity=0] (5.36,-2.57) -- (0,0) -- (5.36,2.57) -- (3.56,0) -- cycle;
            \draw [fill={rgb,255:red,0;green,0;blue,0},fill opacity=0.03,dash pattern={on 4.5pt off 4.5pt}] (194.99,117.62) .. controls (194.99,111.98) and (199.57,107.4) .. (205.21,107.4) -- (268.08,107.4) .. controls (273.72,107.4) and (278.3,111.98) .. (278.3,117.62) -- (278.3,153.18) .. controls (278.3,158.82) and (273.72,163.4) .. (268.08,163.4) -- (205.21,163.4) .. controls (199.57,163.4) and (194.99,158.82) .. (194.99,153.18) -- cycle;
            \draw (161.41,134.11) .. controls (161.41,129.72) and (164.96,126.17) .. (169.34,126.17) .. controls (173.72,126.17) and (177.27,129.72) .. (177.27,134.11) .. controls (177.27,138.49) and (173.72,142.04) .. (169.34,142.04) .. controls (164.96,142.04) and (161.41,138.49) .. (161.41,134.11) -- cycle;
            \draw (208.01,122.11) .. controls (208.01,117.72) and (211.56,114.17) .. (215.94,114.17) .. controls (220.32,114.17) and (223.87,117.72) .. (223.87,122.11) .. controls (223.87,126.49) and (220.32,130.04) .. (215.94,130.04) .. controls (211.56,130.04) and (208.01,126.49) .. (208.01,122.11) -- cycle;
            \draw [fill={rgb,255:red,0;green,0;blue,0},fill opacity=0.03,dash pattern={on 4.5pt off 4.5pt}] (372.79,116.63) .. controls (372.79,112.21) and (376.37,108.63) .. (380.79,108.63) -- (408.6,108.63) .. controls (413.02,108.63) and (416.6,112.21) .. (416.6,116.63) -- (416.6,148.49) .. controls (416.6,152.91) and (413.02,156.49) .. (408.6,156.49) -- (380.79,156.49) .. controls (376.37,156.49) and (372.79,152.91) .. (372.79,148.49) -- cycle;
            \draw (381.01,130.31) .. controls (381.01,125.92) and (384.56,122.37) .. (388.94,122.37) .. controls (393.32,122.37) and (396.87,125.92) .. (396.87,130.31) .. controls (396.87,134.69) and (393.32,138.24) .. (388.94,138.24) .. controls (384.56,138.24) and (381.01,134.69) .. (381.01,130.31) -- cycle;
            \draw (172.94,126.97) .. controls (179.17,117.27) and (187.27,117.69) .. (200.39,117.79);
            \draw [shift={(203.2,117.8)},rotate=180,fill={rgb,255:red,0;green,0;blue,0},line width=0.08,draw opacity=0] (5.36,-2.57) -- (0,0) -- (5.36,2.57) -- (3.56,0) -- cycle;
            \draw (172.14,141.77) .. controls (178.02,149.32) and (184.74,152.25) .. (198.81,152.57);
            \draw [shift={(201.6,152.6)},rotate=180,fill={rgb,255:red,0;green,0;blue,0},line width=0.08,draw opacity=0] (5.36,-2.57) -- (0,0) -- (5.36,2.57) -- (3.56,0) -- cycle;
            \draw (177.27,134.11) -- (200.2,134.19);
            \draw [shift={(203.2,134.2)},rotate=180.21,fill={rgb,255:red,0;green,0;blue,0},line width=0.08,draw opacity=0] (5.36,-2.57) -- (0,0) -- (5.36,2.57) -- (3.56,0) -- cycle;
            \draw (162.2,129.83) node [anchor=north west,inner sep=0.75pt,font=\footnotesize,align=left] {$v_{0}$};
            \draw (342,119.53) node [anchor=north west,inner sep=0.75pt,font=\footnotesize,align=left] {$v_{0},v_{0}$};
            \draw (302.59,100) node [anchor=north west,inner sep=0.75pt,align=left] {$\aut{A}'$};
            \draw (147.49,95.3) node [anchor=north west,inner sep=0.75pt,align=left] {$\game$};
            \draw (230.76,115.95) node [anchor=north west,inner sep=0.75pt,font=\small,align=left] {$\forall a \in \Sigma $};
            \draw (207.76,140.15) node [anchor=north west,inner sep=0.75pt,font=\small,align=left] {connected};
            \draw (211.2,119.23) node [anchor=north west,inner sep=0.75pt,font=\footnotesize,align=left] {$a$};
            \draw (396.7,134.62) node [anchor=north west,inner sep=0.75pt,align=left] {$\aut{A}$};
            \draw (382.6,125.63) node [anchor=north west,inner sep=0.75pt,font=\footnotesize,align=left] {$q_{0}$};
            \draw (326.6,122.63) node [anchor=north west,inner sep=0.75pt,font=\footnotesize,align=left] {$q_{0}'$};
        \end{tikzpicture}
        \caption{The reduction used for the \pspace{}-hardness of \cref{theorem:NEexistenceGeneral}.}
        \label{fig:reduction-nash-existence-pspace}
    \end{minipage}
\end{figure}

\begin{proof}[Proof of \cref{theorem:NEexistenceGeneral} for one-player games]
    Let us first present a \pspace{} algorithm for the membership. Intuitively, given a one-player game $\game = (A,\leqRelationStrict)$ and an initial vertex $v_0$, we will show that a play $\pi$ is an NE outcome from $v_0$ if and only if it is maximal for $\leqRelationStrict'$, a well chosen preference relation constructed from $\leqRelationStrict$. Hence, by \cref{prop:maximum-dpw-pspace}, we will get that the existence of an NE is in \pspace{}.

    We define the relation $\leqRelationStrict'$ as follows: for all $x, y \in V^\omega$, we have $x \leqRelationStrict' y$ if and only if either $x,y$ are both plays starting with $v_0$ such that $x \leqRelationStrict y$, or $x$ is not a play starting with $v_0$. Clearly, if $\pi$ is an NE outcome from $v_0$, then it is maximal for $\leqRelationStrict'$. Conversely, if $x$ is maximal for $\leqRelationStrict'$, then $x$ is a play starting with $v_0$ and it is thus an NE outcome as $x$ is maximal.

    It remains to show that $\leqRelationStrict'$ is accepted by a polynomial-size \DPW{} $\aut{A}'$ (see \cref{fig:maximum-ne-one-player-product}). Let $\delta$ be the transition function of the complete \DPW{} $\aut{A}$ accepting $\leqRelationStrict$. While both tapes are legal, states have the form $(q,v_1,v_2)$ and inherit the priority of $q$. The fresh initial state $q'_0$ reads $(v_0,v_0)$ and enters $(q_1,v_0,v_0)$ for $q_1 = \delta(q_0,(v_0,v_0))$. From $(q,v_1,v_2)$, reading $(v'_1,v'_2)$ with $(v_i,v'_i)\in E$ for $i=1,2$ leads to $(\delta(q,(v'_1,v'_2)),v'_1,v'_2)$. Whenever the first tape becomes illegal, the automaton enters an accepting sink $q_s$ of priority~$0$. If only the second tape becomes illegal, it enters a state $r_v$ of priority~$1$, where $v$ is the last vertex on the first tape. From $r_v$, reading $(v',v'')$ leads to $r_v$ if $(v,v')\in E$, and to $q_s$ otherwise. Initially, $(v,v')$ with $v\neq v_0$ leads to $q_s$, and $(v_0,v')$ with $v' \neq v_0$ leads to $r_{v_0}$. Thus, an illegal first tape is always accepted, whereas a legal first tape paired with an illegal second tape is rejected. The number of states is in $\bigO{|\aut{A}|\cdot |V|^2}$.

    \medskip

    We now prove that the NE existence problem is \pspaceHard{} for one-player games, with a reduction from the existence of a maximal element in an $\omega$-automatic relation, a \pspaceComplete{} problem by \cref{prop:maximum-dpw-pspace}. The reduction works as follows. Let $\leqRelationStrict$ be an $\omega$-automatic relation on $\Sigma$ accepted by a \DPW{} $\aut{A}$.
    We construct a one-player game $\game = (\arena,\leqRelationStrict')$ depicted in \cref{fig:reduction-nash-existence-pspace} as follows. The set of vertices of $A$ is $\Sigma' = \Sigma \cup \{v_0\}$ with a new initial vertex $v_0$; there is an edge between every pair of letters $a,b \in \Sigma$ and an edge between $v_0$ and each letter $a \in \Sigma$. The preference relation $\leqRelationStrict'$ is the one accepted by the \DPW{} $\aut{A}'$ of \cref{fig:reduction-nash-existence-pspace} where a fresh initial state $q'_0$ has a transition to the old initial state $q_0$ labeled by $(v_0,v_0)$. All states and transitions of $\aut{A}$ are retained, and its priority function is extended to $q'_0$ with priority~$1$. All other transitions involving $v_0$ are rejecting.

    It remains to show that the reduction is correct. Suppose that $\pi = v_0 x$ is an NE outcome starting in $v_0$ in $\game$. As it is a one-player game, $\pi$ is a maximal play with respect to $\leqRelationStrict'$ among all plays starting in $v_0$. As those plays all belong to $v_0\Sigma^\omega$, it follows that $x$ is a maximal element in $\leqRelationStrict$, by definition of $\leqRelationStrict'$. The other implication is proved similarly: given a maximal element $x$ in $\leqRelationStrict$, we get a play $\pi = v_0x$ that is maximal in $\leqRelationStrict'$ among all plays starting in $v_0$, thus an NE outcome.
\end{proof}

\section{Constrained NE Existence Problem}
\label{app:constrained-nash-existence}

\constrainedNEexistenceGeneral*

\begin{proof}[Sketch of proof of \cref{theorem:constrainedNEexistenceGeneral}, membership]
    From the given game $\game$ and constraints $\pi_i = \mu_i(\nu_i)^\omega$ for each player~$i$, we construct the same \pcp{} game as for the NE existence problem. The only difference is on the winning condition $W_{\proverone\provertwo} = W_{acc} \cup W_{dev}$ where $W_{acc}$ is modified in a way to take into account the constraints:
    \[
    W_{acc} = \{\pi \in \Plays(\pcp(\game)) \mid \projDev{\pi} = \bot \text{ and } \pi_i \leqRelationStrict[i] \projOne{\pi}, ~\forall i \}
    \]
    This modification has an impact on the Rabin condition encoding $W_{\proverone\provertwo}$. To translate $W_{acc}$ into a Rabin condition, we proceed as follows. For each $i$, we construct a \DPW{} $\aut{A}'_i$ accepting the set $\{x \in V^\omega \mid \mu_i(\nu_i)^\omega \leqRelationStrict[i] x \}$. It has $|\aut{A}_i|\cdot |\pi_i|$ states and $d_i$ priorities. We then construct the product $\aut{A}'$ of all those automata $\aut{A}'_i$ that accepts the set $\{x \in V^\omega \mid \mu_i(\nu_i)^\omega \leqRelationStrict[i] x, ~\forall i \}$. This is a generalized \DPW{} whose condition can be translated into a Streett condition with $d = \Sigma_{i \in \Players} d_i$ pairs~\cite{handbook-of-model-checking-orna-kupferman}. The latter automaton is equivalent to a deterministic automaton $\aut{B}$ with $\bigO{|\aut{A}'| \cdot 2^{d \log(d)}}$ states and $d$ Rabin pairs~\cite{phd-safra}. Finally, we replace the arena of the \pcp{} game by its product with the automaton $\aut{B}$. Thanks to the previous argument, $W_{acc}$ is encoded as a Rabin condition on the modified arena with $d$ pairs (without forgetting the condition $\projDev{\pi} = \bot$). As in the proof of the membership result of \cref{theorem:NEexistenceGeneral}, $W_{dev}$ is translated into a Rabin condition with $d$ pairs (step 1), and the rest of the proof (steps 2 and 3) is then the same. Nevertheless, the complexity of the constrained existence problem is different as the number of vertices of the \pcp{} game now also depends on $|\aut{B}|$. It follows that the constrained NE existence problem is exponential in $|V|$, $\Pi_{i\in \Players}|\aut{A}_i|$, $\Pi_{i\in \Players}|\pi_i|$, and doubly exponential in $\Sigma_{i\in \Players}d_i$.
\end{proof}

To obtain the PSPACE-hardness result of \Cref{theorem:constrainedNEexistenceGeneral}, it it enough to prove it for one-player games.

\begin{proof}[Proof of \cref{theorem:constrainedNEexistenceGeneral}, one-player games]
    Let us start with the \pspace{}-membership of the NE constrained problem for one-player games. Given a $\game = (\arena,\leqRelationStrict)$ with a single preference relation, an initial vertex $v_0$, and a constraint given by a lasso $\pi$, we have to decide the existence of an NE $\sigma$ from $v_0$ in $\game$ such that $\pi \leqRelationStrict \outcomefrom{\sigma}{v_0}$.

    First, we define a relation $\leqRelationStrict'$ from $\leqRelationStrict$ as we did for the membership proof of \cref{theorem:NEexistenceGeneral} for one-player games in \cref{app:nash-existence-tools}. Recall that a play $\pi$ is an NE outcome from $v_0$ if and only if it is maximal for $\leqRelationStrict'$. Recall also that $\leqRelationStrict'$ was accepted by a \DPW{} of size $\bigO{|\aut{A}| \cdot |V|^2}$. Second, we define a \DPW{} $\aut{B}$ with $\lang{\aut{B}} = \{\rho \in V^\omega \mid \pi \leqRelationStrict \rho\}$, of size $|\aut{A}| \cdot |\pi|$. Maximality for $\leqRelationStrict'$ already ensures that $\rho$ is a legal play from $v_0$.

    Our goal is then to find a maximal element of $\leqRelationStrict'$ in $\aut{B}$, i.e., to satisfy
    \[
    \exists x \forall y, ~ x \not\leqRelationStrict' y \wedge x \in \lang{\aut{B}}.
    \]
    As \pspace{} is closed under complementation, we can equivalently study the negation of this property:
    \begin{equation}
    \forall x \exists y, ~ x \leqRelationStrict' y \lor x \not\in \lang{\aut{B}}.
    \label{eq:streett}
    \end{equation}
    The set $\{(x,y) \mid x \leqRelationStrict' y \lor x \not\in \lang{\aut{B}}\}$ is accepted by a generalized \DPW{} with a disjunction of two parity conditions, hence by a deterministic Rabin automaton with polynomially many pairs. Its size is in $\bigO{|V|^2 \cdot |\aut{A}|^2 \cdot |\pi|}$. Projecting on the first component yields a nondeterministic Rabin automaton for $\{x \mid \exists y,~x \leqRelationStrict' y \lor x \not\in \lang{\aut{B}}\}$. This automaton can be converted into an \NBW{} of polynomial size by guessing a Rabin pair and a point after the last visit to its forbidden set. Universality of this \NBW{} decides~\eqref{eq:streett} in \pspace{}~\cite{handbook-of-model-checking-orna-kupferman}.

    \medskip

    Let us now prove the \pspaceHard{}ness. We use a reduction from the NE existence problem and suppose that there is only one player, as the NE existence problem is already \pspaceHard{} in this case. Given a game $\game = (\arena,\leqRelationStrict)$ with only player~$1$ with a preference relation $\leqRelationStrict$ accepted by a \DPW{} $\aut{A}$, and $v_0$ an initial vertex, we construct a new game $\game' = (\arena',\leqRelationStrict')$ as depicted in \cref{fig:reduction-pspace-hardness-nash-constrained}. Its arena $\arena'$ has $V' = V \cup \{v'_0,v'_1\}$ as set of vertices, with $v'_0$ and $v'_1$, two new vertices owned by player~$1$. It is a copy of $\arena$ with $v_0'$, the new initial vertex, having $v_0$ and $v'_1$ as successors, and $v'_1$ having itself as successor.
    Put $P=v'_0V^\omega$ and $d=v'_0(v'_1)^\omega$. We define
    \[
        \mathord{\leqRelationStrict'} = \{(v'_0x,v'_0y)\mid x\leqRelationStrict y,~x,y\in V^\omega\}
        \ \cup\ (\{d\}\times P).
    \]
    A polynomial-size \DPW{} $\aut{A}'$ recognizes this relation as follows (see \cref{fig:reduction-pspace-hardness-nash-constrained}). Its fresh initial state $q'_0$ reads $(v'_0,v'_0)$ and enters a fresh entry state $q_e$. From $q_e$, a letter $(v,w)\in V\times V$ leads to $\delta(q_0,(v,w))$ in a copy of $\aut{A}$, and a letter $(v'_1,v)$ with $v\in V$ leads to a state $q_s$ of priority~$0$ with loops on all such letters. The copy of $\aut{A}$ keeps its old priorities and transitions over $V\times V$; all unspecified transitions are rejecting. In particular, the transition to $q_s$ is available only from $q_e$, never from the old initial state $q_0$. If $\leqRelationStrict$ is a strict partial order, so is $\leqRelationStrict'$: the old order is lifted to $P$, the word $d$ lies below all of $P$, and other words are incomparable. In particular,
    \begin{align}\label{eq:lasso}
        v'_0(v'_1)^\omega \leqRelationStrict' v'_0x \text{ for all } x \in V^\omega.
    \end{align}
    We impose a constraint given by the lasso $v'_0(v_1')^\omega$.

    Let us show that there is an NE in $\game$ from $v_0$ if and only if there is one in $\game'$ from $v'_0$ whose outcome $\pi$ satisfies $v'_0(v_1')^\omega \leqRelationStrict' \pi$ (Recall that NEs are composed of a single strategy, as there is only one player). Let us first suppose that there exists an NE $\sigma$ in $\game$. Let $\sigma'$ be a strategy in $\game'$ defined as $\sigma'(v_0') = v_0$ and $\sigma'(v_0'h) = \sigma(h)$ for any history $h$. Clearly, by construction of $\leqRelationStrict'$, see \eqref{eq:lasso}, we have $v'_0(v_1')^\omega \leqRelationStrict' \outcomefrom{\sigma'}{v'_0}$. Let us explain why the profile $\sigma'$ is an NE. The deviation $v'_0(v'_1)^\omega$ is not profitable for player~$1$, since no pair in $\leqRelationStrict'$ has this word as its second component. Any other deviation is necessarily of the form $v'_0\pi$ with $\pi$ a deviation in $\game$ with respect to $\sigma$. Hence, it cannot be profitable as $\sigma$ is an NE and by definition of $\leqRelationStrict'$. Conversely, suppose that there exists an NE $\sigma'$ in $\game'$ such that $v'_0(v_1')^\omega \leqRelationStrict' \pi' = \outcomefrom{\sigma'}{v'_0}$. By the definition of $\leqRelationStrict'$, the constraint implies $\pi'\in P$. Therefore, we can define a strategy profile $\sigma $ in $\game$ such that $\sigma(h) = \sigma'(v_0'h)$ for any history $h$. This is an NE in $\game$ since $\sigma'$ is an NE in~$\game'$.

    This shows that the constrained NE existence problem is \pspace{}-hard for the single player case.
    \begin{figure}
        \centering
        \begin{tikzpicture}[x=0.75pt,y=0.75pt,yscale=-1]
            \draw (303.41,120.61) .. controls (303.41,115.98) and (306.96,112.23) .. (311.35,112.23) .. controls (315.74,112.23) and (319.29,115.98) .. (319.29,120.61) .. controls (319.29,125.23) and (315.74,128.98) .. (311.35,128.98) .. controls (306.96,128.98) and (303.41,125.23) .. (303.41,120.61) -- cycle ;
            \draw (285.26,120.98) -- (300.41,120.67) ;
            \draw [shift={(303.41,120.61)}, rotate = 178.81] [fill={rgb, 255:red, 0; green, 0; blue, 0},line width=0.08,draw opacity=0] (5.36,-2.57) -- (0,0) -- (5.36,2.57) -- (3.56,0) -- cycle;
            \draw  [fill={rgb, 255:red, 0; green, 0; blue, 0 }  ,fill opacity=0.03 ][dash pattern={on 4.5pt off 4.5pt}] (191.66,111.12) .. controls (191.66,106.59) and (195.33,102.91) .. (199.87,102.91) -- (228.43,102.91) .. controls (232.97,102.91) and (236.65,106.59) .. (236.65,111.12) -- (236.65,140.41) .. controls (236.65,144.95) and (232.97,148.63) .. (228.43,148.63) -- (199.87,148.63) .. controls (195.33,148.63) and (191.66,144.95) .. (191.66,140.41) -- cycle ;
            \draw (161.41,116.11) .. controls (161.41,111.72) and (164.96,108.17) .. (169.34,108.17) .. controls (173.72,108.17) and (177.27,111.72) .. (177.27,116.11) .. controls (177.27,120.49) and (173.72,124.04) .. (169.34,124.04) .. controls (164.96,124.04) and (161.41,120.49) .. (161.41,116.11) -- cycle ;
            \draw  [fill={rgb, 255:red, 0; green, 0; blue, 0 }  ,fill opacity=0.03 ][dash pattern={on 4.5pt off 4.5pt}] (400.08,124.79) .. controls (400.08,121.05) and (403.11,118.03) .. (406.84,118.03) -- (486.14,118.03) .. controls (489.87,118.03) and (492.9,121.05) .. (492.9,124.79) -- (492.9,148.31) .. controls (492.9,152.04) and (489.87,155.07) .. (486.14,155.07) -- (406.84,155.07) .. controls (403.11,155.07) and (400.08,152.04) .. (400.08,148.31) -- cycle ;
            \draw (177.27,116.11) -- (200.2,116.19) ;
            \draw [shift={(203.2,116.2)}, rotate = 180.21] [fill={rgb, 255:red, 0; green, 0; blue, 0},line width=0.08,draw opacity=0] (5.36,-2.57) -- (0,0) -- (5.36,2.57) -- (3.56,0) -- cycle;
            \draw (147.48,116.01) -- (158.41,116.08) ;
            \draw [shift={(161.41,116.11)}, rotate = 180.39] [fill={rgb, 255:red, 0; green, 0; blue, 0},line width=0.08,draw opacity=0] (5.36,-2.57) -- (0,0) -- (5.36,2.57) -- (3.56,0) -- cycle;
            \draw (319.29,120.61) -- (356.78,121) ;
            \draw [shift={(359.78,121.03)}, rotate = 180.6] [fill={rgb, 255:red, 0; green, 0; blue, 0},line width=0.08,draw opacity=0] (5.36,-2.57) -- (0,0) -- (5.36,2.57) -- (3.56,0) -- cycle;
            \draw (161.41,147.11) .. controls (161.41,142.72) and (164.96,139.17) .. (169.34,139.17) .. controls (173.72,139.17) and (177.27,142.72) .. (177.27,147.11) .. controls (177.27,151.49) and (173.72,155.04) .. (169.34,155.04) .. controls (164.96,155.04) and (161.41,151.49) .. (161.41,147.11) -- cycle ;
            \draw (169.34,124.04) -- (169.34,136.17) ;
            \draw [shift={(169.34,139.17)}, rotate = 270] [fill={rgb, 255:red, 0; green, 0; blue, 0},line width=0.08,draw opacity=0] (5.36,-2.57) -- (0,0) -- (5.36,2.57) -- (3.56,0) -- cycle;
            \draw (162.7,142.47) .. controls (143.6,139.03) and (141.65,153.47) .. (160.06,151.58) ;
            \draw [shift={(162.8,151.2)}, rotate = 170.23] [fill={rgb, 255:red, 0; green, 0; blue, 0},line width=0.08,draw opacity=0] (5.36,-2.57) -- (0,0) -- (5.36,2.57) -- (3.56,0) -- cycle;
            \draw (408.69,101.98) .. controls (408.69,97.35) and (412.25,93.6) .. (416.64,93.6) .. controls (421.02,93.6) and (424.58,97.35) .. (424.58,101.98) .. controls (424.58,106.6) and (421.02,110.35) .. (416.64,110.35) .. controls (412.25,110.35) and (408.69,106.6) .. (408.69,101.98) -- cycle ;
            \draw (423.57,97.8) .. controls (444.19,94.94) and (446.22,107.94) .. (426.59,106.79) ;
            \draw [shift={(423.67,106.53)}, rotate = 6.7] [fill={rgb, 255:red, 0; green, 0; blue, 0},line width=0.08,draw opacity=0] (5.36,-2.57) -- (0,0) -- (5.36,2.57) -- (3.56,0) -- cycle;
            \draw (203.07,116.77) .. controls (203.07,112.39) and (206.63,108.84) .. (211.01,108.84) .. controls (215.39,108.84) and (218.94,112.39) .. (218.94,116.77) .. controls (218.94,121.15) and (215.39,124.7) .. (211.01,124.7) .. controls (206.63,124.7) and (203.07,121.15) .. (203.07,116.77) -- cycle ;
            \draw (360.07,121.27) .. controls (360.07,116.65) and (363.63,112.9) .. (368.02,112.9) .. controls (372.4,112.9) and (375.96,116.65) .. (375.96,121.27) .. controls (375.96,125.9) and (372.4,129.65) .. (368.02,129.65) .. controls (363.63,129.65) and (360.07,125.9) .. (360.07,121.27) -- cycle ;
            \draw (375.96,124.27) -- (403.25,133.25) ;
            \draw [shift={(406.1,134.18)}, rotate = 198.2] [fill={rgb, 255:red, 0; green, 0; blue, 0},line width=0.08,draw opacity=0] (5.36,-2.57) -- (0,0) -- (5.36,2.57) -- (3.56,0) -- cycle;
            \draw (406.1,131.18) .. controls (406.1,129.47) and (407.49,128.07) .. (409.21,128.07) -- (471.37,128.07) .. controls (473.08,128.07) and (474.48,129.47) .. (474.48,131.18) -- (474.48,140.52) .. controls (474.48,142.24) and (473.08,143.63) .. (471.37,143.63) -- (409.21,143.63) .. controls (407.49,143.63) and (406.1,142.24) .. (406.1,140.52) -- cycle ;
            \draw (374.96,118.27) -- (406.27,107.19) ;
            \draw [shift={(409.1,106.18)}, rotate = 160.5] [fill={rgb, 255:red, 0; green, 0; blue, 0},line width=0.08,draw opacity=0] (5.36,-2.57) -- (0,0) -- (5.36,2.57) -- (3.56,0) -- cycle;

            \draw (161.8,108.82) node [anchor=north west,inner sep=0.75pt,font=\small,align=left] {$v'_{0}$};
            \draw (322.8,106.12) node [anchor=north west,inner sep=0.75pt,font=\footnotesize,align=left] {$v'_{0} ,v'_{0}$};
            \draw (272.09,89.57) node [anchor=north west,inner sep=0.75pt,align=left] {$\mathcal{A} '$};
            \draw (218.95,128.37) node [anchor=north west,inner sep=0.75pt,align=left] {$\mathcal{G}$};
            \draw (304.2,113.22) node [anchor=north west,inner sep=0.75pt,font=\small,align=left] {$q'_{0}$};
            \draw (161.8,139.82) node [anchor=north west,inner sep=0.75pt,font=\small,align=left] {$v'_{1}$};
            \draw (161.82,81.69) node [anchor=north west,inner sep=0.75pt,align=left] {$\mathcal{G} '$};
            \draw (410.07,98.16) node [anchor=north west,inner sep=0.75pt,font=\small,align=left] {$q_{s}$};
            \draw (442.13,94.32) node [anchor=north west,inner sep=0.75pt,font=\footnotesize,align=left] {$v'_{1} ,v$};
            \draw (372.33,103.31) node [anchor=north west,inner sep=0.75pt,font=\footnotesize,rotate=-341.2,align=left] {\begin{minipage}[lt]{20.14pt}\setlength\topsep{0pt}
            \begin{center}
            $v'_{1} ,v$
            \end{center}
            \end{minipage}};
            \draw (476.79,138.4) node [anchor=north west,inner sep=0.75pt,align=left] {$\mathcal{A}$};
            \draw (392.4,80.78) node [anchor=north west,inner sep=0.75pt,font=\scriptsize,align=left] {\begin{minipage}[lt]{32.37pt}\setlength\topsep{0pt}
            \begin{center}
            $(\forall v\in V)$
            \end{center}
            \end{minipage}};
            \draw (204.47,113.15) node [anchor=north west,inner sep=0.75pt,font=\small,align=left] {$v_{0}$};
            \draw (361.87,116.89) node [anchor=north west,inner sep=0.75pt,font=\small,align=left] {$q_{e}$};
            \draw (375.98,125.73) node [anchor=north west,inner sep=0.75pt,font=\footnotesize,rotate=-17.32,align=left] {$v,v'$};
            \draw (407.47,129.1) node [anchor=north west,inner sep=0.75pt,font=\scriptsize,align=left] {$\delta _{A}( q_{0} ,( v,v'))$};
            \draw (338.07,142.44) node [anchor=north west,inner sep=0.75pt,font=\scriptsize,align=left] {\begin{minipage}[lt]{42.28pt}\setlength\topsep{0pt}
            \begin{center}
            $(\forall v,v'\in V)$
            \end{center}
            \end{minipage}};
        \end{tikzpicture}
        \caption{The reduction used for the \pspaceHard{}ness of \cref{theorem:constrainedNEexistenceGeneral}.}
        \label{fig:reduction-pspace-hardness-nash-constrained}
    \end{figure}
\end{proof}

\section{Proof of \texorpdfstring{\cref{prop:PropertiesRelations}}
{Proposition~\ref{prop:PropertiesRelations}}}
\label{app:proof-PropertiesRelations}

For the proof of \cref{prop:PropertiesRelations}, we use logspace reductions to show the \nlHard{}ness results. Such a reduction consists in computing a polynomially bounded function $f$ with a deterministic Turing Machine using three tapes: a read-only input tape containing the input $x$ of length $n$, a write-only output tape that will contain $f(x)$ at the end of the execution, and a read-write work tape with $\log(n)$ cells. The reader can consult~\cite{Computational-Complexity-2009,sipser13} for more details about logspace reductions. Let us also recall that the complexity class \nl{} is closed under complementation~\cite{sipser13}. Before providing the proof of \cref{prop:PropertiesRelations}, we need the following result on generalized nondeterministic parity automata (\NPWs{}) with a positive Boolean combination of a constant number of parity conditions.

\begin{proposition}
\label{prop:emptiness-universality-dpw-generalized-nl-complete}
    The problem of deciding whether an $\omega$-regular language $L \subseteq \Sigma^\omega$ is not empty (resp.\ universal) is \nl{}-complete if $L$ is accepted by a generalized \NPW{} (resp.\ \DPW{}) with a constant number of parity conditions.
\end{proposition}

\begin{proof}
    Let us show the \nl{}-membership for non-emptiness. Consider a generalized \NPW{} $\aut{A}$ with a positive Boolean combination of $d$ parity conditions, with $d$ constant. Since parity conditions only handle infinite occurrences of states of a run, by~\cite[Proposition 3.1]{PatriciaBouyerBMU15}, we know that if there is a word $w$ accepted by $\aut{A}$, then there exists a lasso $\mu(\nu)^\omega$ accepted by $\aut{A}$ such that $|\mu|,|\nu| \leq |Q|^2$, where $|Q|$ is the number of states of $\aut{A}$. Hence, we can guess the length $n \leq 2|Q|^2$ of such a lasso, and the lasso itself on the fly, state by state. To check whether the guessed lasso is accepted by $\aut{A}$, we retain the cycle-start state, length counters, and the maximum priority occurring in $\nu$ for each of the $d$ parity conditions, in a way to verify whether the Boolean combination of those conditions is true or not. At any time, the amount of information to be stored is logarithmic since a state and
    $d$ priorities take logarithmic space when written in binary (recall that $d$ is constant).

    For the universality problem in case $\aut{A}$ is deterministic, we test emptiness of the complement automaton $\aut{A}'$. The latter automaton is constructed from $\aut{A}$ where in its Boolean combination, each $\vee$ (resp.\ $\wedge$) is replaced by $\wedge$ (resp.\ $\vee$), and the priorities are all incremented by one in each of its parity conditions. So, the size of $\aut{A}'$ is the same and the determinism is preserved. Hence, we can use the above algorithm and the equality $\mathsf{NL} = \mathsf{coNL}$.

    Since the non-emptiness for \NBWs{} and the universality problems for \DBWs{} are both \nlComplete{}~\cite{handbook-of-model-checking-orna-kupferman}, and a B\"uchi condition is a special case of a generalized parity condition, we get the \nlHard{}ness of every statement.
\end{proof}

We now proceed to the proof of \cref{prop:PropertiesRelations}. We denote by $\aut{A}$ a \DPW{} accepting an $\omega$-automatic relation $R \subseteq \Sigma^\omega \times \Sigma^\omega$.

\propertiesrelations*

\begin{proof}[Proof of \cref{prop:PropertiesRelations} - reflexivity]
    Let us start with the \nl{}-membership. From $\aut{A}$, we construct a new automaton $\aut{A}'$ over $\Sigma$ by first erasing all transitions labeled by $(a,b)$ with $a \neq b$, and then replacing each label $(a,a)$ of the remaining transitions by $a$. Clearly, $\aut{A}'$ is a \DPW{} and we get that $\{(x,x) \mid x \in \Sigma^\omega\} \subseteq \lang{\aut{A}}$ if and only if $\lang{\aut{A}'} = \Sigma^\omega$. Therefore, testing whether $R$ is reflexive reduces to checking the universality of $\aut{A}'$. The latter problem is \nl{} for \DPWs{} by \cref{prop:emptiness-universality-dpw-generalized-nl-complete}.
    Note that $\aut{A}'$ is constructed on the fly from $\aut{A}$ while guessing a lasso not accepted by $\aut{A}'$. This requires to use pointers on states and transitions of $\aut{A}$, those pointers being stored in a logarithmic space of the work tape.

    \medskip

    Let us continue with the \nl{}-hardness. We proceed by reduction from the universality problem for \DBWs{} which is \nlComplete{}~\cite{handbook-of-model-checking-orna-kupferman}. Given such an automaton $\aut{A}$ over $\Sigma$, we construct a new automaton $\aut{A}' = \aut{A} \times \aut{A}$ that accepts the relation $R = \lang{A} \times \lang{A}$ over $\Sigma \times \Sigma$. The automaton $\aut{A}'$ has a size $|\aut{A}|^2$ and is a generalized \DBW{} with a conjunction of two B\"uchi conditions, and can thus be transformed into a \DBW{} $\aut{B}$ of size $2|\aut{A}'|^2$~\cite{handbook-of-model-checking-orna-kupferman}. Moreover, we have that for all $x \in \Sigma^\omega$, $x \in \lang{\aut{A}}$ if and only if $(x,x) \in \lang{\aut{B}}$. That is, $\lang{\aut{A}} = \Sigma^\omega$ if and only if $R$ is reflexive. This establishes the correctness of the reduction. Finally, it is a logspace reduction. Indeed, we need to use pointers on the input tape, stored in the work tape, in a way to iterate on the states and transitions of $\aut{A}$, and construct step by step the states and transitions of $\aut{B}$ on the output tape; we also need to store an extra bit to remember the copy of $\aut{A}'$ (to get the \DBW{}~$\aut{B}$).
\end{proof}

\begin{proof}[Proof of \cref{prop:PropertiesRelations} - irreflexivity]
    For irreflexivity, note that $R$ is irreflexive if and only if $(\Sigma^\omega \times \Sigma^\omega) \ssetminus R$ is reflexive. So, as the complement of a \DPW{} is still a \DPW{} of the same size, we can repeat the same \nl{} algorithm as above with $(\Sigma^\omega \times \Sigma^\omega) \ssetminus R$, and a reduction from deciding the reflexivity of a relation to obtain the \nlHard{}ness.
\end{proof}

\begin{proof}[Proof of \cref{prop:PropertiesRelations} - transitivity]
    We begin with the \nl{}-membership. The relation $R$ is not transitive if there exist $x,y,z \in \Sigma^\omega$ such that $(x,y) \in R$ and $(y,z) \in R$ but $(x,z) \not\in R$. Testing whether $R$ is not transitive thus reduces to testing whether the language
    \[
    L = \{(x,y,z) \in (\Sigma^\omega)^3 \mid (x,y) \in R \wedge (y,z) \in R \wedge (x,z) \not\in R \}
    \]
    is not empty. We can construct an automaton $\aut{A}'$ accepting $L$ which is a generalized \DPW{} with a conjunction of three parity conditions. This automaton has a polynomial size (the complementation of $\aut{A}$ leads to an automaton with the same size as it suffices to increment each priority by $1$). Hence, we can check whether $L$ is not empty with the \nl{} algorithm of \cref{prop:emptiness-universality-dpw-generalized-nl-complete}. Note that $\aut{A}'$ is constructed on the fly while guessing a lasso accepted by $\aut{A}'$.

    \medskip

    We now prove the \nlHard{}ness. We use a reduction from the universality problem of \DBWs{} which is \nlComplete{}~\cite{handbook-of-model-checking-orna-kupferman}. Given a \DBW{} $\aut{A}$ over the alphabet $\Sigma$, we construct the relation $R = R_1 \cup R_2 \cup R_3 \cup R_4$ on $(\Sigma')^\omega \times (\Sigma')^\omega$ with $\Sigma' = \Sigma \cup \{\#\}$ such that:
    \begin{multicols}{2}
    \begin{itemize}
        \item $R_1 = \Sigma^\omega \times \{\#^\omega\}$,
        \item $R_2 = \{\#^\omega\} \times \Sigma^\omega$,
        \item $R_3 = \lang{\aut{A}} \times \lang{\aut{A}}$,
        \item $R_4 = \{(\#^\omega,\#^\omega)\}$.
    \end{itemize}
    \end{multicols}
    Let us show the correctness of the reduction. Suppose that $\lang{\aut{A}} = \Sigma^\omega$, it means that $R_3 = \Sigma^\omega \times \Sigma^\omega$. Let $x,y,z \in (\Sigma')^\omega$ such that $(x,y) \in R$ and $(y,z) \in R$, and let us show that $(x,z) \in R$. If $x = \#^\omega$, then $y \in \Sigma^\omega$ or $y = \# ^\omega$, so $(x,z) \in R_2 \cup R_4$. Otherwise $x \neq \#^\omega$, thus by $R_3 = \Sigma^\omega \times \Sigma^\omega$ and $R_1$, we have $(x,z) \in R_1 \cup R_3$.
    Conversely, if $\lang{\aut{A}} \neq \Sigma^\omega$, then there exists $x \in \Sigma^\omega \ssetminus \lang{\aut{A}}$. So, we have $(x,\#^\omega) \in R_1$, $(\#^\omega,x) \in R_2$ but $(x,x) \not\in R$. Hence, $\lang{\aut{A}} = \Sigma^\omega$ if and only if $R$ is transitive. Furthermore, $R$ is accepted by a \DBW{} $\aut{A}'$ that we construct as follows. Each $R_i$, $i\neq 3$, is trivially accepted by a \DBW{} of constant size. For $R_3$, we construct the automaton $\aut{A} \times \aut{A}$ which is a generalized \DBW{} with a conjunction of two B\"uchi conditions, that can be transformed into a \DBW{}. We then create a new initial state $q_0$ (replacing the initial state $q_0^i$ for each $R_i$) and duplicate the transitions outgoing\footnote{of the form $(a,\#)$ (resp.\ $(\#,a), (a,b), (\#,\#))$ for $R_1$ (resp.\ $R_2, R_3, R_4$), with $a,b \in \Sigma$.} each $q_0^i$ as a transition outgoing $q_0$. In this way we get a \DBW{} $\aut{A}'$ accepting $R$ of polynomial size. One can check that this reduction is a logspace reduction.
\end{proof}

\begin{proof}[Proof of \cref{prop:PropertiesRelations} - $\neg$-transitivity]
    Note that $R$ is $\neg$-transitive if and only if $(\Sigma^\omega \times \Sigma^\omega) \ssetminus R$ is transitive. Hence, as the complement of a \DPW{} is still a \DPW{} of the same size, we can repeat the same \nl{} algorithm as for transitivity with $(\Sigma^\omega \times \Sigma^\omega) \ssetminus R$, and a reduction from deciding the transitivity of a relation to obtain the \nlHard{}ness.
\end{proof}

\begin{proof}[Proof of \cref{prop:PropertiesRelations} - totality]
    We begin with the \nl{}-membership. Recall that $R$ is total if for all $x,y \in \Sigma^\omega$, we have $(x=y) \vee ((x,y) \in R) \vee ((y,x) \in R)$. Equivalently, $R$ is total if $R' = R_{=} \cup R \cup \tilde{R}$ is equal to $ \Sigma^\omega \times \Sigma^\omega$ such that $R_{=} = \{(x,x) \mid x \in \Sigma^\omega\}$ and $\tilde{R} = \{(y,x) \mid (x,y) \in R \}$. From $\aut{A}$, it is easy to construct a generalized \DPW{} $\aut{A}'$ of polynomial size with a disjunction of three parity conditions which accepts $R'$. We then run the \nl{} algorithm of \cref{prop:emptiness-universality-dpw-generalized-nl-complete} to check whether $\aut{A}'$ is universal while constructing $\aut{A}'$ on the fly.

    \medskip

    Let us turn to the \nl{}-hardness. We use a reduction from the universality problem of \DBWs{} which is \nlComplete{}~\cite{handbook-of-model-checking-orna-kupferman}. From a \DBW{} $\aut{A}$ over the alphabet $\Sigma$, we define the alphabet $\Sigma' = \Sigma \cup \{\#\}$ and the relation $R = R_1 \cup R_2$ where:
    \begin{itemize}
        \item $R_1 = (\Sigma^\omega \cup \{\#^\omega\}) \times \lang{\aut{A}}$ and
        \item $R_2 = L \times (\Sigma')^\omega$, with $L = (\Sigma')^\omega \ssetminus (\Sigma^\omega \cup \{\#^\omega\})$.
    \end{itemize}
    Let us prove the correctness of the reduction.
    Suppose that $\lang{\aut{A}} \neq \Sigma^\omega$. Given $x \in \Sigma^\omega \ssetminus \lang{\aut{A}}$, we have $x \neq \#^\omega$, $(\#^\omega,x) \not\in R$, and $(x,\#^\omega) \not\in R$. Hence, $R$ is not total.
    Suppose now that $R$ is not total, i.e., there exists $x \neq y$ such that $(x,y), (y,x) \not\in R$. By definition of $R_2$, we must have $x, y \in \Sigma^\omega \cup \{\#^\omega\}$. Since $x \neq y$, we know that $x$ or $y$ is in $\Sigma^\omega$, w.l.o.g., let us say $y \in \Sigma^\omega$. As $(x,y) \not\in R_1$, we get $y \not\in \lang{\aut{A}}$. Therefore $\lang{\aut{A}} = \Sigma^\omega$ if and only if $R$ is total.

    Moreover, $R$ is accepted by the \DBW{} $\aut{A}'$ depicted in \cref{fig:reduction-totality}. Let us give some explanations. W.l.o.g., we suppose that the automaton $\aut{A}$ is complete. The part of $\aut{A}'$ composed of the initial state $q_0$ and the two copies of $\aut{A}$ accepts the relation $R_1 = (\Sigma^\omega \times \lang{\aut{A}}) \cup (\{\#^\omega\} \times \lang{\aut{A}})$. The relation $R_2$ is accepted thanks to the accepting sink state $q_s$. Note that $L$ is composed of all words containing at least one symbol of $\Sigma$ and one symbol $\#$, a condition that has to be satisfied to reach $q_s$. The states $q_{\Sigma}$ and $q_{\#}$ retain the information that it is impossible to have $(x,y) \in R_1$ due to the symbol $\#$ seen inside $y$, implying $y \not\in \lang{\aut{A}}$. Finally, note that $\aut{A}'$ is deterministic, thus a \DBW{} and in particular a \DPW{}. This completes the proof as the proposed reduction is a logspace reduction.
\end{proof}
\begin{figure}
    \centering
    \begin{tikzpicture}[x=0.75pt,y=0.75pt,yscale=-1]
        \draw (370.97,117.93) .. controls (391.4,108.78) and (391.22,133.91) .. (373.88,126.87);
        \draw [shift={(371.32,125.65)},rotate=28.34,fill={rgb,255:red,0;green,0;blue,0},line width=0.08, draw opacity=0] (5.36,-2.57) -- (0,0) -- (5.36,2.57) -- (3.56,0) -- cycle;
        \draw (175.67,124.09) .. controls (175.67,118.56) and (180.15,114.09) .. (185.67,114.09) .. controls (191.19,114.09) and (195.67,118.56) .. (195.67,124.09) .. controls (195.67,129.61) and (191.19,134.09) .. (185.67,134.09) .. controls (180.15,134.09) and (175.67,129.61) .. (175.67,124.09) -- cycle;
        \draw (162.37,124.39) -- (172.67,124.15);
        \draw [shift={(175.67,124.09)},rotate=178.7,fill={rgb,255:red,0;green,0;blue,0},line width=0.08, draw opacity=0] (5.36,-2.57) -- (0,0) -- (5.36,2.57) -- (3.56,0) -- cycle;
        \draw (185.67,114.09) .. controls (191,96.11) and (221.27,97.32) .. (253.02,98.8);
        \draw [shift={(255.97,98.93)},rotate=182.65,fill={rgb,255:red,0;green,0;blue,0},line width=0.08, draw opacity=0] (5.36,-2.57) -- (0,0) -- (5.36,2.57) -- (3.56,0) -- cycle;
        \draw (185.67,134.09) .. controls (188.68,153.07) and (219.94,154.76) .. (252.64,154.76);
        \draw [shift={(255.16,154.75)},rotate=179.86,fill={rgb,255:red,0;green,0;blue,0},line width=0.08, draw opacity=0] (5.36,-2.57) -- (0,0) -- (5.36,2.57) -- (3.56,0) -- cycle;
        \draw (185.67,134.09) .. controls (197.1,180.12) and (224.48,182.06) .. (327.99,185.5);
        \draw [shift={(329.56,185.55)},rotate=181.9,fill={rgb,255:red,0;green,0;blue,0},line width=0.08, draw opacity=0] (5.36,-2.57) -- (0,0) -- (5.36,2.57) -- (3.56,0) -- cycle;
        \draw (185.67,114.09) .. controls (193.52,78.13) and (223.26,62.91) .. (329.66,61.8);
        \draw [shift={(331.27,61.79)},rotate=179.49,fill={rgb,255:red,0;green,0;blue,0},line width=0.08, draw opacity=0] (5.36,-2.57) -- (0,0) -- (5.36,2.57) -- (3.56,0) -- cycle;
        \draw (330.17,184.01) .. controls (330.17,178.49) and (334.64,174.01) .. (340.17,174.01) .. controls (345.69,174.01) and (350.17,178.49) .. (350.17,184.01) .. controls (350.17,189.54) and (345.69,194.01) .. (340.17,194.01) .. controls (334.64,194.01) and (330.17,189.54) .. (330.17,184.01) -- cycle;
        \draw (331.27,61.79) .. controls (331.27,56.26) and (335.75,51.79) .. (341.27,51.79) .. controls (346.79,51.79) and (351.27,56.26) .. (351.27,61.79) .. controls (351.27,67.31) and (346.79,71.79) .. (341.27,71.79) .. controls (335.75,71.79) and (331.27,67.31) .. (331.27,61.79) -- cycle;
        \draw (297.96,165.55) -- (327.97,177.06);
        \draw [shift={(330.77,178.13)},rotate=200.98,fill={rgb,255:red,0;green,0;blue,0},line width=0.08, draw opacity=0] (5.36,-2.57) -- (0,0) -- (5.36,2.57) -- (3.56,0) -- cycle;
        \draw (296.77,86.93) -- (330.09,69.73);
        \draw [shift={(332.76,68.35)},rotate=152.7,fill={rgb,255:red,0;green,0;blue,0},line width=0.08, draw opacity=0] (5.36,-2.57) -- (0,0) -- (5.36,2.57) -- (3.56,0) -- cycle;
        \draw (350.17,58.33) .. controls (376.46,51.82) and (374.94,69.63) .. (353.18,67.52);
        \draw [shift={(350.36,67.15)},rotate=9.31,fill={rgb,255:red,0;green,0;blue,0},line width=0.08, draw opacity=0] (5.36,-2.57) -- (0,0) -- (5.36,2.57) -- (3.56,0) -- cycle;
        \draw (349.16,189.15) .. controls (375.79,192.24) and (376.73,175.21) .. (351.6,179.41);
        \draw [shift={(348.76,179.95)},rotate=347.91,fill={rgb,255:red,0;green,0;blue,0},line width=0.08, draw opacity=0] (5.36,-2.57) -- (0,0) -- (5.36,2.57) -- (3.56,0) -- cycle;
        \draw [fill={rgb,255:red,0;green,0;blue,0},fill opacity=0.03,dash pattern={on 4.5pt off 4.5pt}] (255.96,83.18) .. controls (255.96,78.69) and (259.6,75.05) .. (264.08,75.05) -- (294.24,75.05) .. controls (298.72,75.05) and (302.36,78.69) .. (302.36,83.18) -- (302.36,111.43) .. controls (302.36,115.92) and (298.72,119.55) .. (294.24,119.55) -- (264.08,119.55) .. controls (259.6,119.55) and (255.96,115.92) .. (255.96,111.43) -- cycle;
        \draw [fill={rgb,255:red,0;green,0;blue,0},fill opacity=0.03,dash pattern={on 4.5pt off 4.5pt}] (255.96,138.18) .. controls (255.96,133.69) and (259.6,130.05) .. (264.08,130.05) -- (294.24,130.05) .. controls (298.72,130.05) and (302.36,133.69) .. (302.36,138.18) -- (302.36,166.43) .. controls (302.36,170.92) and (298.72,174.55) .. (294.24,174.55) -- (264.08,174.55) .. controls (259.6,174.55) and (255.96,170.92) .. (255.96,166.43) -- cycle;
        \draw (299.17,102.13) -- (343.7,116.19);
        \draw [shift={(346.56,117.09)},rotate=197.52,fill={rgb,255:red,0;green,0;blue,0},line width=0.08, draw opacity=0] (5.36,-2.57) -- (0,0) -- (5.36,2.57) -- (3.56,0) -- cycle;
        \draw (298.37,142.53) -- (345.28,129.12);
        \draw [shift={(348.16,128.29)},rotate=164.04,fill={rgb,255:red,0;green,0;blue,0},line width=0.08, draw opacity=0] (5.36,-2.57) -- (0,0) -- (5.36,2.57) -- (3.56,0) -- cycle;
        \draw (340.17,174.01) -- (353.9,137.5);
        \draw [shift={(354.96,134.69)},rotate=110.62,fill={rgb,255:red,0;green,0;blue,0},line width=0.08, draw opacity=0] (5.36,-2.57) -- (0,0) -- (5.36,2.57) -- (3.56,0) -- cycle;
        \draw (341.27,71.79) -- (354.31,106.68);
        \draw [shift={(355.36,109.49)},rotate=249.51,fill={rgb,255:red,0;green,0;blue,0},line width=0.08, draw opacity=0] (5.36,-2.57) -- (0,0) -- (5.36,2.57) -- (3.56,0) -- cycle;
        \draw (349.47,122.09) .. controls (349.47,116.9) and (353.68,112.69) .. (358.87,112.69) .. controls (364.06,112.69) and (368.26,116.9) .. (368.26,122.09) .. controls (368.26,127.28) and (364.06,131.49) .. (358.87,131.49) .. controls (353.68,131.49) and (349.47,127.28) .. (349.47,122.09) -- cycle;
        \draw (346.5,122.09) .. controls (346.5,115.26) and (352.04,109.72) .. (358.87,109.72) .. controls (365.7,109.72) and (371.23,115.26) .. (371.23,122.09) .. controls (371.23,128.92) and (365.7,134.46) .. (358.87,134.46) .. controls (352.04,134.46) and (346.5,128.92) .. (346.5,122.09) -- cycle;
        \draw (283.14,157.7) node [anchor=north west,inner sep=0.75pt, font=\small,align=left] {$\aut{A}$};
        \draw (264.77,137.3) node [anchor=north west,inner sep=0.75pt, font=\small,align=left] {$\#,\Sigma$};
        \draw (283.3,102.5) node [anchor=north west,inner sep=0.75pt, font=\small,align=left] {$\aut{A}$};
        \draw (318.84,95.45) node [anchor=north west,inner sep=0.75pt, font=\scriptsize,rotate=-17.52,align=left] {$\#,*$};
        \draw (308.58,126.32) node [anchor=north west,inner sep=0.75pt, font=\scriptsize,rotate=-342.72,align=left] {$\Sigma,*$};
        \draw (212.17,102.7) node [anchor=north west,inner sep=0.75pt, font=\small,align=left] {$\Sigma,\Sigma$};
        \draw (211.37,138.7) node [anchor=north west,inner sep=0.75pt, font=\small,align=left] {$\#,\Sigma $};
        \draw (248.58,184.05) node [anchor=north west,inner sep=0.75pt, font=\small,rotate=-4.16,align=left] {$\#,\#$};
        \draw (254.36,50.43) node [anchor=north west,inner sep=0.75pt, font=\small,rotate=-355.33,align=left] {$\Sigma,\#$};
        \draw (351.97,82.9) node [anchor=north west,inner sep=0.75pt, font=\scriptsize,align=left] {$\#,*$};
        \draw (350.37,149.9) node [anchor=north west,inner sep=0.75pt, font=\scriptsize,align=left] {$\Sigma,*$};
        \draw (370.77,55.5) node [anchor=north west,inner sep=0.75pt, font=\footnotesize,align=left] {$\Sigma,*$};
        \draw (374.37,176.3) node [anchor=north west,inner sep=0.75pt, font=\footnotesize,align=left] {$\#,*$};
        \draw (307.45,155.28) node [anchor=north west,inner sep=0.75pt, font=\scriptsize,rotate=-20.36,align=left] {$\#,\#$};
        \draw (306.77,83.68) node [anchor=north west,inner sep=0.75pt, font=\scriptsize,rotate=-331.55,align=left] {$\Sigma,\#$};
        \draw (333.57,58.5) node [anchor=north west,inner sep=0.75pt, font=\small,align=left] {$q_{\Sigma}$};
        \draw (178.77,119.9) node [anchor=north west,inner sep=0.75pt, font=\small,align=left] {$q_{0}$};
        \draw (331.77,179.1) node [anchor=north west,inner sep=0.75pt, font=\small,align=left] {$q_{\#}$};
        \draw (351.37,118.1) node [anchor=north west,inner sep=0.75pt, font=\small,align=left] {$q_{s}$};
        \draw (387.77,118.1) node [anchor=north west,inner sep=0.75pt, font=\small,align=left] {$*,*$};
        \draw (265.77,84.3) node [anchor=north west,inner sep=0.75pt, font=\small,align=left] {$\Sigma,\Sigma$};
    \end{tikzpicture}
    \caption{The \DBW{} $\aut{A}'$ used in the reduction for the totality property of \cref{prop:PropertiesRelations}.}
    \label{fig:reduction-totality}
\end{figure}

\section{Adaptation of Reductions for Preorders}
\label{app:reduction-preorders}

In this section, we present how to adapt the hardness proofs of \cref{theorem:nash-checking-pspace,theorem:OutcomeCheck,theorem:NEexistenceGeneral,theorem:constrainedNEexistenceGeneral} for $\omega$-automatic preorders (instead of $\omega$-automatic relations). For every proof, the general idea is to modify each \DPW{} used in the reductions so that it accepts a preorder $\leqRelation$ extending the relation $\leqRelationStrict$ initially accepted (i.e., in the sense that $x \leqRelationStrict y$ if and only if $x \leqRelation y$ and $y \not\leqRelation x$).

\subparagraph*{NE Checking Problem.} For the \pspaceHard{}ness of \cref{theorem:nash-checking-pspace} presented in \cref{app:hardness-nash-checking}, relations $\leqRelationStrict[i]$ are empty for $1 \leq i \leq n$, while $\mathord{\leqRelationStrict[n+1]} = L^c \times L$. The only small modification is to define $\mathord{\leqRelation[i]} = V^\omega \times V^\omega$, for $1 \leq i \leq n$, and $\mathord{\leqRelation[n+1]} = (L^c \times L) \cup (L^c \times L^c) \cup (L \times L)$. All these relations are now reflexive and transitive. However, some explanation is necessary for the transitivity of $\leqRelation[n+1]$. Let $x,y,z \in V^\omega$ be such that $x \leqRelation[n+1] y$ and $y \leqRelation[n+1] z$. If $y \in L^c$, we must have $x \in L^c$ so $x \leqRelation[n+1] z$. Otherwise, $y \in L$, then we must have $z \in L$ so we also have $x \leqRelation[n+1] z$. As $L$ and $L^c$ are \DBWs{}, the cartesian product of two \DBWs{} is a generalized \DBW{} with a conjunction of two Büchi conditions, thus a \DBW{}~\cite{handbook-of-model-checking-orna-kupferman}, and the union of two \DBWs{} is still a \DBW{}, we can conclude because a \DBW{} is in particular a \DPW{}.

\subparagraph*{NE Outcome Checking Problem.} For the $\mathsf{Parity}$-hardness of \cref{theorem:OutcomeCheck}, the empty relation $\leqRelationStrict[2]$ is transformed into $\mathord{\leqRelation[2]} = (V')^\omega \times (V')^\omega$ (recall that $V' = V \cup \{v_0'\}$ for $v_0'$ the new initial state). We then modify the \DPW{} $\aut{A}_1$ for $\leqRelationStrict[1]$ of \Cref{fig:outcome-checking-parity-reduction} by adding a new state $q_s$ with priority~$0$ and new transitions $(q_0,(v,v),q_s)$ for all $v \in V$, and $(q_s,(v,v),q_s)$ for all $v \in V'$. We also assign priority~$0$ to $q_0$, so that $((v'_0)^\omega,(v'_0)^\omega)$ is accepted. The modified \DPW{} accepts $\leqRelation[1]$, the reflexive closure of $\leqRelationStrict[1]$. One can easily check that it is transitive.

\subparagraph*{NE Existence Problem.} To show the \pspaceHard{}ness of \cref{theorem:NEexistenceGeneral} for preorders, we need to show the \pspaceHard{}ness of deciding whether a preorder has a maximal element, understood with respect to its strict part (see \Cref{app:nash-existence-tools}).

Let us consider the strict partial order $\leqRelationStrict$ from the hardness proof of \cref{prop:maximum-dpw-pspace}. We modify its definition to construct a relation $\leqRelation$ as follows: $X \leqRelation Y$ if one of the following holds:
\begin{itemize}
    \item $X_w = Y_w$ and $\neg\mathsf{isAcc}(X_w, X_\rho)$ and $\mathsf{isAcc}(Y_w, Y_\rho)$, or
    \item $X_w = Y_w$ and $\mathsf{isAcc}(X_w, X_\rho)$ and $\mathsf{isAcc}(Y_w, Y_\rho)$ and $X_b \preceq_1 Y_b$, or
    \item $X = Y$.
\end{itemize}

The two differences are the equality test $X = Y$ to ensure reflexivity, and the use of $\preceq_1$ instead of $\prec_1$ for max-reward-reachability over the target $1$. The \DPW{} for $\preceq_1$ can be adapted from the one depicted in~\cref{fig:dpw-for-classical-objectivers}: in addition to strictly preferred words, we additionally need to accept equivalent words, i.e., pairs $(w,w')$ that both visit the target set $T$ at the same time, or both never visit it. To do so, the initial state is assigned a new priority of $2$ and leads to a sink state of priority $2$ labeled $(T,T)$. The deterministic Rabin automaton for $\leqRelation$ is constructed similarly to the one for $\leqRelationStrict$, with one additional Rabin pair to capture the equality $X = Y$.

Let us show that $\leqRelation$ is a preorder. It is clearly reflexive by the third condition; we now prove transitivity. Let $X,Y,Z \in \Gamma^\omega$ such that $X \leqRelation Y$ and $Y \leqRelation Z$. If $X = Y$ or $Y = Z$, we trivially have $X \leqRelation Z$. Otherwise, we have $X_w = Y_w = Z_w$. Let $w$ be this word. We know by hypothesis that $Y_\rho$ and $Z_\rho$ are both valid accepting runs of $w$. If $X_\rho$ is not a valid accepting run of $w$, the first condition applies directly between $X$ and $Z$, yielding $X \leqRelation Z$. Otherwise, $\mathsf{isAcc}(X_w, X_\rho)$ holds, meaning $X \leqRelation Y$ and $Y \leqRelation Z$ were both derived from the second condition. Thus, we have $X_b \preceq_1 Y_b \preceq_1 Z_b$. Because $\preceq_1$ is transitive, $X_b \preceq_1 Z_b$, which satisfies the second condition for $X$ and $Z$, yielding $X \leqRelation Z$.

Now that \cref{prop:maximum-dpw-pspace} holds for preorders, let us modify the reduction of \cref{theorem:NEexistenceGeneral} presented in \cref{app:nash-existence-tools}. In \cref{fig:reduction-nash-existence-pspace}, we can assume that $\aut{A}$ is complete. Moreover, we add an extra state $q_s$ with a priority $0$ and transitions $(q_0',(a,a),q_s)$ for all $a \in \Sigma$, $(q_s,(a,a),q_s)$ for all $a \in \Sigma \cup \{v_0\}$, and $(q,(v_0,v_0),q_s)$ for all states $q$ of the original automaton $\aut{A}$, including $q_0$ and excluding $q'_0,q_s$.
One can easily verify that the new relation $\leqRelation'$ is reflexive. Let us prove that it is transitive. Let $x \leqRelation' y$ and $y \leqRelation' z$, and let us discuss the following three cases:
\begin{itemize}
    \item If $y$ starts with some symbol $a \in \Sigma$, then $x=y=z$ and $x \leqRelation' z$.
    \item If $y$ is of the form $v_0y'$ with $y' \in \Sigma^\omega$, then it is also the case for $x$ and $z$. As $\leqRelation$ is transitive, we get that $x \leqRelation' z$.
    \item If $y$ is of the form $v_0w_yv_0y'$ with $w_y \in \Sigma^*$, then $x = v_0w_xv_0x'$ and $z = v_0w_zv_0z'$ with $w_x,w_z \in \Sigma^*$ and $|w_x| = |w_y| = |w_z|$ and $x'=y'=z'$. Therefore, as $\aut{A}$ is complete, we get $x \leqRelation' z$.
\end{itemize}

\begin{figure}
    \centering
    \begin{tikzpicture}[x=0.75pt,y=0.75pt,yscale=-1]
        \draw (176.21,91.01) .. controls (176.21,86.38) and (179.76,82.63) .. (184.15,82.63) .. controls (188.54,82.63) and (192.09,86.38) .. (192.09,91.01) .. controls (192.09,95.63) and (188.54,99.38) .. (184.15,99.38) .. controls (179.76,99.38) and (176.21,95.63) .. (176.21,91.01) -- cycle ;
        \draw (158.06,91.38) -- (173.21,91.07) ;
        \draw [shift={(176.21,91.01)}, rotate = 178.81] [fill={rgb, 255:red, 0; green, 0; blue, 0},line width=0.08,draw opacity=0] (5.36,-2.57) -- (0,0) -- (5.36,2.57) -- (3.56,0) -- cycle;
        \draw  [fill={rgb, 255:red, 0; green, 0; blue, 0 }  ,fill opacity=0.03 ][dash pattern={on 4.5pt off 4.5pt}] (270.88,99.25) .. controls (270.88,94.38) and (274.83,90.43) .. (279.71,90.43) -- (364.84,90.43) .. controls (369.72,90.43) and (373.67,94.38) .. (373.67,99.25) -- (373.67,129.95) .. controls (373.67,134.83) and (369.72,138.78) .. (364.84,138.78) -- (279.71,138.78) .. controls (274.83,138.78) and (270.88,134.83) .. (270.88,129.95) -- cycle ;
        \draw (192.09,91.01) -- (229.58,91.4) ;
        \draw [shift={(232.58,91.43)}, rotate = 180.6] [fill={rgb, 255:red, 0; green, 0; blue, 0},line width=0.08,draw opacity=0] (5.36,-2.57) -- (0,0) -- (5.36,2.57) -- (3.56,0) -- cycle;
        \draw (281.49,69.38) .. controls (281.49,64.75) and (285.05,61) .. (289.44,61) .. controls (293.82,61) and (297.38,64.75) .. (297.38,69.38) .. controls (297.38,74) and (293.82,77.75) .. (289.44,77.75) .. controls (285.05,77.75) and (281.49,74) .. (281.49,69.38) -- cycle ;
        \draw (296.37,65.2) .. controls (316.99,62.34) and (319.02,75.34) .. (299.39,74.19) ;
        \draw [shift={(296.47,73.93)}, rotate = 6.7] [fill={rgb, 255:red, 0; green, 0; blue, 0},line width=0.08,draw opacity=0] (5.36,-2.57) -- (0,0) -- (5.36,2.57) -- (3.56,0) -- cycle;
        \draw (232.87,91.67) .. controls (232.87,87.05) and (236.43,83.3) .. (240.82,83.3) .. controls (245.2,83.3) and (248.76,87.05) .. (248.76,91.67) .. controls (248.76,96.3) and (245.2,100.05) .. (240.82,100.05) .. controls (236.43,100.05) and (232.87,96.3) .. (232.87,91.67) -- cycle ;
        \draw (248.76,94.67) -- (276.05,103.65) ;
        \draw [shift={(278.9,104.58)}, rotate = 198.2] [fill={rgb, 255:red, 0; green, 0; blue, 0},line width=0.08,draw opacity=0] (5.36,-2.57) -- (0,0) -- (5.36,2.57) -- (3.56,0) -- cycle;
        \draw (278.9,101.58) .. controls (278.9,99.87) and (280.29,98.47) .. (282.01,98.47) -- (344.17,98.47) .. controls (345.88,98.47) and (347.28,99.87) .. (347.28,101.58) -- (347.28,110.92) .. controls (347.28,112.64) and (345.88,114.03) .. (344.17,114.03) -- (282.01,114.03) .. controls (280.29,114.03) and (278.9,112.64) .. (278.9,110.92) -- cycle ;
        \draw (247.76,87.67) -- (279.13,74.73) ;
        \draw [shift={(281.9,73.58)}, rotate = 157.58] [fill={rgb, 255:red, 0; green, 0; blue, 0},line width=0.08,draw opacity=0] (5.36,-2.57) -- (0,0) -- (5.36,2.57) -- (3.56,0) -- cycle;
        \draw (184.15,99.38) -- (190.7,136.74) ;
        \draw [shift={(191.22,139.7)}, rotate = 260.06] [fill={rgb, 255:red, 0; green, 0; blue, 0},line width=0.08,draw opacity=0] (5.36,-2.57) -- (0,0) -- (5.36,2.57) -- (3.56,0) -- cycle;
        \draw (183.27,148.07) .. controls (183.27,143.45) and (186.83,139.7) .. (191.22,139.7) .. controls (195.6,139.7) and (199.16,143.45) .. (199.16,148.07) .. controls (199.16,152.7) and (195.6,156.45) .. (191.22,156.45) .. controls (186.83,156.45) and (183.27,152.7) .. (183.27,148.07) -- cycle ;
        \draw (237.35,98.58) -- (198.57,138.62) ;
        \draw [shift={(196.48,140.78)}, rotate = 314.08] [fill={rgb, 255:red, 0; green, 0; blue, 0},line width=0.08,draw opacity=0] (5.36,-2.57) -- (0,0) -- (5.36,2.57) -- (3.56,0) -- cycle;
        \draw (184.08,144.78) .. controls (165.33,140.76) and (164.88,157.37) .. (181.34,153.22) ;
        \draw [shift={(184.08,152.38)}, rotate = 160.5] [fill={rgb, 255:red, 0; green, 0; blue, 0},line width=0.08,draw opacity=0] (5.36,-2.57) -- (0,0) -- (5.36,2.57) -- (3.56,0) -- cycle;
        \draw  [fill={rgb, 255:red, 0; green, 0; blue, 0 }  ,fill opacity=0.03 ][dash pattern={on 4.5pt off 4.5pt}] (394.48,99.25) .. controls (394.48,94.38) and (398.43,90.43) .. (403.31,90.43) -- (485.24,90.43) .. controls (490.12,90.43) and (494.07,94.38) .. (494.07,99.25) -- (494.07,129.95) .. controls (494.07,134.83) and (490.12,138.78) .. (485.24,138.78) -- (403.31,138.78) .. controls (398.43,138.78) and (394.48,134.83) .. (394.48,129.95) -- cycle ;
        \draw (366.89,111.41) -- (404.38,111.8) ;
        \draw [shift={(407.38,111.83)}, rotate = 180.6] [fill={rgb, 255:red, 0; green, 0; blue, 0},line width=0.08,draw opacity=0] (5.36,-2.57) -- (0,0) -- (5.36,2.57) -- (3.56,0) -- cycle;
        \draw (248.76,91.67) .. controls (304.93,82.37) and (375.21,81.8) .. (409.24,102.49) ;
        \draw [shift={(411.28,103.78)}, rotate = 213.53] [fill={rgb, 255:red, 0; green, 0; blue, 0},line width=0.08,draw opacity=0] (5.36,-2.57) -- (0,0) -- (5.36,2.57) -- (3.56,0) -- cycle;
        \draw (284.07,132.78) .. controls (281.69,153.77) and (269.12,154.37) .. (201.42,151.26) ;
        \draw [shift={(199.36,151.17)}, rotate = 2.65] [fill={rgb, 255:red, 0; green, 0; blue, 0},line width=0.08,draw opacity=0] (5.36,-2.57) -- (0,0) -- (5.36,2.57) -- (3.56,0) -- cycle;

        \draw (195.6,76.52) node [anchor=north west,inner sep=0.75pt,font=\footnotesize,align=left] {$v'_{0} ,v'_{0}$};
        \draw (144.89,59.97) node [anchor=north west,inner sep=0.75pt,align=left] {$\mathcal{B}$};
        \draw (177,83.62) node [anchor=north west,inner sep=0.75pt,font=\small,align=left] {$q'_{0}$};
        \draw (281.87,65.56) node [anchor=north west,inner sep=0.75pt,font=\small,align=left] {$q_{+}$};
        \draw (314.93,61.72) node [anchor=north west,inner sep=0.75pt,font=\footnotesize,align=left] {$v'_{1} ,v$};
        \draw (244.43,72.14) node [anchor=north west,inner sep=0.75pt,font=\footnotesize,rotate=-338.03,align=left] {\begin{minipage}[lt]{20.14pt}\setlength\topsep{0pt}
        \begin{center}
        $v'_{1} ,v$
        \end{center}
        \end{minipage}};
        \draw (286.99,116.4) node [anchor=north west,inner sep=0.75pt,align=left] {$\mathcal{A} \times \left\{\text{equal}\right\}$};
        \draw (265.2,47.38) node [anchor=north west,inner sep=0.75pt,font=\scriptsize,align=left] {\begin{minipage}[lt]{32.37pt}\setlength\topsep{0pt}
        \begin{center}
        $(\forall v\in V)$
        \end{center}
        \end{minipage}};
        \draw (234.67,87.29) node [anchor=north west,inner sep=0.75pt,font=\small,align=left] {$q_{e}$};
        \draw (248.78,96.13) node [anchor=north west,inner sep=0.75pt,font=\footnotesize,rotate=-17.32,align=left] {$v,v$};
        \draw (280.27,99.5) node [anchor=north west,inner sep=0.75pt,font=\scriptsize,align=left] {$\delta _{A}( q_{0} ,( v,v))$};
        \draw (152.8,105.72) node [anchor=north west,inner sep=0.75pt,font=\footnotesize,align=left] {$v,v$\\$v\neq v'_{0}$};
        \draw (184.07,143.69) node [anchor=north west,inner sep=0.75pt,font=\small,align=left] {$q_{=}$};
        \draw (209.11,121.68) node [anchor=north west,inner sep=0.75pt,font=\footnotesize,align=left] {$v'_{0} ,v'_{0}$\\$v'_{1} ,v'_{1}$};
        \draw (104.4,138.38) node [anchor=north west,inner sep=0.75pt,font=\scriptsize,align=left] {\begin{minipage}[lt]{34.61pt}\setlength\topsep{0pt}
        \begin{center}
        $(\forall v\in V')$
        \end{center}
        \end{minipage}};
        \draw (143.2,150.92) node [anchor=north west,inner sep=0.75pt,font=\footnotesize,align=left] {$v,v$};
        \draw (401.39,114.4) node [anchor=north west,inner sep=0.75pt,align=left] {$\mathcal{A} \times \left\{\text{unequal}\right\}$};
        \draw (357.33,71.07) node [anchor=north west,inner sep=0.75pt,font=\footnotesize,rotate=-1.14,align=left] {$v,v'$};
    \end{tikzpicture}
    \caption{The modified automaton $\aut{B}$ for the hardness of \cref{theorem:constrainedNEexistenceGeneral} in the case of preorders. In the copies of $\aut{A}$, if all inputs are in the form $(v,v)$ for $v \in V$, stay in left copy with only accepting states of priority $0$. If there is some $(v,v')$ for $v,v' \in V$ with $v \neq v'$, go to the right copy.}
    \label{fig:reduction-pspace-hardness-nash-constrained-preorder}
\end{figure}

\subparagraph*{Constrained NE Existence Problem.} We use the same arena as in \cref{app:constrained-nash-existence}. Put $W=(V')^\omega$, $P=v'_0V^\omega$, and $d=v'_0(v'_1)^\omega$. For the original preorder $\leqRelation$, define
\[
    \mathord{\leqRelation'} = \{(v'_0x,v'_0y)\mid x\leqRelation y,~x,y\in V^\omega\}
    \ \cup\ (\{d\}\times P)\ \cup\ \{(x,x)\mid x\in W\}.
\]
This relation is reflexive. To see transitivity, ignore equality comparisons: all remaining comparisons either stay in $P$, where transitivity follows from $\leqRelation$, or start at $d$ and end in $P$. Thus $\leqRelation'$ is a preorder on the whole domain $W$. Its strict part is precisely the lifted old strict relation together with $\{d\}\times P$, so the constrained NE reduction is preserved.

The lifted relation is recognized by a copy of the old \DPW{} with a monitor checking the initial $(v'_0,v'_0)$ and subsequent letters in $V\times V$. The other two alternatives form a deterministic safety language, recognized with constant memory. Their union has a polynomial-size \DPW{} $\aut{B}$: run both monitors, assigning a fresh largest even priority while either safety alternative remains possible, and otherwise use the priorities of the lifted relation. This also accepts words containing new vertices at arbitrary positions (see \cref{fig:reduction-pspace-hardness-nash-constrained-preorder}).

\section{Proof of \texorpdfstring{\cref{prop:recognizable-finite-index}}{Proposition~\ref{prop:recognizable-finite-index}}}
\label{app:recognizable-finite-index}

\recognizablefiniteindex*

\begin{proof}
    We use~\cite[Lemma~$3$]{rational-relations-automatic-loding} to prove this proposition.
    Let $E_1 = \{(x,y) \mid \forall z,\, x \leqRelation z \text{ iff } y \leqRelation z\}$ and $E_2 = \{((x_1,x_2),(y_1,y_2)) \mid x_1 \leqRelation x_2 \mbox{ iff } y_1 \leqRelation y_2\}$ be two equivalence relations. From~\cite[Lemma~$3$]{rational-relations-automatic-loding}, we know that $\leqRelation$ is $\omega$-recognizable if and only if both $E_1$ and $E_2$ have finite index. Note that by definition, $E_2$ has at most two equivalence classes. Hence, to establish \cref{prop:recognizable-finite-index}, we will show that $E_1 = ~ \equivRelation$.
    Let $x,y \in \Sigma^\omega$. If $(x,y) \in E_1$, we can particularize $z$ to $x$ and then $y$, to get that $x \equivRelation y$ by reflexivity of $\leqRelation$.
    If $x \equivRelation y$, for all $z \in \Sigma^\omega$, we get by transitivity of $\leqRelation$ that $x \leqRelation z$ if and only if $y \leqRelation z$, thus $(x,y) \in E_1$.
\end{proof}

\section{Proof of the Existence of NE with Total \texorpdfstring{$\omega$}{ω}-Recognizable Preorders}
\label{app:NEtotal}

In this section, we assume that each game $\game$ has an $\omega$-recognizable preorder $\leqRelation[i]$ for each $i \in \Players$. We aim at proving \cref{theorem:existenceNE} when every relation is total.

By \cref{prop:recognizable-finite-index}, we denote by $\lattice[i]$ the finite partially ordered set formed by the equivalence classes of $\equivRelation[i]$. We also denote by $[x]_i$ the equivalence class of the word $x$ with respect to $\equivRelation[i]$. By abusive notation, we write $[x]_i \leqRelation[i] [y]_i$ whenever $x \leqRelation[i] y$.

\subparagraph*{Useful Lemmas}

In addition to \cref{prop:recognizable-finite-index}, we first mention some useful properties about preorders.

\begin{lemma} \label{lem:lasso}
Let $\leqRelation$ be an $\omega$-recognizable preorder over $\Sigma$ and $\equivRelation$ the related equivalence relation. Then, in each equivalence class of $\equivRelation$, there exists a lasso. Moreover, given a word $x \in \Sigma^\omega$, there exists a prefix $\mu\nu$ of $x$ such that $x \equivRelation \mu (\nu)^\omega$.
\end{lemma}

\begin{proof}
Let $x \in \Sigma^\omega$. As $\leqRelation$ is a preorder, we have $x \leqRelation x$. As $\leqRelation$ is $\omega$-recognizable, i.e., $\leqRelation$ is equal to $\cup_{i=1}^\ell X_i \times Y_i$ where $X_i, Y_i$ are $\omega$-regular languages, there exists $i$ such that $(x,x) \in X_i \times Y_i$. Therefore, $x$ belongs to the $\omega$-regular language $X_i \cap Y_i$. In an accepting run on $x$, we can repeat a nonempty accepting loop between two occurrences of a recurrent state, after all transient priorities have disappeared. This gives a lasso $\mu (\nu)^\omega$ in $X_i \cap Y_i$ such that $\mu\nu$ is a prefix of $x$. It follows that $(x,\mu(\nu)^\omega) \in X_i \times Y_i$ and $(\mu(\nu)^\omega,x) \in X_i \times Y_i$, thus $x \leqRelation \mu(\nu)^\omega$ and $\mu(\nu)^\omega \leqRelation x$, that is, $x \equivRelation \mu(\nu)^\omega$.
\end{proof}

When $x$ is a play from a vertex $v$, we can apply the same argument to $X_i \cap Y_i$ intersected with the $\omega$-regular language of plays from $v$. The resulting lasso is then also a play from $v$.

Note that \cref{lem:lasso} does not hold in the general case of $\omega$-automatic preorders. Let us take \leqRelation{} defined as $x \leqRelation y$ if and only if $x = y$. This is an $\omega$-automatic preorder where for each word $x$, its equivalence class is the singleton $\{x\}$. There is an infinite number of such classes, showing that $\leqRelation$ is not $\omega$-recognizable (by \cref{prop:recognizable-finite-index}), and \cref{lem:lasso} does not hold.

\subparagraph*{Total Preorders}

We now suppose that each preorder $\leqRelation[i]$ is total, that is, $\lattice[i]$ is a finite \emph{total} order. This means that the finite number of equivalence classes of $\equivRelation[i]$ are totally ordered from the lowest one to the highest one. Consequently, we have the next important property: for all $x,y \in \Sigma^\omega$,
\[
x \not\leqRelationStrict[i] y ~~\Leftrightarrow~~ y \leqRelation[i] x.
\]

Let us state \cref{theorem:existenceNE} in the particular setting of total preorders.

\begin{theorem}
\label{theorem:existenceNETotal}
    When the preference relations of a game are all $\omega$-recognizable and total preorders, then there always exists an NE.
\end{theorem}

To prove \cref{theorem:existenceNETotal}, we take inspiration from the work of~\cite{Gradel-Ummels-08} and~\cite{BrihayePS13}, where the existence of NEs is studied through the concept of value and optimal strategy. In the next definitions, for a fixed $i \in \Players$, the notation $\Sigma_{i}$ (resp.\ $\Sigma_{-i}$) is used for the set of all strategies of player~$i$ (resp.\ player~$-i$).

\begin{definition} \label{def:values}
    Let $\leqRelation[i]$ be the preference relation of player~$i$ and~$-i$ be the coalition of the other players. Let $v$ be a vertex. We define the following \emph{lower} and \emph{upper values}:

    \begin{itemize}
    \item $\underline{val}_i(v) = \max_{\sigma_i \in \Sigma_i} \min \{[\pi]_{i} \mid \pi$ consistent with $\sigma_i \text{ and } \first{\pi} = v\}$,
    \item $\overline{val}_i(v) = \min_{\sigma_{-i} \in \Sigma_{-i}} \max \{[\pi]_{i} \mid \pi$ consistent with $\sigma_{-i} \text{ and } \first{\pi} = v\}$.
    \end{itemize}
\end{definition}

All extrema in these definitions exist and are attained, since the sets of attainable classes are nonempty subsets of the finite total order $\lattice[i]$.

\begin{lemma}
    We have $\underline{val}_i(v) \leqRelation[i] \overline{val}_i(v)$ for each $v \in V$.
\end{lemma}

\begin{proof}
    Let $\sigma_i^v$ that realizes $\underline{val}_i(v)$ and $\sigma_{-i}^v$ that realizes $\overline{val}_i(v)$. Then, we have $\underline{val}_i(v) \leqRelation[i] [\outcomefrom{\sigma_i^v,\sigma_{-i}^v}{v}]_{i} \leqRelation[i] \overline{val}_i(v)$.
\end{proof}

\begin{definition}
    If $\underline{val}_i(v) = \overline{val}_i(v)$, then we speak about the \emph{value of $v$} denoted by $val_i(v)$ and two strategies $\sigma_i^v \in \Sigma_i$, $\sigma_{-i}^v \in \Sigma_{-i}$ that realize the value are called \emph{optimal}.
\end{definition}

Note that an optimal strategy $\sigma_i^v$ ensures consistent plays $\pi$ starting at $v$ such that $val_i(v) \leqRelation[i] [\pi]_{i}$. This leads us to introduce the concept of \emph{threshold game} $(\arena,\leqRelation[i],\Omega_i)$ defined from $\game$ and a lasso $\rho$, and with $\Omega_i = \{x \in V^\omega \mid \rho \leqRelation[i] x\}$. It is a zero-sum two-player game between the players~$i$ and~$-i$, whose objective of player~$i$ is the set $\Omega_i$ (while the objective of player~$-i$ is $V^\omega \ssetminus \Omega_i$). Hence, if $\rho$ is a lasso play in the equivalence class $val_i(v)$ such that $\first{\rho} = v$ (it exists by \cref{lem:lasso} and the observation following its proof), $\sigma_i^v$ is nothing more than a winning strategy in this threshold game. Similarly, an optimal strategy $\sigma_{-i}^v$ ensures consistent plays $\pi$ starting at $v$ such that $[\pi]_{i} \leqRelation[i] val_i(v)$. It is thus a winning strategy in the variant of threshold game $(\arena,\leqRelation[i],\Omega_{-i})$ where $\Omega_{-i} = \{x \in V^\omega \mid x \leqRelation[i] \rho\}$ is the objective of player~$-i$.

\begin{proposition}\label{theorem:zero-sum-automatic-relation-determinacy}
    Every threshold game (and its variant) is determined, with finite-memory winning strategies for both players.
\end{proposition}

\begin{proof}
    Let $(\arena,\leqRelation[i],\Omega_i)$ be a threshold game defined from $\game$ and a lasso $\rho$.
    From the \DPW{} $\aut{A}_i$ accepting $\leqRelation[i]$, we construct a \DPW{} $\aut{A}'_i$ accepting $\Omega_i$, which is the product between the lasso $\rho$ and $\aut{A}_i$. Then, we construct a zero-sum parity game $\mathcal{H}$ from the product between the arena $\arena$ of $\game$ and $\aut{A}'_i$, such that its parity condition encodes the objective $\Omega_i$.\footnote{Note that in the proof of the membership result of \cref{theorem:OutcomeCheck}, such a zero-sum generalized parity game was already constructed for the objective $\{x \in V^\omega \mid \rho \leqRelationStrict[i] x\}$.} Clearly, the existence of a winning strategy in the threshold game is equivalent to the existence of a winning strategy in $\mathcal{H}$. As parity games are determined and have memoryless winning strategies for both players~\cite{lncs2500}, this completes the proof. Note that the argument is similar for the variant with the objective $\Omega_{-i}$ for player~$-i$.
\end{proof}

\begin{proposition}
    There always exists a value $val_i(v)$ for each $v \in V$. Moreover, there exist finite-memory optimal strategies $\sigma_i^v$ and $\sigma_{-i}^v$ whose outcome $\outcomefrom{\sigma_i^v,\sigma_{-i}^v}{v}$ is a lasso in the equivalence class $val_i(v)$.
\end{proposition}

\begin{proof}
    Let us suppose for a contradiction that $\underline{val}_i(v) \leqRelationStrict[i] \overline{val}_i(v)$ for some $v$. It implies the existence of a play $\rho \in \overline{val}_i(v)$ with $\first{\rho}=v$ such that $\underline{val}_i(v) \leqRelationStrict[i] [\rho]_{i}$. By \cref{lem:lasso}, we can assume that $\rho$ is a lasso. In other words, we have
    \begin{align}
    \label{eq:value}
        \neg (\exists \sigma_i \in \Sigma_i, \forall \sigma_{-i} \in \Sigma_{-i}, ~ \rho \leqRelation[i] \outcomefrom{\sigma_i,\sigma_{-i}}{v}).
    \end{align}
    We consider the threshold game $(\arena,\leqRelation[i],\Omega_i)$ with the objective $\Omega_i = \{x \in V^\omega \mid \rho \leqRelation[i] x\}$
    for player~$i$.
    By determinacy, given by \cref{theorem:zero-sum-automatic-relation-determinacy}, and totality of $\leqRelation[i]$, \eqref{eq:value} is equivalent to
    \[
    \exists \sigma_{-i} \in \Sigma_{-i}, \forall \sigma_{i} \in \Sigma_{i}, ~ \outcomefrom{\sigma_i,\sigma_{-i}}{v} \leqRelationStrict[i] \rho.
    \]
    Since there are finitely many equivalence classes, the maximum class consistent with this strategy $\sigma_{-i}$ is strictly below $[\rho]_i = \overline{val}_i(v)$, contradicting the definition of $\overline{val}_i(v)$. Hence, for all $v$, we have $\underline{val}_i(v) = \overline{val}_i(v)$ and $val_i(v)$ exists.

    Let us now prove the existence of finite-memory optimal strategies. Let $\rho$ be a lasso in the equivalence class $val_i(v)$ such that $\first{\rho} = v$.
    By \cref{theorem:zero-sum-automatic-relation-determinacy} (where $\Omega_i$ is defined with this lasso $\rho$), there exists a (winning) optimal strategy $\sigma_i$ that is finite-memory. Similarly, there exists a finite-memory optimal strategy $\sigma_{-i}$ for player~$-i$ and his objective $\Omega_{-i} = \{x \in V^\omega \mid x \leqRelation[i] \rho\}$. Their joint outcome is a lasso, and the two guarantees put it in the class $val_i(v)$.
\end{proof}

Now that the values and optimal strategies are defined, we still need to introduce the concept of subgame before proceeding to the proof of \cref{theorem:existenceNETotal}. Given a game $\game = (\arena,(\leqRelation[i])_{i \in \Players})$ and a history $h \in V^*$, we denote by $\rest{\game}{h} = (\arena,(\leqRelation[i]^h)_{i \in \Players})$ the \emph{subgame from $h$}, where for each player~$i$, we define the relation $\leqRelation[i]^h$ by:
\[
\forall x,y \in V^\omega,~ x \leqRelation[i]^h y ~\Leftrightarrow~ hx \leqRelation[i] hy.
\]
Hence, this relation only refers to the preference relation $\leqRelation[i]$ restricted to words having $h$ as a prefix. We also define the relations $\leqRelationStrict[i]^h$ and $\equivRelation[i]^h$ as expected.
By definition, we have $x \equivRelation[i]^h y$ if and only if $hx \equivRelation[i] hy$. Hence, the equivalence class $[x]_i^h$ of $x$ for $\equivRelation[i]^h$ is mapped to the equivalence class $[hx]_{i}$ of $hx$ for $\equivRelation[i]$, called the \emph{projection} of $[x]_i^h$.

\begin{lemma}\label{lem:properties-subgame-relations}
    Given the $\omega$-recognizable total preorder $\leqRelation[i]$ and any history $h \in V^*$, the relation $\leqRelation[i]^{h}$ is also an $\omega$-recognizable total preorder.
\end{lemma}

\begin{proof}
    By definition, $\leqRelation[i]^{h}$ is clearly a total preorder. Furthermore, we can easily define a \DPW{} accepting $\leqRelation[i]^{h}$ from the \DPW{} $\aut{A}_i$ accepting $\leqRelation[i]$. Indeed, we simply take a copy of $\aut{A}_i$ and replace its initial state with the unique state reached by reading the pair $(h,h)$.
    Finally, to show that $\leqRelation[i]^h$ is $\omega$-recognizable, let us prove that $\equivRelation[i]^h$ has finite index (by \cref{prop:recognizable-finite-index}). The projection $[x]_i^h \mapsto [hx]_i$ is injective: $hx \equivRelation[i] hy$ implies $x \equivRelation[i]^h y$. Since $\equivRelation[i]$ has finite index, so does $\equivRelation[i]^h$.
\end{proof}

We are now ready to prove \cref{theorem:existenceNETotal}.

\begin{proof}[Proof of \cref{theorem:existenceNETotal}]
As done for $\game$, thanks to \cref{lem:properties-subgame-relations}, for each player~$i$ and each history $hv$, we can define for the subgame $\rest{\game}{h}$ the notion of value denoted $val_i^h(v)$ and of optimal strategies denoted $\sigma_i^{hv}$, $\sigma_{-i}^{hv}$ for player~$i$ and~$-i$ respectively. We also denote by $\nu_i^{hv}$ the equivalence class for $\equivRelation[i]$, equal to the projection of $val_i^h(v)$.
This means that $\sigma_i^{hv}$, (resp.\ $\sigma_{-i}^{hv}$) guarantees consistent plays $\pi$ starting at $v$ such that
\[
val_i^h(v) \leqRelation[i]^h [\pi]_i^h \quad(\text{resp.\ } [\pi]_i^h \leqRelation[i]^h val_i^h(v)),
\]
or equivalently
\begin{align}\label{eq:optimal}
\nu_i^{hv} \leqRelation[i] [h\pi]_i \quad(\text{resp.\ } [h\pi]_i \leqRelation[i] \nu_i^{hv}).
\end{align}
For each $i$, let us define the following strategy $\tau_i$\footnote{Note that the strategy profile composed of the $\tau_i$'s is not yet the required NE.} in $\game$: for all histories $hv$ with $v \in V_i$,
\[ \tau_i(hv) = \sigma^{h_1u}_i(h_2v) \]
such that
\begin{itemize}
    \item $h = h_1h_2$,
    \item $\nu_i^{h_1u} = \nu_i^{hv}$ with $u$ the first vertex of $h_2v$,
    \item $|h_1|$ is minimal with respect to the two previous items.
\end{itemize}
(We use the optimal strategy corresponding to the smallest prefix $h_1$ of $h$ such that the projections of both values $val_i^h(v)$ and $val_i^{h_1}(u)$ are the same.)

Consider the outcome $\rho = \outcomefrom{\tau}{v_0}$ of the strategy profile $\tau = (\tau_i)_{i \in \Players}$ from a given initial vertex $v_0$. For each $i$, we define $\nu^*_i = \max\{\nu_i^{hv} \mid hv \text{ prefix of } \rho\}$. Notice that $\nu_i^*$ is well defined as any $\nu_i^{hv}$ is an equivalence class of $\equivRelation[i]$ and $\lattice[i]$ is a finite total order. Let us prove that
\begin{align}\label{eq:LowerBound}
\nu_i^* \leqRelation[i] [\rho]_i \text{ for all } i \in \Players.
\end{align}
For this purpose, let us consider the smallest prefix $hv$ of $\rho$ such that $\nu_i^{hv} = \nu_i^*$, and let us prove that for any $hgu$ prefix of $\rho$,
\begin{itemize}
    \item $\tau_i(hgu) = \sigma_i^{hv}(gu)$ if $u \in V_i$,
    \item $\nu_i^{hgu} = \nu_i^*$.
\end{itemize}
(From $h$, $\tau_i$ plays as dictated by $\sigma_i^{hv}$ and the projection of the value $val_i^{hg}(u)$ remains equal to $\nu_i^*$.) The proof is by induction on $|g|$. By definition of $\tau_i$ and $hv$, this property trivially holds when $g = \varepsilon$ (in which case $u=v$). Let us assume that the property is true for $gu$ and let us prove that it remains true for $guu'$, with $u' \in V$ such that $hguu'$ is a prefix of $\rho$.
\begin{itemize}
    \item Let us first observe that $guu'$ is consistent with $\sigma_i^{hv}$. This is verified for $gu$ by induction hypothesis. Moreover, again by induction hypothesis, if $u \in V_i$, then $\tau_i(hgu) = \sigma_i^{hv}(gu) = u'$.
    \item Let us then prove that $\nu_i^{hguu'} = \nu_i^*$. By definition of $\nu_i^*$, we have $\nu_i^{hguu'} \leqRelation[i] \nu_i^*$. Assume by contradiction that $\nu_i^{hguu'} \leqRelationStrict[i] \nu_i^*$. By \eqref{eq:optimal}, it follows that the optimal strategy $\sigma_{-i}^{hguu'}$ can impose $[\pi]_i \leqRelation[i]\nu_i^{hguu'} \leqRelationStrict[i] \nu_i^*$ on plays $\pi$ having $hguu'$ as prefix. This is in contradiction with the optimal strategy $\sigma_i^{hv}$, the consistency of $guu'$ with this strategy, and thus by \eqref{eq:optimal} $\nu_i^* = \nu_i^{hv} \leqRelation[i] [\pi]_i$.
    \item Suppose that $u' \in V_i$, as $\nu_i^{hguu'} = \nu_i^*$, it follows by definition of $\tau_i$ that $\tau_i(hguu') = \sigma_i^{hv}(guu')$.
\end{itemize}
As a consequence, the property is proved and we get that $\rho = h\rho'$ with a play $\rho'$ consistent with $\sigma_i^{hv}$. By \eqref{eq:optimal} again, it follows that $\nu_i^* \leqRelation[i] [\rho]_i$, and \eqref{eq:LowerBound} is thus established.

We are now ready to define the required NE $\tau' = (\tau'_i)_{i \in \Players}$. It is partially defined in such a way as to generate $\rho = \outcomefrom{\tau'}{v_0}$. At the first divergent edge from a prefix $hv$ of $\rho$, its source $v \in V_i$ identifies the deviator~$i$. Thereafter, the other players use their components of $\sigma_{-i}^{hv}$. These prescriptions are compatible, since distinct first-divergence histories have disjoint continuation subtrees. In this way, by \eqref{eq:optimal}, the coalition imposes $[\rho']_i \leqRelation[i] \nu_i^{hv}$ on the deviating play $\rho'$. As $\nu_i^{hv} \leqRelation[i] \nu_i^*$ by definition of $\nu_i^*$ and by \eqref{eq:LowerBound}, we get $[\rho']_i \leqRelation[i] \nu_i^{hv} \leqRelation[i] \nu_i^* \leqRelation[i] [\rho]_i$. The deviating play is therefore not profitable for player~$i$, showing that $\tau'$ is an NE from $v_0$.
\end{proof}

\section{Proofs of \texorpdfstring{\cref{prop:partial-preorder-to-total-preorder,theorem:existenceNE}}{Proposition~\ref{prop:partial-preorder-to-total-preorder} and Theorem~\ref{theorem:existenceNE}}}
\label{app:generalNE}

\embeddingpartialtotal*

\begin{proof}
    The equivalence classes of \equivRelation{} form a finite partially ordered set. A topological sort gives an ordering $C_1,\dots,C_n$ such that $C_i \leqRelationStrict C_j$ implies $i<j$. We define
    \[
    \leqRelation' ~ = ~ \bigcup_{1\leq i \leq j \leq n} C_i \times C_j.
    \]
    Each class $C_i$ is $\omega$-regular: it is a section of $\equivRelation = \leqRelation \cap \geqRelation$, which is a finite union of regular rectangles. Hence, $\leqRelation'$ is $\omega$-recognizable. It is a total preorder whose equivalence classes are exactly $C_1,\dots,C_n$, and the induced total order extends the original partial order. Thus weak comparisons and equivalence are preserved, and the topological sort also preserves strict comparisons.
\end{proof}

\existenceNE*

For the proof, we suppose that the preference relations $\leqRelation[i]$, $i \in \Players$, are not necessarily total. This means that when $x \not\leqRelationStrict[i] y$, then either $x \not\leqRelation[i] y$ or $y \leqRelation[i] x$.

\begin{proof}
    Given a game $\game = (\arena, (\leqRelation[i])_{i \in \Players})$, we construct $\game' = (\arena, (\leqRelation[i]')_{i \in \Players})$ with the same arena $\arena$ and, for each $i \in \Players$, the $\omega$-recognizable total preference relations $\leqRelation[i]'$ as defined in \cref{prop:partial-preorder-to-total-preorder}.
    By \cref{theorem:existenceNETotal}, as $\leqRelation[i]'$ are total, there exists an NE $\sigma = (\sigma_i)_{i \in \Players}$ in $\game'$. This strategy profile $\sigma$ is also an NE in $\game$. Otherwise, we would get a profitable deviation $\outcomefrom{\sigma}{v_0} \leqRelationStrict[i] \pi$ for some player~$i$ in $\game$, thus the same profitable deviation $\outcomefrom{\sigma}{v_0} \leqRelationStrict[i]' \pi$ in $\game'$ by \cref{prop:partial-preorder-to-total-preorder}. This is a contradiction with $\sigma$ being an NE in $\game'$. As there exists an NE, there also exists an NE composed of finite-memory strategies by \cref{cor:existence-ne-finite-memory}.
\end{proof}

\section{Characterization of NE Outcomes and Prefix-Independency}
\label{app:carac}

We begin with a proof of \cref{theorem:characterization}, then give a characterization using subgame values without any prefix assumption.

\characterization*

\begin{proof}
Suppose first that $\rho$ is the outcome of an NE $(\tau_{i})_{i \in \Players}$. Assume by contradiction that for some vertex $v = \rho_n \in V_i$ belonging to $\rho$, we have $val_i(v) \not\leqRelation[i] [\rho]_{i}$, i.e., $[\rho]_{i} \leqRelationStrict[i] val_i(v)$ as $\leqRelation[i]$ is total. Let $\rho = h\rho_{\geq n}$, and let us consider an optimal strategy $\sigma_{i}^v$ for player~$i$. Player~$i$ can use this strategy to deviate at $v$ and thus produce a play $\pi$ starting at $v$ such that $val_i(v) \leqRelation[i] [\pi]_{i}$. As $\leqRelation[i]$ is prefix-independent, we have $\pi \equivRelation[i] h\pi$, and thus $[\rho]_i \leqRelationStrict[i] val_i(v) \leqRelation[i] [h\pi]_{i}$. Hence $h\pi$ is a profitable deviation for $i$, which is impossible since $\rho$ is an NE outcome.

Suppose now that $\rho$ is a play such that for all vertices $\rho_n$ of $\rho$, if $\rho_n \in V_i$, then $val_i(\rho_n) \leqRelation[i] [\rho]_{i}$. We define a strategy profile $\tau = (\tau_i)_{i \in \Players}$ such that it produces $\rho$, and at the first divergent edge from $\rho$, with source $v \in V_i$, the coalition~$-i$ switches to an optimal strategy $\sigma_{-i}^v$ from $v$. The coalition remembers this first deviator, so the prescriptions are compatible. Let us prove that $\tau$ is an NE. Suppose that player~$i$ deviates at $v = \rho_n$ from $\rho$ and consider the deviating play $\rho' = h\pi'$. By definition of $\tau$, we have $[\pi']_{i} \leqRelation[i] val_i(\rho_n)$. As $\leqRelation[i]$ is prefix-independent, $[\rho']_i = [h\pi']_i \leqRelation[i] val_i(\rho_n)$. By hypothesis, $val_i(\rho_n) \leqRelation[i] [\rho]_i$. Therefore, $[\rho']_i \leqRelation[i] [\rho]_i$, showing that $\rho'$ is not a profitable deviation. Hence, $\rho$ is an NE outcome.
\end{proof}

Without prefix-independency, we use the subgame values defined in \cref{app:NEtotal}.

\begin{theorem}
\label{theorem:characterization-residual}
    Let $\game$ be a game such that each preference relation $\leqRelation[i]$ is an $\omega$-recognizable total preorder. Let $\rho = \rho_0\rho_1\ldots$ be a play and $h_n = \rho_0\cdots\rho_{n-1}$. Then $\rho$ is an NE outcome if and only if for all $n$ and $i$ such that $\rho_n \in V_i$,
    \[
    val_i^{h_n}(\rho_n) \leqRelation[i]^{h_n} [\rho_{\geq n}]_i^{h_n}.
    \]
\end{theorem}

\begin{proof}
Suppose first that $\rho$ is an NE outcome and the condition fails at $v=\rho_n \in V_i$. By totality, $[\rho_{\geq n}]_i^{h_n} \leqRelationStrict[i]^{h_n} val_i^{h_n}(v)$. Player~$i$ can follow $\rho$ until $v$ and then use an optimal strategy in $\rest{\game}{h_n}$. Every resulting continuation $\pi$ satisfies $val_i^{h_n}(v) \leqRelation[i]^{h_n} [\pi]_i^{h_n}$, so $\rho \leqRelationStrict[i] h_n\pi$, a contradiction.

Conversely, follow $\rho$ until the first divergent edge, with source $v=\rho_n \in V_i$. Thereafter, the coalition~$-i$ uses its optimal strategy $\sigma_{-i}^{h_nv}$ in $\rest{\game}{h_n}$. Remembering the first deviator makes these prescriptions compatible. Any deviating play $h_n\pi$ satisfies
\[
[\pi]_i^{h_n} \leqRelation[i]^{h_n} val_i^{h_n}(v) \leqRelation[i]^{h_n} [\rho_{\geq n}]_i^{h_n},
\]
hence $h_n\pi \leqRelation[i] \rho$. Thus no deviation is profitable.
\end{proof}

For prefix-linear preorders, the conditions $val_i(\rho_n) \leqRelation[i] [\rho_{\geq n}]_i$ for all $\rho_n \in V_i$ remain sufficient. Indeed, the same first-deviator construction using $\sigma_{-i}^{\rho_n}$ gives $\pi \leqRelation[i] \rho_{\geq n}$, and prefix-linearity yields $h_n\pi \leqRelation[i] \rho$. These conditions need not be necessary, since adding $h_n$ may erase a strict suffix improvement.

We might also want to check whether a given $\omega$-automatic relation satisfies prefix-independency or prefix-linearity, as for the classical properties considered in \cref{prop:PropertiesRelations}. We get the same complexity class:
\begin{proposition}\label{prop:Properties-prefixindep-prefixlinear}
    The problem of deciding whether an $\omega$-automatic relation $R$ is prefix-independent (resp.\ prefix-linear) is \nlComplete{}.
\end{proposition}

\begin{figure}[t]
    \centering
    \begin{tikzpicture}[x=0.75pt,y=0.75pt,yscale=-1]
        \draw (220.07,134.29) .. controls (220.07,128.76) and (224.55,124.29) .. (230.07,124.29) .. controls (235.59,124.29) and (240.07,128.76) .. (240.07,134.29) .. controls (240.07,139.81) and (235.59,144.29) .. (230.07,144.29) .. controls (224.55,144.29) and (220.07,139.81) .. (220.07,134.29) -- cycle;
        \draw (264.27,114.09) -- (298.67,114.09);
        \draw [shift={(301.67,114.09)},rotate=180,fill={rgb,255:red,0;green,0;blue,0},line width=0.08, draw opacity=0] (5.36,-2.57) -- (0,0) -- (5.36,2.57) -- (3.56,0) -- cycle;
        \draw (204.77,134.59) -- (217.07,134.35);
        \draw [shift={(220.07,134.29)},rotate=178.87,fill={rgb,255:red,0;green,0;blue,0},line width=0.08, draw opacity=0] (5.36,-2.57) -- (0,0) -- (5.36,2.57) -- (3.56,0) -- cycle;
        \draw [fill={rgb,255:red,0;green,0;blue,0},fill opacity=0.03,dash pattern={on 4.5pt off 4.5pt}] (212.99,107.58) .. controls (212.99,99.09) and (219.88,92.21) .. (228.37,92.21) -- (316.78,92.21) .. controls (325.28,92.21) and (332.16,99.09) .. (332.16,107.58) -- (332.16,161.08) .. controls (332.16,169.57) and (325.28,176.45) .. (316.78,176.45) -- (228.37,176.45) .. controls (219.88,176.45) and (212.99,169.57) .. (212.99,161.08) -- cycle;
        \draw (244.27,114.09) .. controls (244.27,108.56) and (248.75,104.09) .. (254.27,104.09) .. controls (259.79,104.09) and (264.27,108.56) .. (264.27,114.09) .. controls (264.27,119.61) and (259.79,124.09) .. (254.27,124.09) .. controls (248.75,124.09) and (244.27,119.61) .. (244.27,114.09) -- cycle;
        \draw (301.67,114.09) .. controls (301.67,108.56) and (306.15,104.09) .. (311.67,104.09) .. controls (317.19,104.09) and (321.67,108.56) .. (321.67,114.09) .. controls (321.67,119.61) and (317.19,124.09) .. (311.67,124.09) .. controls (306.15,124.09) and (301.67,119.61) .. (301.67,114.09) -- cycle;
        \draw (306.74,160.26) .. controls (326.14,180.63) and (373.14,175.75) .. (387.62,149.52);
        \draw [shift={(388.87,147.03)},rotate=114.26,fill={rgb,255:red,0;green,0;blue,0},line width=0.08, draw opacity=0] (5.36,-2.57) -- (0,0) -- (5.36,2.57) -- (3.56,0) -- cycle;
        \draw (402.2,130.71) .. controls (421.05,123.86) and (423.12,139.97) .. (404.91,136.79);
        \draw [shift={(402.2,136.21)},rotate=14.2,fill={rgb,255:red,0;green,0;blue,0},line width=0.08, draw opacity=0] (5.36,-2.57) -- (0,0) -- (5.36,2.57) -- (3.56,0) -- cycle;
        \draw (379.47,133.49) .. controls (379.47,127.96) and (383.95,123.49) .. (389.47,123.49) .. controls (394.99,123.49) and (399.47,127.96) .. (399.47,133.49) .. controls (399.47,139.01) and (394.99,143.49) .. (389.47,143.49) .. controls (383.95,143.49) and (379.47,139.01) .. (379.47,133.49) -- cycle;
        \draw (376.73,133.49) .. controls (376.73,126.45) and (382.43,120.75) .. (389.47,120.75) .. controls (396.51,120.75) and (402.21,126.45) .. (402.21,133.49) .. controls (402.21,140.52) and (396.51,146.23) .. (389.47,146.23) .. controls (382.43,146.23) and (376.73,140.52) .. (376.73,133.49) -- cycle;
        \draw (270.64,98) node [anchor=north west,inner sep=0.75pt, font=\footnotesize,align=left] {$a,\#$};
        \draw (223.07,129.9) node [anchor=north west,inner sep=0.75pt,align=left] {$q_{0}$};
        \draw (248.67,110.5) node [anchor=north west,inner sep=0.75pt,align=left] {$q$};
        \draw (306.47,105.3) node [anchor=north west,inner sep=0.75pt,align=left] {$q'$};
        \draw (218.1,156.1) node [anchor=north west,inner sep=0.75pt, font=\small,align=left] {$\aut{A}$};
        \draw (341.74,138.4) node [anchor=north west,inner sep=0.75pt, font=\footnotesize,align=left] {$\#,*$\\$*,b$};
        \draw (419.94,128.6) node [anchor=north west,inner sep=0.75pt, font=\footnotesize,align=left] {$*,*$};
        \draw (381.67,129.1) node [anchor=north west,inner sep=0.75pt,align=left] {$q_{s}$};
    \end{tikzpicture}
    \caption{The \DBW{} $\aut{A}'$ used in the reduction for the prefix-independency property.}
    \label{fig:nl-hardness-prefix-independent}
\end{figure}

\begin{proof}[Proof of \cref{prop:Properties-prefixindep-prefixlinear} - prefix-independency]
    We begin with the \nl{}-membership. Let $\aut{A}$ be a complete \DPW{} accepting $R$. By iterating one-letter prefix additions and deletions, $R$ is prefix-independent if and only if, for all $x,y \in \Sigma^\omega$ and $a \in \Sigma$,
    \[
    (x,y) \in R \Leftrightarrow (ax,y) \in R,
    \qquad
    (x,y) \in R \Leftrightarrow (x,ay) \in R.
    \]
    A violation of either equivalence is detected by two copies of $\aut{A}$, one complemented, together with a one-letter buffer for the shifted coordinate. For instance, to compare $(x,y)$ with $(ax,y)$, the second copy reads $(a,y_0)$ first and then $(x_{n-1},y_n)$ at step $n\geq 1$. Guessing $a$, the coordinate, and the direction of the violation gives a polynomial-size generalized \NPW{} with a conjunction of two parity conditions. Its nonemptiness is in \nl{} by \cref{prop:emptiness-universality-dpw-generalized-nl-complete}, and the product is constructed on the fly. Closure of \nl{} under complementation proves membership for prefix-independency.
    \medskip

    Let us now show the \nl{}-hardness, by reduction from the universality problem of \DBWs{}~\cite{handbook-of-model-checking-orna-kupferman}. Let $\aut{A}$ be a \DBW{} over the alphabet $\Sigma$. Let $\Sigma' = \Sigma \cup \{\#\}$ with a new symbol $\#$. From $R' = (\Sigma^\omega \ssetminus \lang{\aut{A}}) \times \{\#^\omega\}$, we define the complementary relation
    \[
    R = ((\Sigma')^\omega \times (\Sigma')^\omega) \setminus R'.
    \]
    Let us show that the reduction is correct. If $\lang{\aut{A}} = \Sigma^\omega$, then $R = (\Sigma')^\omega \times (\Sigma')^\omega$, which is prefix-independent. Conversely, if there exists $x \in \Sigma^\omega \ssetminus \lang{\aut{A}}$, then $(x,\#^\omega) \not\in R$ but $(x,a\#^\omega) \in R$ with $a \in \Sigma$, so $R$ is not prefix-independent.
    Let us now describe a \DBW{} $\aut{A}'$ accepting $R$, as depicted in \cref{fig:nl-hardness-prefix-independent}. W.l.o.g., we suppose that $\aut{A}$ is complete. The automaton $\aut{A'}$ is composed of a copy of $\aut{A}$ extended with a second component $\#$ in a way to accept $\lang{\aut{A}} \times \{\#^\omega\}$. As soon as a pair $(\#,*)$ or $(*,b)$ (with $b \in \Sigma$) is read from this copy, we go in an accepting sink state $q_{s}$. Therefore, the only rejected words are those in $R'$. Note that this is a logspace reduction as we need a logarithmic space on the work tape to construct $\aut{A}'$.
\end{proof}

\begin{proof}[Proof of \cref{prop:Properties-prefixindep-prefixlinear} - prefix-linearity]
    Let us begin with the \nl{}-membership. Let $R$ be a preference relation and $\aut{A}$ be a \DPW{} accepting it. Let us show that $R$ is not prefix-linear, i.e., there exist $x,y \in \Sigma^\omega$ and $u \in \Sigma^*$ such that $(x,y) \in R$ and $(ux,uy) \not\in R$. We first compute the set $I = \{q \mid q$ is reachable by a run on $(u,u)$, for $u \in \Sigma^*\}$ by deleting all edges labeled $(a,b)$ for $a \neq b$. Then, we construct a nondeterministic parity automaton $\aut{B}$ by taking the complement of $\aut{A}$ that accepts $\neg R$ and saying that every state $q \in I$ is initial. Therefore, $(x,y) \in \lang{\aut{B}}$ if and only if there exists $u \in \Sigma^*$ such that $(ux,uy) \not\in R$. Thus, we have to check whether $\lang{\aut{A} \cap \aut{B}} \neq \varnothing$. This is the nonemptiness problem of a generalized \NPW{} with a conjunction of two parity objectives, an \nlComplete{} problem by \cref{prop:emptiness-universality-dpw-generalized-nl-complete}. The product and the reachability tests defining $I$ are handled on the fly, and closure of \nl{} under complementation gives membership for prefix-linearity.

    \medskip

    Let us now prove the \nl{}-hardness with a reduction from the emptiness of a \DBW{} $\aut{A}$, an \nlComplete{} problem~\cite{handbook-of-model-checking-orna-kupferman}. We construct $R = \lang{\aut{A}} \times \{\#^\omega\}$. Clearly, $R$ is accepted by a \DBW{} (thus a \DPW{}) consisting in a copy of $\aut{A}$ where any label $a$ is replaced by $(a,\#)$ for all $a \in \Sigma$. Moreover, if there exists a word $w \in \lang{\aut{A}}$, then $(w,\#^\omega) \in R$ but adding the common prefix $\#$ gives $(\#w,\#^\omega) \not\in R$, so $R$ is not prefix-linear. Conversely, if $\lang{\aut{A}} = \varnothing$, then $R = \varnothing$ and thus it is prefix-linear.
\end{proof}

\end{document}